\documentclass{LMCS}

\usepackage{eucal}
\usepackage{url}
\usepackage{amsmath}
\usepackage{amssymb}
\usepackage{stmaryrd}
\usepackage{mathptm}
\usepackage{array}
\usepackage{verbatim}
\usepackage{pslatex}
\usepackage{mathrsfs}
\usepackage{psboxit}
\usepackage{booktabs}
\usepackage{url}
\usepackage[pdftex]{hyperref}
\usepackage[numbers,sort&compress]{natbib}
\usepackage{capt-of,enumerate}
\usepackage{hypernat}
\usepackage{graphicx}

\theoremstyle{plain}

\newtheorem{corollary}[thm]{Corollary}
\newtheorem{lemma}[thm]{Lemma}
\newtheorem{proposition}[thm]{Proposition}
\newtheorem{theorem}[thm]{Theorem}

\theoremstyle{definition}
\newtheorem{definition}[thm]{Definition}

\newtheorem{remarks}[thm]{Remark}
\newtheorem{notation}[thm]{Notation}

 \newif\ifcolour
\colourfalse

\newcommand{\mfb}[1]{\ifcolour
            {\color{yellow}{#1}} \else {#1} \fi}

 \newcommand{\martinb}[1]{\ifcolour
                          {\color{red}{#1}} \else {#1} \fi}

 \newcommand{\someCon}[2]{\ENCan{!#1}{#2}}
 \newcommand{\allCon}[2]{[!#1]{#2}}

 \newcommand{\ARROWb}[1]{\FS}

 \newcommand{\ARROW}[1]{\Rightarrow}
 \newcommand{\ENCan}[1]{\langle #1 \rangle}
 \newcommand{\NI}{\noindent}
 \newcommand{\TVX}{{\TYPEVAR{X}}}
 \newcommand{\TVXscript}{{\TYPEVARsc{X}}}
 \newcommand{\TVY}{{\TYPEVAR{Y}}}
 
 \newcommand{\TYPEVAR}[1]{{\text{\footnotesize\sc #1}}}
 \newcommand{\TYPEVARsc}[1]{{\text{\scriptsize\sc #1}}}
 \newcommand{\ASET}[1]{\{ #1 \}}
 \newcommand{\VEC}[1]{{\tilde{#1}}}                    
 \newcommand{\inj}[2]{\mathtt{in}_{#1}(#2)}
 \newcommand{\LOGICINJ}[3]{\mathsf{inj}_{#2}^{#1}(#3)}
 \newcommand{\CASE}[3]{\mathtt{case}\ #1\ \mathtt{of}\ \{\inj{i}{#2}.#3\}_{i \in \{1, 2\}}}

 \newcommand{\LET}[3]{\mathtt{let}\ #1=\, #2\ \mathtt{in}\ #3}

 \newcommand{\UNIT}{\mathsf{Unit}} 
 
 \newcommand{\FS}{\!\Rightarrow\!}
 
 \newcommand{\SKIP}{\mathtt{skip}}

 \newcommand{\FV}[1]{{\sf fv}(#1)}
 \newcommand{\FL}[1]{{\sf fl}(#1)}

 \newcommand{\fv}[1]{\mbox{\sf fv}(#1)}

 \newcommand{\BOOL}{\mathsf{Bool}}
 \newcommand{\NAT}{\mathsf{Nat}}

 \newcommand{\proves}{\vdash} 
 \newcommand{\WB}{\approx}                
 \newcommand{\MSUBS}[2]{[#1/#2]}
 \newcommand{\MMSUBS}[2]{[#1/#2;#2/#1]}
 \newcommand{\RED}{\longrightarrow}  
 \newcommand{\DEFEQ}{\stackrel{\mbox{\scriptsize def}}{=}}  
 
 \newcommand{\NUL}{\varepsilon}

 \newcommand{\IFTHENELSE}[3]{\mathtt{if}\ #1\ \mathtt{then}\ #2\ \mathtt{else}\ #3}

 \newenvironment{myfigure}{\begin{figure}\rule{\linewidth}{.5pt}\vspace{3.2mm}}{\rule{\linewidth}{.5pt}\end{figure}}
 \newcommand{\infer}[2]{\frac{\displaystyle{ #1 }}{\displaystyle{ #2 }}}
 \newcommand{\inferBase}[2]{#2}
 \newcommand{\ATb}[2]{#1\!:\!#2}
 
 \newcommand{\AT}[2]{#1\!:\!#2}
 
 \newcommand{\truth}{\mathsf{T}}
 \newcommand{\falsity}{\mathsf{F}}
 \newcommand{\minus}{{\mbox{\bf\small -}}}
 \newcommand{\CONVERGES}{\Downarrow}
 \newcommand{\ASSERT}[4]{\{#1\}\; #2 :_{#3} \{#4\}}

 \newcommand{\SHORTASSERT}[3]{\{#1\}\; #2  \{#3\}}
 \newcommand{\EVALASSERT}[5]{\{#1\}\; #2 \bullet #3 \comesbefore #4\; \{#5\}}

\newcommand{\CAL}[1]{\mathcal{#1}}

 \newcommand{\semb}[1]{\lbrack\!\lbrack #1 \rbrack\!\rbrack}
 \newcommand{\true}{\mathsf{t}}
 \newcommand{\false}{\mathsf{f}}
 \newcommand{\TYPES}[3]{#1 \vdash #2 : #3}
 \newcommand{\LITEQ}{\equiv}              
 \newcommand{\RULENAME}[1]{\mbox{\emph{#1}}}

  \newcommand{\PAIR}[2]{\langle #1, #2 \rangle}

 \renewcommand{\boldsymbol}[1]{\mathbf{#1}}

 \newcommand{\bmq}{\begin{quote}$}
 \newcommand{\emq}{$\end{quote}}
 \newcommand{\bitem}{\begin{list}{$\bullet$}{%
 \setlength{\parsep}{0mm}%
 \setlength{\rightmargin}{\leftmargin}%
 \setlength{\topsep}{0mm}%
 \setlength{\parskip}{0mm}}}
 \newcommand{\eitem}{\end{list}}

 \newcommand{\OOR}{\bigvee}

 \newcommand{\THENS}{{  \Rightarrow }}
 \newcommand{\THENs}{{\  \Rightarrow\ }}
 
 \newcommand{\THEN}{{\,\,\, \; \Rightarrow \; \,\,\,}}
 \newcommand{\IFFs}{\Leftrightarrow}

 \newcommand{\IFFDEF}{\stackrel{\mbox{\scriptsize def}}{\ \LITEQ \ }}

 \newcommand{\SET}[1]{{\bf #1}}                     

 \newcommand{\ENCda}[1]{\langle\!\langle #1 \rangle\!\rangle}

 \newcommand{\ENCdb}[1]{\lbrack\!\lbrack #1 \rbrack\!\rbrack}

 \newcommand{\OL}{\overline}  

 \newcounter{anumber}

 {\rm%
 \begin{list}%
 {(\roman{anumber})}%
  {\usecounter{anumber}%
   \addtolength{\labelwidth}{15mm}%
   \addtolength{\leftmargin}{5mm}%
   \setlength{\rightmargin}{0pt}%
   \setlength{\itemsep}{1mm}%
   \setlength{\parsep}{0pt}}}%
 {\end{list}}
 {\rm%
 \begin{list}%
 {(\roman{anumber})}%
  {\usecounter{anumber}%
   \addtolength{\labelwidth}{15mm}%
   \addtolength{\leftmargin}{0mm}%
   \setlength{\rightmargin}{0mm}%
   \setlength{\itemsep}{1mm}%
   \setlength{\parsep}{0pt}}}%
 {\end{list}}
 \newcounter{analphabet}

 {\rm%
  \begin{list}%
  {(\arabic{analphabet})}%
  {\usecounter{analphabet}%
   \addtolength{\labelwidth}{15mm}%
   \addtolength{\leftmargin}{5mm}%
   \setlength{\rightmargin}{0pt}%
   \setlength{\itemsep}{1mm}%
   \setlength{\parsep}{0pt}}}%
  {\end{list}}

 {\rm%
  \begin{list}%
  {(\arabic{analphabet})}%
  {\usecounter{analphabet}%
   \addtolength{\labelwidth}{15mm}%
   \addtolength{\leftmargin}{0mm}%
   \setlength{\rightmargin}{0pt}%
   \setlength{\itemsep}{1mm}%
   \setlength{\parsep}{0pt}}}%
  {\end{list}}

 {\rm%
  \begin{list}%
  {\arabic{analphabet}.}%
  {\usecounter{analphabet}%
   \addtolength{\labelwidth}{0mm}%
   \addtolength{\leftmargin}{-6mm}%
   \setlength{\rightmargin}{0pt}%
   \setlength{\itemsep}{1mm}%
   \setlength{\parsep}{0pt}}}%
  {\end{list}}

  {\begin{eqnarray*}}%
  {\end{eqnarray*}}

  {\[%
   \left\{%
   \begin{array}{cccc}}%
  {\end{array}%
   \right.%
   \]}%

\newcommand{\TPL}[2]{\ENCan{#1,\; #2}}

 \newcommand{\CD}{\cdot}
 \newcommand{\CDs}{\!\cdot\!}

 \newcommand{\MRED}{\rightarrow\!\!\!\!\rightarrow}

 {\end{array} \end{displaymath}}

 {\end{array} \end{displaymath}}

 \newenvironment{reasoning}
 {\begin{displaymath} \begin{array}{rclll}}
 {\end{array} \end{displaymath}}

 {\rm%
  \begin{list}%
  {(\arabic{analphabet})}%
  {\usecounter{analphabet}%
   \addtolength{\labelwidth}{15mm}%
   \addtolength{\leftmargin}{0pt}%
   \setlength{\rightmargin}{0pt}%
   \setlength{\itemsep}{1mm}%
   \setlength{\parsep}{0pt}}}%
  {\end{list}}

 \newcommand{\dom}[1]{{\sf dom}(#1)}      

 \newcommand{\ID}{\mathsf{id}}      

 \newcommand{\IGNORE}[1]{}

 \newcommand{\FPRM}[2]{\mbox{${{#1}\choose{#2}}$}}

 \newcommand{\OMEGAalg}[1]{{\bf $\mbox{$\Omega$-Alg}^{\SET{C}}$}}

 {\begin{enumerate}[(ii)]%
 }
 {\end{enumerate}}

 \newcommand{\new}{\nu}

 {\end{array}%
  \right.}

 \newcommand{\rulenameB}[1]{[\text{\em #1}]}

 \newcommand{\comesbefore}{=}

 \newcommand{\ACTIVE}[1]{{\ENCan{\VEC{y}_i\;\mbox{\footnotesize
 active}}}}

 \def\fps@figure{tp}      
 \def\fps@table{tp}
 \newcommand{\DEREF}[1]{{!{#1}}}

 \newenvironment{BITEM}%
 {\rm%
  \begin{list}%
  {$\bullet$}%
  {\addtolength{\labelwidth}{10mm}%
   \addtolength{\leftmargin}{0mm}%
   \setlength{\rightmargin}{0pt}%
   \setlength{\itemsep}{1mm}%
   \setlength{\parsep}{0pt}}}%
  {\end{list}}

 \newcommand{\SHORTVERSIONONLY}[1]{{}}
 
 \newenvironment{myfiguretwo}{\begin{figure*}\rule{\linewidth}{.5pt}}
 {\rule{\linewidth}{.5pt}\end{figure*}}

 \newcommand{\PCFv}{PCFv}

 \newcommand{\LETNEW}[3]{\mathtt{new}\ #1:=#2\ \mathtt{in}\ #3}

 \newcommand{\factorial}[1]{#1\boldsymbol{!}}

 \newcommand{\New}[2]{{(\nu\: #1)#2}}

 \newcommand{\MAP}[1]{\ENCdb{#1}}

 \newcommand{\MMMM}{\CAL{M}}

 {\end{array} \end{displaymath}}
 \renewcommand{\ARROW}[1]{\stackrel{\!\!\!#1}{\Rightarrow}}

 \newcommand{\linebreakA}{\\[1.2mm]}

 \newcommand{\commentOut}[1]{}

 \newcommand{\commentout}[1]{}

 \newcommand{\RC}[1]{\CAL{E}\![#1]}
 \newcommand{\RCCD}{\RC{\,\CD\,}}

 \newcommand{\thankstoBig}[1]{\hfill\pcom{(#1)}\nextLineBig}
 \newcommand{\nothanks}{\nextLineBig}

 \newenvironment{inference}
 {\begin{displaymath} \begin{array}{lrl}}%
 {\end{array} \end{displaymath}}

 \newcommand{\LSUBS}[2]{\llbracket #1 / #2 \rrbracket}

 \renewcommand{\LSUBS}[2]{\{\!\!|#1/#2|\!\!\}}

 \newcommand{\nextLine}{\\[1mm] \hline \\[-2.5mm]}
 \newcommand{\nextLineBig}{\\[1mm] \hline\\[-2.5mm]}
 \newcommand{\LINESKELETON}[3]{\!\!\!#1#2\;\hfill\quad\text{{#3}}\!\!\!}
 
 \newcommand{\LASTLINEs}[3]{#1#2}
 \newcommand{\LASTLINE}[3]{\LINESKELETON{#1}{#2}{#3}}
 \newcommand{\LINE}[3]{\LINESKELETON{#1}{#2}{#3}\nextLine}
 \newcommand{\LINEs}[2]{#1#2\nextLine}

 \newenvironment{DERIVATION}{\[\begin{array}{l}}{\end{array}\]}
 \newcommand{\NOANCHORASSERT}[3]{\{#1\}\ #2\ \{ #3\}}

 \newcommand{\ONEPREMISERULE}[2]{\infer{#1}{#2}}
 \newcommand{\TWOPREMISERULE}[3]{\infer{#1 \quad #2}{#3}}

\renewcommand{\RULENAME}[1]{\text{\emph{#1}}}

\newcommand{\LEFTONEPREMISERULENAMED}[3]{[\RULENAME{#1}]\,\ONEPREMISERULE{#2}{#3}}
\newcommand{\LEFTTWOPREMISERULENAMED}[4]{[\RULENAME{#1}]\,\TWOPREMISERULE{#2}{#3}{#4}}

\newcommand{\ASSERTSURFACE}[5]{\ASSERT{#1}{#2}{#3}{#4} @ #5}

\newcommand{\NOANCHORASSERTSURFACE}[4]{\NOANCHORASSERT{#1}{#2}{#3} @ #4}

\newcommand{\MODELTYPES}[2]{#1 \vdash #2}

\newcommand{\EXPRESSIONTYPES}[3]{#1 \vdash #2 : #3}
\newcommand{\LOGIC}[1]{\mathsf{#1}}
\newcommand{\INJ}[2]{\mathtt{inj}_{#1}(#2)}

\newcommand{\INJtp}[3]{\mathtt{inj}^{#3}_{#1}(#2)}

\newcommand{\SUBST}[2]{[#1/#2]}

\newcommand{\congLogic}{\cong_{\CAL{L}}}

\newcommand{\converges}{\Downarrow}

\newcommand{\ncl}[2]{\mathsf{lc}(#1, #2)}

\newcommand{\LETREF}[3]{\LET{#1}{\refPrg{#2}}{#3}}

\newcommand{\EVOLVES}{\!\leadsto\!}
\newcommand{\EVOLVESinc}[1]{\!\stackrel{#1}{\leadsto}\!}

\newcommand{\LocalInv}[3]{\mathsf{Inv}(#1, #2, #3)}

\newcommand{\injection}{\mathsf{inj}}

\newcommand{\fAct}[1]{#1\pmb{!}}

\newcommand{\reftype}[1]{\mathsf{Ref}(#1)}

\newcommand{\AIHax}{{\sf (AIH)}}

\newcommand{\DIVPROG}{\Omega}

\newcommand{\IncPrg}{\mathtt{Inc}}
\newcommand{\IncPrgTwo}{\mathtt{Inc2}}

\newcommand{\IncShared}{\mathtt{incShared}}
\newcommand{\IncUnshared}{\mathtt{incUnShared}}

\newcommand{\REF}[1]{\TYPE{Ref}(#1)}
\newcommand{\refType}[1]{\mathsf{Ref}(#1)}
\newcommand{\refPrg}[1]{\mathtt{ref}(#1)}

\newcommand{\noreach}{\,\#\,}

\newcommand{\eval}[5]{\ASET{#1}\APP{#2}{#3} =  #4\ASET{#5}}

\newcommand{\notreach}{\noreach}
\newcommand{\reach}{\hookrightarrow}
\newcommand{\reachable}{\hookrightarrow}

\newcommand{\reachB}[2]{#1 \hookrightarrow^\circ  #2}

\newcommand{\congmodel}{\WB}

\newcommand{\EVAL}[5]{\eval{#1}{#2}{#3}{#4}{#5}}

\newcommand{\expands}[3]{#1[#2\!:\!#3]}
\newcommand{\updates}[3]{#1[#2\mapsto #3]}

\newcommand{\NEW}[3]{\PROGRAM{new}\ #1\ := #2\ \PROGRAM{in}\ #3}
\newcommand{\PROGRAM}[1]{\mathtt{#1}}
\newcommand{\TYPE}[1]{\mathsf{#1}}

\newcommand{\NOTREACH}[2]{#1 \notreach  #2}

\newcommand{\REACH}[2]{#1 \hookrightarrow  #2}
\newcommand{\REFPROG}[1]{\PROGRAM{ref}(#1)}
\newcommand{\REFPRG}[1]{\PROGRAM{ref}(#1)}

\newcommand{\EVALEFFECT}[6]{\EVAL{#1}{#2}{#3}{#4}{#5} @ #6}

\newcommand{\TRUTH}{\LOGIC{T}}

\newcommand{\PFN}[1]{\mathsf{fpn}(#1)}
\newcommand{\INCSPEC}[1]{\mathsf{inc}(#1)}

\def\squareforqed{\hbox{\rlap{$\sqcap$}$\sqcup$}}
\def\qed{\ifmmode\squareforqed\else{\unskip\nobreak\hfil
\penalty50\hskip1em\null\nobreak\hfil\squareforqed
\parfillskip=0pt\finalhyphendemerits=0\endgraf}\fi}

\newcommand{\MMM}{\CAL{M}}

\newcommand{\QQQ}{\CAL{Q}}

\newcommand{\FRESHQUANT}{\reflectbox{$\mathsf{N}$}}

\newcommand{\INCSPECb}[1]{\mathsf{inc'}(#1)}

\renewcommand{\SHORTASSERT}[3]{\{#1\}\; #2\  \{#3\}}

\renewcommand{\eval}[5]{\ASET{#1}\APP{#2}{#3}\!=\! #4\ASET{#5}}

\newcommand{\iReachable}[3]{\mathsf{reach}(#1,\ #2,\ #3)}

\newcommand{\directreach}{\rhd}

\newcommand{\HIDEa}{\OL{\nu}}
\newcommand{\HIDEe}{\nu}

\newcommand{\HIDEforall}{\HIDEa}

\newcommand{\allworlds}{\square\,}
\newcommand{\always}{\square\,}
\newcommand{\someworld}{\diamondsuit\,}

\newcommand{\piprm}[2]{\FPRM{#1}{#2}}

\newcommand{\ONEEVAL}[4]{#1 \bullet #2 \comesbefore #3\{#4\}}

\newcommand{\SIMPLEONEEVAL}[3]{#1 \bullet #2 \{#3\}}

\newenvironment{theoremapp}[2]{\begin{trivlist}\item[\hskip \labelsep {\bfseries Theorem #1.}] {\bfseries #2} \it}{\end{trivlist}}

\newenvironment{propositionapp}[2]{\begin{trivlist}\item[\hskip \labelsep {\bfseries Proposition #1.}] {{\bfseries #2}} \it}{\end{trivlist}}

\newcommand{\ORGEF}[1]{{#1}^\star} 

\def\fps@figure{tp}      
\def\fps@table{tp}
\newcommand{\modelsBsym}{\models_{\mathsf b}}
\newcommand{\modelsGuar}[2]%
{#1\modelsBsym #2}

\newcommand{\provesLA}[3]%
{#1\proves \! \mathbf{rely}\, #2 \, \mathbf{guar}\, #3} 

\newcommand{\provesLAs}[3]%
{#1\proves' \! \mathbf{rely}\, #2 \, \mathbf{guar}\, #3} 

\newcommand{\modelsLA}[3]%
{#1\models \! \mathbf{rely}\, #2 \, \mathbf{guar}\, #3}

\newcommand{\modelsILA}[3]%
{#1\models^I \! \mathbf{rely}\, #2 \, \mathbf{guar}\, #3} 

\newcommand{\LETa}[3]%
{\mathsf{let}\,#1\!=\!#2\,\mathsf{in}\,#3}

\newcommand{\LETb}[3]%
{\ENCda{#2/#1}#3}

\newcommand{\LETsubs}[3]%
{#3\MSUBS{#2}{#1}}

\newcommand{\LETpartial}[3]%
{\ENCda{#2/#1}#3}

\newcommand{\LETq}[5]%
{\ASET{#1}#2:_{(#3)}^{\nu#4}\ASET{#5}}

\newcommand{\LETprefix}[3]
{[#2/#1]#3}

\newcommand{\LETprefixB}[5]%
{\ASET{#1}#2\!=\!(#3)#4\ASET{#5}}

\newcommand{\SSUBSb}[2]{[#1/#2]}

\newcommand{\LETprefixb}[3]%
{\SSUBSb{#2}{#1}#3}

\newcommand{\LETc}[2]%
{\ENCda{#1}#2}

\newcommand{\LETstate}[5]%
{\ENCda{\ASET{#3}#2\ASET{#4}/#1}#5}

\newcommand{\provesState}[5]%
{#1\proves\! \ASET{#2}\,\mathbf{rely}\,#3\, \mathbf{guar}\, #4\,\ASET{#5}}

\newcommand{\provesLocal}[6]%
{#1\proves\!\,\ASET{#2}\,\mathbf{rely}\,#3\, \mathbf{guar}^{\mathbf{\nu}#4}\,#5\,\ASET{#6}}

\newcommand{\provesStateDist}[6]%
{#1\proves^{\!#6}\!\! \ASET{#2}\,\mathbf{rely}\,#3\, \mathbf{guar}\, #4\,\ASET{#5}}

\newcommand{\modelsStateDist}[6]%
{#1\models^{\!#6}\!\! \ASET{#2}\,\mathbf{rely}\,#3\, \mathbf{guar}\, #4\,\ASET{#5}}

\newcommand{\provesLAdist}[4]%
{#1\proves^{\!#4}\!\! \,\mathbf{rely}\,#2\, \mathbf{guar}\, #3}

\newcommand{\modelsLAdist}[4]%
{#1\models^{\!#4}\!\! \,\mathbf{rely}\,#2\, \mathbf{guar}\, #3}

\newcommand{\provesStateB}[5]%
{#1\proves\!\! \ASET{#2}\,\mathbf{rely}\,#3\, \mathbf{guar}\, #4\,\ASET{#5}}

\newcommand{\provesStateC}[7]%
{#1\proves\! \ASET{#2}\\
 #6 \mathbf{rely}\,#3\, \mathbf{guar}\,#4\\
 #7 \ASET{#5}}

\newcommand{\provesStateGuar}[2]%
{#1\proves \! \mathbf{guar}\, #2} 

\newcommand{\provesStateVal}[3]%
{\provesLA{#1}{#2}{#3}}

\newcommand{\modelsState}[5]%
{#1\models\! \ASET{#2}\,\mathbf{rely}\,#3\, \mathbf{guar}\, #4\,\ASET{#5}}

\newcommand{\modelsStateVal}[3]%
{\modelsLA{#1}{#2}{#3}}

\newcommand{\modelsStateGuar}[4]%
{#1\modelsBsym\! \ASET{#2}\, \mathbf{guar}\, #3\, \ASET{#4}}

\newcommand{\provesBT}[5]%
{#1\proves\! \ASET{#2}\,\mathbf{rely}\,#3\, \mathbf{guar}\, #4\,\ASET{#5}}

\newcommand{\modelsBT}[5]%
{#1\models\! \ASET{#2}\,\mathbf{rely}\,#3\, \mathbf{guar}\, #4\,\ASET{#5}}

\newcommand{\modelsBTguar}[4]%
{#1\modelsBsym\! \ASET{#2}\, \mathbf{guar}\, #3\, \ASET{#4}}

\newcommand{\provesBTram}[5]%
{#1\proves\! \ASET{#2}\, \mathbf{rely}\,#3\ \mathbf{guar}\, #4\, \ASET{#5}}

\newcommand{\modelsBTram}[5]%
{#1\models\! \ASET{#2}\, \mathbf{rely}\,#3\ \mathbf{guar}\, #4\ \ASET{#5}}

\newcommand{\modelsBTramI}[5]%
{#1\models^I\! \ASET{#2}\, \mathbf{rely}\,#3\ \mathbf{guar}\, #4\, \ASET{#5}}

\newcommand{\provesLocalVal}[5]%
{#1\proves\!\,(\NEW #3)\,\mathbf{rely}\,#4\, \mathbf{guar}\,#5}

\newcommand{\provesLocalb}[7]%
{#1\proves\!\,(\NEW #3)\,\ASET{#4}\,\mathbf{rely}\,#5\, \mathbf{guar}\,#6\,\ASET{#7}}

\newcommand{\provesLocalGuar}[3]%
{#1\proves\!\, (\NEW #2)\, \mathbf{guar}\, #3}

\newcommand{\provesLocalValNew}[5]%
{#1\models\!\,(\NEW #3)\,\mathbf{rely}\,#4\, \mathbf{guar}\,#5}

\newcommand{\provesLocalValCoNew}[4]%
{#1\proves\!\,(\CONEW #2)\,\mathbf{rely}\,#3\, \mathbf{guar}\,#4}

\newcommand{\modelsLocalVal}[5]%
{#1\models\!\,(\NEW #3)\,\mathbf{rely}\,#4\, \mathbf{guar}\,#5}

\newcommand{\modelsLocal}[7]%
{#1\models\!\,(\NEW #3)\,\ASET{#4}\,\mathbf{rely}\,#5\, \mathbf{guar}\,#6\,\ASET{#7}}

\newcommand{\modelsLocalGuar}[3]%
{#1\models\! (\NEW #2)\, \mathbf{guar}\, #3}

\newcommand{\modelsLocalValNew}[5]%
{#1\models\!\,(\NEW #3)\,\mathbf{rely}\,#4\, \mathbf{guar}\,#5}

\newcommand{\modelsLocalNew}[7]%
{#1\models\!\,(\NEW #3)\,\ASET{#4}\,\mathbf{rely}#5\, \mathbf{guar}\,#6\,\ASET{#7}}

\newcommand{\CONEW}{\OL{\new\!}\,}

\newcommand{\provesLocalRam}[5]%

\newcommand{\modelsLocalRam}[5]%
{#1\models\! \ASET{#2}\, \mathbf{rely}\,#3\ \mathbf{guar}\, #4\ \ASET{#5}}

\newcommand{\modelsStateRamI}[5]%
{#1\models^I\! \ASET{#2}\, \mathbf{rely}\,#3\ \mathbf{guar}\, #4\, \ASET{#5}}

\newcommand{\APP}[2]{#1\bullet#2}

\newcommand{\APPs}[2]{#1\!\bullet\!#2}

\newcommand{\ENTAILS}{\supset}
\newcommand{\entails}{\supset}

\newcommand{\IFF}{\equiv}

\newcommand{\AND}{\wedge}
\newcommand{\OR}{\vee}

\newcommand{\ANDl}{\ \wedge\ }
\newcommand{\ANDs}{\!\wedge\!}
\newcommand{\ANDls}{\,\wedge\,}

\newcommand{\modelsBTs}[5]%
{#1\modelsBsym\!%
\mathbf{pre}\,#2\, 
\mathbf{rel}\,#3\, 
\mathbf{gua}\, #4\,
\mathbf{post}\, #5}

\newcommand{\NewLine}%
{\qquad\qquad\qquad    \qquad\qquad\qquad    \qquad\qquad\qquad    \qquad\qquad\qquad\\  {\ }\\[1.8mm]}

\newcommand{\NewLineB}%
{\qquad\qquad\qquad    \qquad\qquad\qquad    \qquad\qquad\qquad    \qquad\qquad\qquad\\  {\ }\\[0.2mm]}

\newcommand{\NewLineBB}%
{\qquad\qquad\qquad    \qquad\qquad\qquad    \qquad\qquad\qquad    \qquad\qquad\qquad\\  {\ }\\[1.2mm]}

\newcommand{\NewLineC}%
{\qquad\qquad\qquad    \qquad\qquad\qquad    \qquad\qquad\qquad    \qquad\qquad\qquad\\  {\ }\\[-0.2mm]}

\newcommand{\NewLineE}[1]%
{\qquad\qquad\qquad    \qquad\qquad\qquad    \qquad\qquad\qquad    \qquad\qquad\qquad\\  {\ }\\[#1mm]}

\newcommand{\JustNewLine}%
{{\ }\\  {\ }\\[1.8mm]}

\newcommand{\pcom}[1]{\quad\text{\sf #1}}

\newcommand{\IMPLIES}{\ENTAILS}

\newcommand{\LOGICEQ}{\stackrel{\mathsf{def}}{\LITEQ}}

\def\doi{4 (4:2) 2008}
\lmcsheading%
{\doi}
{1--68}
{}
{}
{Sep.~23, 2007}
{Oct.~20, 2008}
{}   

\begin{document}

\title[Logical Reasoning for Higher-Order Functions with Local State]{Logical Reasoning for Higher-Order Functions with Local State}

\author[N.~Yoshida]{Nobuko Yoshida\rsuper a}      
\address{{\lsuper{a,c}}Department of Computing, Imperial College London, 180 Queen's Gate, London, 2AZ SW7}   
\email{\{yoshida,M.Berger\}@doc.ic.ac.uk}  

\author[K.~Honda]{Kohei Honda\rsuper b}  
\address{{\lsuper b}Department of Computer Science, Queen Mary, University of London, Mile End Road, London, 1E 4NS}       
\email{kohei@dcs.qmul.ac.uk}  

\author[M.~Berger]{Martin Berger\rsuper c}        



\keywords{Languages, Theory, Verification} 

\subjclass{D.3.1, D.3.3, D.3.2, F.3.1, F.3.2, F.4.1}


\begin{abstract}
We introduce an extension of Hoare logic for call-by-value
higher-order functions with ML-like local reference generation.  Local
references may be generated dynamically and exported outside their
scope, may store higher-order functions and may be used to construct
complex mutable data structures. This primitive is captured logically
using a predicate asserting reachability of a reference name from a
possibly higher-order datum and quantifiers over hidden
references. 
We explore the logic's descriptive and reasoning power with non-trivial
programming examples combining higher-order procedures and dynamically
generated local state.  Axioms for reachability and local invariant
play a central role for reasoning about the examples.

\end{abstract}

\maketitle


\setcounter{tocdepth}{2}
\tableofcontents



\section{Introduction}
\label{sec:introduction}

\paragraph{\bf Reference Generation in Higher-Order Programming.}\ 
This paper proposes an extension of Hoare Logic \cite{HOARE} for
call-by-value higher-order functions with ML-like new reference
generation \cite{SML,CAML}, and demonstrates its use through
non-trivial reasoning examples.  New reference generation, embodied
for example in ML's $\tt ref$-construct, is a highly expressive
programming primitive.  The first key functionality of this construct
is to introduce local state into the dynamics of programs by generating a
fresh reference inaccessible from the outside.  Consider the following
program:
 \begin{equation} \label{ex:IncPrg}
        \IncPrg \DEFEQ \mathtt{let} \ x =\refPrg{0} \ \mathtt{in} \
\lambda ().(x:=!x+1;\,!x)
\end{equation}
where ``$\refPrg{M}$'' returns a fresh reference whose content is the value which $M$ evaluates
to; ``$!x$'' denotes dereferencing the imperative
variable $x$; and ``$;$'' is sequential composition.
In (\ref{ex:IncPrg}), a reference with content $0$ is
newly created,  but  never exported to the outside. 
When the anonymous function in $\IncPrg$ is invoked, it 
increments the content of the local variable $x$, and returns the new
content. The procedure returns a different
result at each call, whose source is hidden from  external
observers. This is different from $\lambda ().(x:=!x+1;\,!x)$ where
$x$ is globally accessible. 

Secondly, local references may be exported outside of
their original scope and be shared, contributing 
to the expressivity of significant imperative idioms.
Let us show how stored procedures interact with 
new reference generation and  sharing of references.  
We consider the following program from \cite[\S~6]{PittsAM:operfl}:
\begin{equation}\label{ex:IncPrgshared}
\IncShared
\;\DEFEQ\;
a\!:=\!\IncPrg;
b\!:=\!!a;
z_1\!:=\!(!a)();
z_2\!:=\!(!b)();
(!z_1+!z_2)
\end{equation}
The initial content of the hidden $x$ is $0$. Following the standard
semantics of ML \cite{MilnerR:defostaML}, the assignment $b:=!a$
copies the code (or a pointer to the code) from $a$ to $b$ while
sharing the store $x$.  Hence the content of $x$ is incremented every
time the functions stored in $a$ and $b$, sharing the same
store $x$, are called,  returning $3$ at the 
end of the program $\IncShared$.  To
understand the behaviour of $\IncShared$ precisely and give it an
appropriate specification, we must capture the sharing of $x$ between
the procedures assigned to $a$ and $b$.  From the viewpoint of
visibility, the scope of $x$ is originally restricted to the function
stored in $a$ but gets extruded to and shared by the one stored in
$b$.  If we replace $b:= !a$ by $b:= \mathtt{Inc}$ as follows, two
separate instances of $\IncPrg$ (hence with separate hidden stores)
are assigned to $a$ and $b$, and the final result is not $3$ but $2$. 
\begin{equation}\label{ex:IncPrgshared2}
\IncUnshared
\;\DEFEQ\;
a\!:=\!\IncPrg;
b\!:=\!\IncPrg;
z_1\!:=\!(!a)();
z_2\!:=\!(!b)();
(!z_1+!z_2)
\end{equation}
Controlling the  sharing of local references
is essential 
for writing concise algorithms that manipulate
functions with shared store, or
mutable data structures such as trees and graphs, 
but complicates  formal reasoning,
even for relatively small programs 
\cite{HOAREDataRep,MilnerDataSim,MeyerAR:towfasflv}.



Thirdly, through information hiding, local references can be used for
efficient implementations of highly regular observable behaviour, for
example, purely functional behaviour.
The following program, taken from 
\cite[\S~1]{PittsAM:operfl}, called $\mathtt{memFact}$, 
is a simple memoised factorial.
 \begin{eqnarray}
 \mathtt{memFact}
 & \DEFEQ & 
 \mathtt{let}
 \ a =\refPrg{0}, \ 
 \ b =\refPrg{1} 
 \ \mathtt{in}\nonumber\\ 
& &  
\lambda x.\IFTHENELSE{x=!a}{!b}{(a:=x\,;\,b:=\mathtt{fact}(x)\,;\,!b)}
\label{memfact}
 \end{eqnarray}
Here $\mathtt{fact}$ is the standard factorial function.  To external
observers, $\mathtt{memFact}$ behaves purely functionally.  The
program implements a simple case of memoisation: when
$\mathtt{memFact}$ is called with a stored argument in $a$, it
  immediately returns the stored value $!b$ without calculation. If $x$
differs from  $a$'s content, the factorial $f x$ is calculated
and the new pair is stored. 
\martinb{For complex functions, memoisation can lead to substantial speedups, but}
for this to
be meaningful we  need a memoised function to behave
indistinguishably from the original function except for efficiency.
So we ask: why can we say $\mathtt{memFact}$ is indistinguishable from
the pure factorial function? The answer to this question can be articulated 
clearly through a {\em local
invariant property} \cite{PittsAM:operfl} which 
can be  stated informally as follows:
\begin{quote}
{\em Throughout all possible invocations of $\mathtt{memFact}$,
  the content of $b$ is the factorial of
  the content of $a$.}
\end{quote}
Such local invariants capture one of the basic patterns in programming with
local state, and play a key role in preceding studies of
operational reasoning about program equivalence in the presence of local state
\cite{PittsAM:operfl,stark:namhof,PittsAM:realvo,KoutovasWand06}. 
Can we distill this principle axiomatically and use it
to
validate efficiently properties of higher-order programs with local state
such as $\mathtt{memFact}$?  

As a further example of local invariants,  this time
involving mutually recursive stored functions,   
consider the following program:
\begin{equation}
\begin{array}{rcl}
\mathtt{mutualParity}  \DEFEQ
&
x  :=  \lambda n.\IFTHENELSE{n\!=\!0}{\mathtt{f}}{\mathtt{not}((!y)(n\!-\!1))};\\
&
y  :=  \lambda n.\IFTHENELSE{n\!=\!0}{\mathtt{t}}{\mathtt{not}((!x)(n\!-\!1))}
\end{array}
\end{equation}
After running $\mathtt{mutualParity}$, the application $(!x)n$
returns $\mathtt{t}$  if $n$ is odd and otherwise $\mathtt{f}$;  $(!y)n$ acts dually.
But since $x$ and $y$ are free, a program may disturb 
$\mathtt{mutualParity}$'s functioning by inappropriate assignment: 
if a program reads from $x$ and stores
it in another variable, say $z$, assigns 
a diverging function to $x$, and feeds 
the content of $z$ with $7$, then the program diverges rather than
returning $\mathtt{t}$. 

With  local state, 
we can avoid unexpected interference at $x$
and $y$.
\begin{align}
\mathtt{safeOdd} \ & \ \DEFEQ  & 
\mathtt{let} \ x= \refPrg{\lambda n.\mathtt{t}},  
 \ y = \refPrg{\lambda n.\mathtt{t}} 
\ \mathtt{in} \ (\mathtt{mutualParity};!x)
\label{safeOdd}
\\[-2mm]
\mathtt{safeEven} \ & \ \DEFEQ  & 
\mathtt{let} \ x= \refPrg{\lambda n.\mathtt{t}}, 
\ y = \refPrg{\lambda n.\mathtt{t}} 
\ \mathtt{in} \ (\mathtt{mutualParity};!y)
\label{safeEven}
\end{align}
(Here $\lambda n.\mathtt{t}$ can be any initialising value.)  
Now that $x,y$ are inaccessible,
the programs behave like pure functions, 
e.g.~$\mathtt{safeOdd}(3)$ always returns $\mathtt{true}$ without any side effects.
Similarly $\mathtt{safeOdd}(16)$ always returns $\mathtt{f}$.
In this
case, the invariant says: 
\begin{quote}
{\em Throughout all possible invocations, 
\martinb{$\mathtt{safeOdd}$}
is a procedure which checks if its argument is odd, provided
$y$ stores a procedure which does the dual, whereas
\martinb{$\mathtt{safeEven}$}
is a 
procedure which checks if its argument is even, whenever 
 $x$ stores a dual  procedure}.  
\end{quote}
Later we present
general reasoning principles for local invariants which can verify
properties of these two and many other non-trivial examples 
\cite{MeyerAR:towfasflv,PittsAM:realvo,PittsAM:operfl,MT92a,MT92b,KoutovasWand06}.  

\paragraph{\it \bfseries {Contribution.}}
This paper studies a Hoare logic for imperative higher-order functions
with dynamic reference generation, a core part of ML-like languages.
Starting from their origins in the $\lambda$-calculus, the syntactic
and semantic properties of typed higher-order functional programming
languages such as Haskell and ML have been  studied extensively, making
them an ideal target for the formal validation of properties of
programs on a rigorous semantic basis.  Further, given the expressive
power of imperative higher-order functions (attested to by the  encodability
of objects
\cite{ComparingObjectEncodings,PierceTurner:Object,PierceBC:typsysfpl}
and of low-level idioms \cite{shao97:_overv_flint_ml_compil}), a study
of logics for these languages may have wide repercussions on logics of
programming languages in general.

Such languages 
\cite{SML,CAML} 
combine higher-order
functions and imperative features including new reference generation.
Extending Hoare logic to these languages 
leads to technical difficulties 
due to  three fundamental features:
\begin{BITEM}

\item Higher-order functions, including stored ones.

\item General forms of aliasing induced by nested reference
  types. 

\item Dynamically generated local references and scope extrusion.

\end{BITEM}
The first is the central feature of these languages; the second arises
by allowing reference types to occur in other types;  the third 
feature has been discussed above.
In  preceding studies, we built Hoare logics
for  core parts of ML which cover the first two features
\cite{SHORT1,HY04PPDP,GLOBAL,ALIAS}.
On the basis of these works, the present work introduces an extension
of Hoare logic for ML-like local reference generation.  As noted
above, this construct  enriches programs' behaviour radically,
and has so far defied  clean logical and axiomatic treatment.
A central challenge is to identify simple but expressive logical
primitives, proof rules (for Hoare triples) and axioms (for
assertions), enabling tractable assertions and verification.

The program logic proposed in the present paper introduces a predicate
representing reachability of a reference from an arbitrary datum
in order to represent new reference generation. Since we are working
with higher-order programs, a datum and a reference may as well be, or
store, a higher-order function. We shall show that this predicate is
fully axiomatisable using (in)equality when it only involves
first-order data types (the result is closely related with known
axiomatisations of reachability \cite{NelsonReachability}).
However we shall also show that the predicate becomes 
undecidable  when higher-order types are involved,
indicating an inherent intractability.

A good news is, however, that this predicate enables us,
when combined with a pair of mutually dual
hiding quantifiers (i.e.~quantifiers ranging over variables denoting
hidden references),
to obtain a simple
compositional proof rule for new reference generation, preserving all
the compositional proof rules for the remaining constructs from our
foregoing program logics. 

At the level of assertions, we can find a set of useful axioms for
(un)reachability and the hiding quantifiers, which are 
effectively combined with logical primitives and associated axioms for
higher-order functions and aliasing studied in our preceding works
\cite{GLOBAL,ALIAS}. These axioms for reachability and hiding
quantifiers are closely related with reasoning principles studied in
existing semantic studies on local state, \martinb{such as the principle of
local invariants \cite{PittsAM:operfl}}.
The local invariant axioms capture common patterns
in reasoning about local state, and enable us to verify the examples
in
\cite{MeyerAR:towfasflv,PittsAM:realvo,PittsAM:operfl,MT92a,MT92b,KoutovasWand06}
axiomatically, including programs discussed above. 
The program logic also satisfies strong completeness properties including
the standard relative completeness as discussed later.
 As a whole, our
program logic offers an expressive reasoning framework where (relatively)
simple programs such as pure functions can be reasoned about using simpler
primitives while programs with more  complex behaviour
such as those with non-trivial
use of local state are reasoned about using incrementally more involved
logical constructs and axioms.


\paragraph*{\bf Outline.}  
This paper is a full version of \cite{YHB07}, with complete definitions 
and detailed explanations and proofs. The present version not only gives more
detailed analysis for the properties of the models, axioms and proof rules, 
but also more examples with full derivations and 
comprehensive comparisons with related work. 

Section \ref{sec:languages} presents the 
programming language and the assertion language. Section \ref{sec:models}
gives the semantics of the logic. 
Section \ref{sec:rules} 
proposes the proof rules and proves  soundness. 
Section \ref{sec:axioms} 
explores axioms of the assertion language.
Sections \ref{sec:example} discusses the
use of the logic through non-trivial reasoning examples
centring on local invariants. 
Section \ref{sec:conclusion} summarises extensions 
including the three completeness results of the logic, 
gives the comparisons
with related works, and concludes with further topics. 
Appendix lists auxiliary
definitions and detailed proofs. 
Larger examples of reasoning about mutable data structures 
can be found in \cite{localexample}. 

\section{Assertions for Local State}
\label{sec:languages}
\subsection{A Programming Language}
\label{subsec:languages}
\label{subsec:PL}
\NI Our target programming language is call-by-value PCF with
unit, sums, products and recursive types, 
augmented with imperative constructs.  
Let $a,b,...,x, y, \ldots$ range over an
infinite set of variables, 
and $\TVX, \TVY, \ldots$ over an infinite set of type 
variables.\footnote{For simplicity, 
we omit the polymorphism from the language, see \cite{HY04PPDP}.} 
Then types, values and programs are given by: 
\[
\begin{array}{lcl}
        \alpha,\beta
                &::=&
        \UNIT 
                \ | \
        \BOOL
                \ | \
        \NAT
                \ | \
        \alpha\FS\beta 
                \ | \
        \alpha\times\beta 
                \ | \
        \alpha+\beta 
                \ | \
        \REF{\alpha}{}
                \ | \
         \TVX
                \ | \
        \mu \TVX.\alpha
                \\[1mm]
        V,W
                &::=&  
        \mathtt{c}
                \ |\ 
        x^\alpha
                \ |\ 
        \lambda x^\alpha.M
                \ |\ 
        \mu f^{\alpha\FS \beta}.\lambda y^{\alpha}.M
                \ |\ 
        \PAIR{V}{W}
                \ |\ 
        \INJtp{i}{V}{\alpha+\beta}
                \\[1mm]
        M,N 
                &::=& 
        V
                \ | \
        MN
                \ | \
        M := N
                \ | \
        \refPrg{M}       
                \ |\
        !M
       \ | \ 
        \mathtt{op}(\VEC{M})
                \ |\
        \pi_i(M)
                \ |\
        \PAIR{M}{N}
                \ |\
        \INJtp{i}{M}{\alpha+\beta}
     \\[1mm]
                & \ | \ & 
        \IFTHENELSE{M}{M_1}{M_2}
                \ |\
        \CASE{M}{x_i^{\alpha_i}}{M_i} 
\end{array} 
\]
We use  standard notation \cite{PierceBC:typsysfpl,GunterCA:semprol} 
like constants $\mathtt{c}$ 
(unit $()$; booleans $\mathtt{t}$, $\mathtt{f}$; numbers $\mathtt{n}$;
and location labels also called simply \emph{locations} $l,l',...$) and 
first-order operations $\mathtt{op}$ 
($+,\,-,\,\times,\,=,\, \neg,\,\AND,\,\ldots$).
Locations only appear at runtime when references are generated. 
$\VEC{M}$ etc. denotes a vector and $\NUL$ the empty vector. 
A program is {\em closed} if it has no free variables.
Note that a closed program might contain free locations. 
We use  abbreviations such as:
\begin{eqnarray*}
	\lambda ().M 
		&\quad\DEFEQ\quad& 
	\lambda x^{\UNIT}.M\qquad{(x\not\in\fv{M})}
		\\ 
	M;N 
		&\quad\DEFEQ \quad& 
	(\lambda ().N)M
		\\ 
	\LET{x}{M}{N}&\quad\DEFEQ\quad& (\lambda x.N)M
		\qquad
	{(x\not\in\fv{M})}
\end{eqnarray*}
We use the standard notion of types for imperative $\lambda$-calculi
\cite{PierceBC:typsysfpl,GunterCA:semprol}
and use the equi-isomorphic approach
\cite{PierceBC:typsysfpl} for recursive types. 
$\NAT$, $\BOOL$ and $\UNIT$ are called {\em base types}.
We leave the 
illustration of each language construct to standard
textbooks 
\cite{PierceBC:typsysfpl}, except for  
reference generation $\refPrg{M}$, the focus of the
present study.  
$\refPrg{M}$ behaves as follows: first $M$ of type $\alpha$ is evaluated
and becomes 
a value $V$; then a
\emph{fresh} reference of type $\REF{\alpha}$ with 
initial content $V$ is generated. 

The behaviour of the programs is formalised by the 
reduction rules. 
Let  $\sigma$ denote  a {\em store}, a finite map from locations to closed
values. We use $\sigma\uplus [l\mapsto
V]$ to denote the result of disjointly adding a pair $(l, V)$ to $\sigma$.
A \emph{configuration} is of the form $(\nu \VEC{l})(M, \sigma)$ where $M$ is
a program, $\sigma$ a store, and $\VEC{l}$ a vector of distinct locations (the order 
is irrelevant) occurring in $\sigma$, and
hidden by $\nu$. The need of $\nu$-biniding is discussed in
\S~\ref{sub:ex:local} and Remark \ref{rem:hidden}.

A {\em reduction relation}, or often {\em reduction} 
for short, is a binary relation
between configurations, written
\[
(\new \VEC{l})(M, \sigma_1)\;\RED\;
(\new \VEC{l}')(N, \sigma_2)
\]
The relation is generated by the following rules. 
First we have the
standard rules for 
call-by-value PCF:
\[
\begin{array}{c}
\begin{array}{rcl}
\inferBase
{-}
{
        (\lambda x.M)V & \;\ \rightarrow\;\ & M\MSUBS{V}{x}
}
        \linebreakA
\inferBase
{-}
{
        \pi_1(\PAIR{V_1}{V_2}) &\rightarrow & V_1
}
        \linebreakA
\inferBase
{-}
{
        \IFTHENELSE{\mathtt{t}}{M_1}{M_2}
               & \rightarrow &
        M_1
}
        \linebreakA
\inferBase
{
  - 
}
{
        (\mu f. \lambda g.N)W 
        & \rightarrow &
        N\MSUBS{W}{g}\MSUBS{\mu f. \lambda g.N}{f}
}
\commentOut{     \linebreakA
\inferBase
{
        -
}
{
  \LET{x}{V}{N} 
  & \rightarrow &
  N\MSUBS{V}{x}
}}
\end{array}
        \\[10mm]
        \CASE{\inj{1}{W}}{x_i}{M_i}
                \rightarrow 
        M_1\MSUBS{W}{x_1}
\end{array}
\]
Then we have the reduction rules for imperative
constructs, i.e. assignment, dereference and new-name generation.
\[ 
\begin{array}{rcl}
\inferBase
{
  -
}
{
        (!l,\;\ \sigma) & \;\ \rightarrow\;\ & (\sigma(l),\;\ \sigma)
} 
        \linebreakA
\inferBase
{
  -
}
{
        (l := V,\;\ \sigma)  & \;\ \rightarrow\;\ 
        & ((),\;\ \sigma[l \mapsto V])
}
        \linebreakA
        {(\refPrg{V},\ \sigma)} 
              & \;\ \rightarrow\;\ &
        \New{l}{(l,\ \sigma\uplus [l\mapsto V])} 
\end{array}
\]
In the reduction rule for  references, the resulting configuration 
uses a \emph{$\new$-binder}, which lets us 
directly capture the observational meaning of programs.
\NI Finally we close
$\RED$ under evaluation contexts and $\new$-binders.
\[ 
\begin{array}{c}
\infer
{
        (\nu \VEC{l}_1)(M, \sigma) \rightarrow (\nu \VEC{l}_2)(M', \sigma')
}
{
        (\nu \VEC{l}\VEC{l}_1)(\RC{M}, \sigma) \rightarrow (\nu \VEC{l}\VEC{l}_2)(\RC{M'}, \sigma') 
}
\end{array}
\]
where $\VEC{l}$ are disjoint from both $\VEC{l}_1$ and $\VEC{l}_2$,
$\RCCD$ is the left-to-right call-by-value evaluation context (with eager
evaluation), inductively given by: 
\[
\begin{array}{rcl}
\RCCD
        &\;::=\;&
(\RCCD M)
        \,\;\ |\,\;\ 
(V \RCCD)
        \,\;\ |\,\;\ 
\PAIR{V}{\RCCD}
        \,\;\ |\,\;\ 
\PAIR{\RCCD}{M}
        \,\;\ |\,\;\ 
\pi_i(\RCCD)
        \,\;\ |\,\;\ 
\inj{i}{\RCCD}
\\
        &|&     
\mathtt{op}(\VEC{V},\RCCD,\VEC{M})
        \,\;\ |\,\;\ 
\IFTHENELSE{\RCCD}{M}{N}
        \,\;\ |\,\;\ 
\CASE{\RCCD}{x_i}{M_i}
        \\
        &|&     
!\RCCD
\,\;\ |\,\;\ 
\RCCD:=M
\;\ | \;\
V:=\RCCD\
\;\ | \;\
\refPrg{\RCCD}       
\end{array}
\]
We write $(M, \sigma)$ for
$\New{\NUL}{(M, \sigma)}$ with $\NUL$ denoting the empty vector. 
We define:
\begin{enumerate}[$\bullet$]
\item 
$(\nu \VEC{l})(M, \sigma)\CONVERGES (\nu \VEC{l}')(V, \sigma')$ 
means 
$(\nu \VEC{l})(M, \sigma) \rightarrow^* (\nu \VEC{l}')(V, \sigma')$ 

\item 
$(\nu \VEC{l})(M, \sigma)\CONVERGES$ 
means 
$(\nu \VEC{l})(M, \sigma)\CONVERGES (\nu \VEC{l}')(V, \sigma')$
for some $(\nu \VEC{l}')(V, \sigma')$
\end{enumerate}
\label{obscon:models}
An \emph{environment} $\Gamma,\Delta,...$ is a finite map 
from variables to types and 
from locations to reference types. 
The typing rules are standard  
\cite{PierceBC:typsysfpl} 
and are left to Appendix \ref{app:typing}. 
Sequents have the form $\TYPES{\Gamma}{M}{\alpha}$, 
to be read: $M$ has type $\alpha$ under $\Gamma$.
A store $\sigma$ is typed under
$\Delta$, written $\Delta\proves \sigma$, 
when, for each $l$ in its domain, $\sigma(l)$ is a
closed value which is typed $\alpha$ under $\Delta$, 
where we assume $\Delta(l)=\refType{\alpha}$.
A configuration $(M, \sigma)$ is {\em well-typed} if
for some $\Gamma$ and $\alpha$ we have $\Gamma\proves M:\alpha$ and 
$\Gamma\proves \sigma$.
Standard type safety holds for well-typed configurations.
\emph{Henceforth we 
only consider well-typed programs and configurations.}  

We 
define the observational congruence between configurations. 
Assume $\Gamma,\VEC{l}_{1,2}:\VEC{\alpha}_{1,2}\proves
M_{1,2}:\alpha$
\martinb{and $\Gamma,\VEC{l}_{1,2}:\VEC{\alpha}_{1,2}\proves \sigma_{1, 2}$}.
 Write 
$$\Gamma\proves (\new
\VEC{l}_1)(M_1,\sigma_1) \cong (\new \VEC{l}_2)(M_2,\sigma_2)$$
if, for
each typed context $C[\ \CD\ ]$ which produces a closed program which
is typed as $\UNIT$ under $\Delta$ and in which no labels from
$\VEC{l}_{1,2}$ occur, the following holds:
\begin{equation*}
(\new \VEC{l}_1)(C[M_1],\ \sigma_1)\converges
\quad\text{iff}\quad
(\new \VEC{l}_2)(C[M_2],\ \sigma_2)\converges
\end{equation*}
which we often write $(\new \VEC{l}_1)(M_1,\sigma_1) \cong 
(\new \VEC{l}_2)(M_2,\sigma_2)$ 
leaving type information implicit.
We also write $\Gamma\proves M_1 \cong M_2$,  
or simply $M_1 \cong M_2$ 
leaving type information implicit, if, $\VEC{l}_i=\sigma_i=\emptyset$ 
($i=1,2$).

\subsection{A Logical Language}
\label{sub:assertion}
\NI The logical language we shall use is that of standard first-order
logic with equality \cite[\S~2.8]{MENDELSON}, extended with the
constructs for (1) higher-order application \cite{HY04PPDP,GLOBAL}
(for imperative higher-order functions); (2) quantification over
store content \cite{ALIAS} (for aliasing); (3) reachability and
quantifications over hidden names (for local state).  For (1) we
decompose the original construct \cite{HY04PPDP,GLOBAL} 
into more elementary constructs, 
which becomes important for precisely capturing  the semantics of
higher-order programs with local state and for obtaining 
strong completeness properties of the logic, as we shall
discuss in later sections.

The grammar follows, letting $\star\in\ASET{\AND, \OR, \ENTAILS}$, 
$\CAL{Q}\in\ASET{\exists,\forall,\HIDEe,\HIDEa}$ 
and $\CAL{Q}'\in\ASET{\exists,\forall}$.
\[
\begin{array}{lcl}
        e &\ \;::=\;&\ \;
        x
            \;\ |\;\  
        \LOGIC{c}
            \;\ |\;\  
        \LOGIC{op}(\VEC{e})
            \;\ |\;\  
        \PAIR{e}{e'}
           \;\ |\;\   
        \LOGICINJ{\alpha_1+\alpha_2}{i}{e}
            \;\ |\;\  
        !e
                \\[1.2mm]
        C &\;::=\;&\;
        e\!=\!e'
            \;\ |\;\  
        \neg C
            \;\ |\;\  
        C \star C'
            \;\ |\;\  
        \CAL{Q} x^\alpha.C
           \;\ |\;\  
        \CAL{Q}'\TVX.C 
           \;\ |\;\  
          \allCon{{e}}{C}
             \;\ |\;\  
         \someCon{e}{C}
          \\[1.2mm]
            &|&    
        \ONEEVAL{e}{e'}{x}{C}
            \;\ |\;\ 
         \allworlds C 
            \;\ |\;\ 
         \someworld C 
            \;\ |\;\ 
        \REACH{e}{e'}
             \ |\
         e\noreach e'
\end{array}
\]
The first grammar ($e, e', \ldots$) defines {\em terms};
the second {\em formulae} ($A,B,C, C',E, \ldots$). Terms include variables,
constants $\mathsf{c}$ 
(unit $()$, numbers $\LOGIC{n}$, booleans $\LOGIC{t}$, $\LOGIC{f}$ and
locations $l,l',...$), 
pairing, injection and standard first-order operations.
$\DEREF{e}$ denotes the dereference
of a reference $e$.  
Formulae include  standard logical connectives and first-order quantifiers
\cite{MENDELSON}. 

The remaining constructs in the logical language are for capturing  the
behaviour of imperative higher-order functions with local state.
First, the universal and existential quantifiers,  
$\forall x.C$ and $\exists x.C$, are standard. 
We include, following \cite{HY04PPDP,ALIAS}, 
quantification over type variables  
($\TVX, \TVY, \ldots$).
We also use the two quantifiers for aliasing 
introduced in  \cite{ALIAS}. 
$\allCon{x}{C}$ 
is \emph{universal content quantification of $x$ in $C$}, while
$\someCon{x}{C}$
is \emph{existential content quantification of $x$ in $C$}.
In both, $x$ should have a reference type.
$\allCon{x}{C}$ says $C$ holds regardless of the value stored in
  a memory cell named $x$; and 
$\someCon{x}{C}$ says $C$ holds for some value that may be
  stored in the memory cell named $x$.
In both, what is being quantified is the content of a store, {\em
  not} the name of that store. 
In 
$\allCon{x}{C}$
and  
$\someCon{x}{C}$,
$C$ is the {\em
  scope} of the quantification.
The free variable $x$ is not a binder: we have
$\FV{\someCon{x}{C}} = \FV{\allCon{x}{C}} = \{x\} \cup \FV{C}$
where 
$\FV{C}$ denotes the set of free variables in $C$.  
We define $\someCon{e}{C}$ as a shorthand for $\exists x.(x = e \AND \someCon{x}{C})$,
assuming $x \notin \FV{C}$. Likewise, $\allCon{e}{C}$ is short for
$\forall x.(x = e \IMPLIES \allCon{x}{C})$ with $x$ being fresh.
The scope of a content quantifier is as small as possible, 
e.g.~$\allCon{x}{C} \IMPLIES C'$ 
stands for $(\allCon{x}{C}) \IMPLIES C'$.

Decomposing the original evaluation formulae
\cite{HY04PPDP,GLOBAL} into $\ONEEVAL{e}{e'}{x}{C}$ and $\allworlds C$,
is used  for describing the behaviour of functions.\footnote{We later
show $\allworlds C$ is expressible by
$\ONEEVAL{e}{e'}{x}{C}$: nevertheless treating $\allworlds C$
independently is convenient for our technical development.}
$\ONEEVAL{e}{e'}{x}{C}$, which we call (one-sided) {\em evaluation
formula}, intuitively says:
\begin{quote}
{\em The application of a function $e$ to an argument $e'$ starting from the {\em
present} state will terminate with a resulting value (name it $x$) and
a final state, together satisfying $C$}.
\end{quote}
whereas $\allworlds C$, which we read
{\em always $C$}, intuitively means: 
\begin{quote}
{\em $C$ holds in any possible state reachable
from the current one}.
\end{quote}
Its dual is written $\someworld C$ (defined as
$\neg \allworlds \neg C$), which we read {\em someday $C$}. 
We call $\allworlds$ (resp.~$\someworld$) {\em necessity} 
(resp.~{\em possibility})  
operators.   
As a typical usage of these  primitives, consider:
\begin{equation} \label{encoding:A}
\allworlds (C \ENTAILS \ONEEVAL{f}{x}{y}{C'})
\end{equation}
This can be read: ``for now or any future state, once $C$ holds, then
the application of $f$ to $x$ terminates, with both a return value $y$ and a 
final state satisfying $C'$''.  Note that (\ref{encoding:A}) 
corresponds to the original evaluation formula in \cite{HY04PPDP,GLOBAL}. 
\martinb{Further,  in the presence of local
state, (\ref{encoding:A}) can describe  situations which
cannot be represented using the  original evaluation formula (see \S~\ref{sub:ex:local} for examples). The
decomposition (\ref{encoding:A})  can also 
 generalise the local invariant axiom in Proposition 
\ref{pro:localinv:firstorder} from \cite{YHB07}.  }
Thus this decomposed
form is strictly more \mfb{expressive. It} also allows a more streamlined
theory.


There are two new logical primitives for representing local state ---
in other words, for describing the effects of generating and using a
fresh reference.  First, the {\em hiding-quantifiers}, $\HIDEe x.C$
({\em for some hidden reference $x$, $C$ holds}) and $\HIDEa x.C$
({\em for each hidden reference $x$, $C$ holds}), quantify over
reference variables, i.e.~the type of $x$ above should be of the form
$\REF{\beta}{}$.  These quantifiers range over hidden references, such
as $x$ generated by $\IncPrg$ in (\ref{ex:IncPrg}) in 
\S~\ref{sec:introduction}. The need for having these quantifiers in
addition to the standard ones is illustrated in 
\S~\ref{sub:ex:local} and Remark \ref{rem:hidden}. The 
formal difference of $\HIDEe$ as a quantifier from $\exists$ will be
clarified in \S~\ref{subsec:axiom:hiding}, Proposition
\ref{pro:hiding}.

The second new primitive for local state is $\REACH{e_1}{e_2}$ (with
$e_2$ of a reference type), which we call  {\em reachability predicate}.
This predicate says: 
\begin{quote}
{\em We can reach the reference denoted by $e_2$ from a datum
denoted by $e_1$.  }  
\end{quote}
As an example, if $x$ denotes a starting point
of a linked list, $\REACH{x}{y}$ says a reference $y$ occurs in one of
the cells reachable from $x$.  We set its dual
\cite{SPECIFICATION,GabbayM:newapprasib}, written $e \noreach e'$,
to mean $\neg e' \reachable e$. This negative form says: 
\begin{quote}
{\em One can never reach a
reference $e$ starting from a datum denoted by $e'$.}  
\end{quote}
$\noreach$ is frequently used for representing freshness of new
references. 

Note that expressions of our logical language do not include arbitrary
programs. If we enlarge terms in the present logical language to
encompass arbitrary programs, then terms in the logic will have
effects when being evaluated (such as $\lambda y.x:=3$). In addition,
the axiomatisation of equality would feature involved axioms like $()
= (x:=3)$.  Note also that the inclusion of application leads to
expressions whose evaluation may be non-terminating.  Excluding such
arbitrary terms means that we can use standard first-order logic with
equality and its usual axiomatisation as its basis, avoiding
non-termination and side-effects when calculating assertions.

Terms are typed inductively starting from types for variables and
constants and signatures for operators.  The typing rules for terms 
follow the standard ones for programs \cite{PierceBC:typsysfpl}
and are given in Figure~\ref{fig:typingrules:formulae} in Appendix \ref{app:typing}. 
We write $\Gamma \vdash e : \alpha$ when $e$ has type $\alpha$
such that free variables in $e$ have types following $\Gamma$; and
$\Gamma \vdash C$ when all terms in $C$ are well-typed under
~$\Gamma$. 

Equations between terms of different types
will always evaluate to $\falsity$.\footnote{To be precise,
``terms of unmatchable types'':
this is because of the presence of type variables. For example, 
the equation ``$x^X = 1^{\NAT}$'' can hold depending on models but 
``$x^{\refType{X}} = 1^{\NAT}$'' never holds.} 
The falsity $\falsity$ is definable as $1 \not= 1$, and 
its dual $\truth \DEFEQ \neg \falsity$. 
The \emph{syntactic substitution} $C\SUBST{e}{!x}$ is also used
frequently: the definition is standard, save for some subtlety
regarding substitution into the post-condition of evaluation
formulae, details can be found in Appendix B in \cite{ALIAS}.
\emph{Henceforth we only treat well-typed terms and
formulae.}

Further notational conventions follow.

\begin{notation}[Assertions]\rm  \label{con:assertions} \ 
\begin{enumerate}[(1)]
\item
In the subsequent technical development, logical connectives are used
with their standard precedence/association, with content
quantification given the same precedence as standard quantification
(i.e.~they associate stronger than binary connectives).  For example,
\[
\neg \,A\,\AND\, B\, \ENTAILS\, \forall x.C \,\OR\, \someCon{e}{D} \,\ENTAILS\, E
\]
is a shorthand for
$(\,(\neg A)\,\AND\, B)\, \ENTAILS\, (\,(\,(\forall x.C) \,\OR\,
(\someCon{e}{D})\,) \,\ENTAILS\, E)$.
The standard binding
convention is always assumed.
\item 
$C_1 \LITEQ C_2$ stands for $(C_1\ENTAILS C_2)\AND
(C_2\ENTAILS C_1)$, stating the logical equivalence of $C_1$ and
$C_2$.  
\item $e\neq e'$ stands for $\neg e=e'$. 

\item
Logical connectives are used not only syntactically but also
semantically, i.e. when discussing meta-logical and other notions of
validity.

\item We write
\ $\EVALASSERT{C}{e_1}{e_2}{z}{C'}$ \ 
for 
\ $C \ENTAILS \ONEEVAL{e_1}{e_2}{z}{C'}$. 

\item
$\ONEEVAL{e_1}{e_2}{e'}{C}$
stands for $\ONEEVAL{e_1}{e_2}{x}{x=e'\AND C}$ where $x$ is fresh 
and $e'$ is not a variable; 
$\SIMPLEONEEVAL{e_1}{e_2}{C}$ stands for
$\ONEEVAL{e_1}{e_2}{()}{C}$; and 
$\APP{e_1}{e_2}\Downarrow$ stands for
the convergence $\ONEEVAL{e_1}{e_2}{x}{\truth}$.  
We apply the same abbreviations to 
$\EVALASSERT{C}{e_1}{e_2}{z}{C'}$.

\item 
For convenience of rule presentation we will use projections $\pi_i(e)$ as
a derived term. They are redundant in that any formula containing 
projections can be translated into one without: for example 
$\pi_1(e)  = e'$ can be expressed as $\exists y.e = \PAIR{e'}{y}$.

\item We denote $\FV{C}$\ (resp.~$\FL{C}$)\ for the set of 
the free variables (resp. free locations) in $C$.  

\item 
$\allCon{x_1..x_n}{C}$  
for $[!x_1]..[!x_n]C$.
Similarly for $\someCon{x_1..x_n}{C}$.   

\item 
We write $\NOTREACH{\VEC{e}}{e}$ for 
$\AND_i \NOTREACH{e_i}{e}$;\ \
 $\NOTREACH{e}{\VEC{e}}$ for 
$\AND_i \NOTREACH{e}{e_i}$; \ \ 
and  $\NOTREACH{\VEC{e}}{\VEC{e}'}$ for $\AND_{ij} \NOTREACH{e_i}{e'_j}$.
\end{enumerate}
\end{notation}



\subsection{Assertions for Local State}
\label{sub:ex:local}
\NI
We explain assertions with examples.
\begin{enumerate}[(1)]

\item The assertion $x = 6$ says that $x$ of type $\NAT$ is equal to 6.

\item  Assuming $x$ has type $\REF{\NAT}$, $!x = 2$ means $x$ stores 2.
Next assume that $e_1$ and $e_2$ have a reference type carrying 
a functional type, say  
$\REF{\NAT\to \NAT}$. 
Then we can specify equality of the contents of the reference as: 
$!e_1 = !e_2$. Note that neither $e_1$ nor $e_2$ contains
$\lambda$-expressions. 
Section \ref{sub:axiom:eq} shall show that  
the standard axioms for the equality hold in our logic. 

\item 
Consider a simple command 
$x:=y;y:=z;w:=1$. After its run, we can 
reach 
reference name 
$z$ 
by dereferencing $y$, and $y$ by dereferencing $x$. 
Hence $z$ is reachable from $y$, $y$  from $x$, 
hence 
$z$ from $x$. So the final state satisfies
$\REACH{x}{y} \AND \REACH{y}{z} \AND \REACH{x}{z}$ 
which implies 
by transitivity. 

\item 
Next, assuming 
$w$ is newly generated, we may wish to say 
$w$ is {\em unreachable} from $x$, to ensure  freshness 
of $w$.  For this we assert $\NOTREACH{w}{x}$,
which, as noted, stands for $\neg (x \reachable w)$. 
$\NOTREACH{x}{y}$ always implies $x\not =y$.  
Note that $\REACH{x}{x} \LITEQ \REACH{x}{!x} \LITEQ \truth$ and 
$\NOTREACH{x}{x} \LITEQ 
\falsity$. But $\REACH{!x}{x}$ may or may not hold
(since there may be a cycle between $x$'s content and $x$ in the presence of recursive types).


\item 
\label{ex:funcreachable}
We consider reachability in
procedures.
Assume $\lambda ().(x:=1)$ is named as $f_w$, similarly
$\lambda ().!x$ as $f_r$.
Since $f_w$ can write to $x$, we have $\REACH{f_w}{x}$. 
Similarly $\REACH{f_r}{x}$. 
Next suppose $\mathtt{let} \ x=
\REFPROG{z} \ \mathtt{in} \ \lambda ().x$ has name
$f_c$ and $z$'s type is 
${\refType{\NAT}}$. Then $\REACH{f_c}{z}$ (e.g. consider $!(f_c()):=1$). 
However $x$ is {\em not} reachable from 
$\lambda ().((\lambda y.())(\lambda ().x))$
since semantically, this function never touches $x$.

\item $\allworlds !x=1$ says that $x$'s content 
is unchanged from $1$ forever, which is logically equivalent 
to $\falsity$ (since $x$ might be updated in the future).  
Instead $\someworld !x=1 \LITEQ \truth$. 
On the other hand, 
$\allworlds x=1 \LITEQ \someworld x=1 \LITEQ x=1$ (since 
a value of a functional variable is not affected by the state).   

\item The following program:
\begin{equation}\label{func:inc} 
f\;\DEFEQ\;  \lambda ().(x:=!x+1;!x)
\end{equation} 
satisfies the following assertion, when named $u$:
\[ 
\allworlds\forall  i^\NAT.\,  \eval{!x=i}{u}{()}{z}{!x=z \AND !x=i+1}
\]
saying: 
\begin{quote}
{\em now or for any future state, invoking the function named $u$
increments the content of $x$ and returns that content.} 
\end{quote}
Stating it for a future state is important since 
a closure is potentially invoked many times in different states. 

\item
We often wish to say that
the write effects of an application are restricted to specific 
locations.
The following {\em located assertion} \cite{ALIAS}
is used for this purpose:
$
\ONEEVAL{e}{e'}{x}{C}@\VEC{e}
$
where each $e_i$ is of  reference type and does not contain a dereference. 
$\VEC{e}$ is called {\em effect set}, which \martinb{might} be modified by the
evaluation. 
As an example:
\begin{equation}\label{INC}
\INCSPEC{u, x} \DEFEQ 
\allworlds \forall  i. \eval{!x=i}{u}{()}{z}{z=!x \ANDl !x=i+1}@{x}
\end{equation}
is satisfied by $f$ in (\ref{func:inc}), saying that a function named
$u$, when invoked, will: (1) increment the content of $x$ and (2)
return the original content of $x$, without \martinb{modifying (in an
observational fashion)} any state except $x$.  
As in \cite{ALIAS},  located assertions can be translated
into  non-located evaluation formulae together with content
quantification in \S~\ref{sub:assertion}, see Proposition
\ref{pro:located_decompose}.


\item Assuming $f$ denotes
the result of evaluating $\IncPrg$ in the introduction, 
we can assert, using the existential hiding quantifier and naming by $u$:
\begin{equation}\label{INCNew}
  \HIDEe x.(!x=0\ANDl \INCSPEC{u, x})
\end{equation}
which says: there is a hidden reference $x$
storing $0$ such that, whenever $u$ is invoked, it \martinb{writes at}  $x$
and returns the increment of the value stored in $x$ at the time of invocation.
\item \label{combine}
We illustrate that combining
 hiding quantifiers and the non-reachability predicate is
necessary for describing the effects and use of newly generated
references. 
Consider:
\begin{equation}
\label{aaabbbccc}  \mathtt{let}\ x=\refPrg{2}\ \mathtt{in}\ y:=x 
\end{equation}
The location denoted by the bound variable $x$ is, at the time when
the new reference is generated, hidden and disjoint from any existing
datum. The location represented by $x$ is still hidden but it has now
become accessible from a variable $y$, and this location is still
unreachable from other references. \martinb{Thus hiding and disjointness are
separate concerns, and,  assuming $z$ to be a reference disjoint from $y$,
the post-state of (\ref{aaabbbccc}) can be
described as:}
\begin{equation}
  \HIDEe x.(!y=x\ \AND\ !x=2\ \AND\ z\noreach x)
\end{equation}

\item \label{refagent}
The function 
$f_1 \;\DEFEQ\; \lambda n^{\NAT}. \REFPROG{n}$, named $u$, meets the following
specification. Let $i$ and $\TVX$ be fresh.
\begin{equation}\label{SEEABOVEa}
\mathsf{fresh}  
\quad
\DEFEQ  
\quad
\allworlds \forall n^{\NAT}.\,
\forall \TVX.\forall i^\TVXscript.
\ONEEVAL{u}{n}{z}
{\HIDEe x.(!z=n\ANDl \NOTREACH{z}{i} \AND z=x)}
@{\emptyset}.
\end{equation}
The above assertion says that $u$, when applied to $n$, 
will always return a hidden fresh reference $z$
whose content is $n$ 
and which is unreachable from any datum 
existing at the time of the invocation;
and in the execution it will leave no writing effects
to the existing state.
Since $i$ ranges over
arbitrary data, 
unreachability of $x$ from 
each such $i$ 
in the post-condition 
indicates that $x$ is freshly generated
and is not stored in any existing reference. 


\item Now let us consider the following three formulae:
\begin{eqnarray}
\mathsf{fresh}_1 & \DEFEQ & 
\forall n^{\NAT}.\,
\forall \TVX.\forall i^\TVXscript.
\ONEEVAL{u}{n}{z}
{\HIDEe x.(!z=n\ANDl \NOTREACH{z}{i} \AND z=x)}
 @{\emptyset}
 \label{SEEABOVEb}\\
\mathsf{fresh}_2 & \DEFEQ & 
\forall n^{\NAT}.\,
\forall \TVX.\forall i^\TVXscript.
\allworlds 
\ONEEVAL{u}{n}{z}
{\HIDEe x.(!z=n\ANDl \NOTREACH{z}{i} \AND z=x)}
@{\emptyset}
 \label{SEEABOVEc}\\
 \mathsf{fresh}_3 & \DEFEQ & 
 \allworlds 
 \forall n^{\NAT}.\,
 \forall \TVX.\forall i^\TVXscript.
 \allworlds 
 \ONEEVAL{u}{n}{z}
 {\HIDEe x.(!z=n\ANDl \NOTREACH{z}{i} \AND z=x)}
 @{\emptyset}
 \label{SEEABOVEd}
 \end{eqnarray}


Each formula is read as follows:
\begin{enumerate}[$\bullet$]
\item $\mathsf{fresh}_1$ means that the procedure named by $u$, when invoked in
  the present state with number $n$, will create a cell with that content
   which is fresh {\em in the
    current state}.

\item $\mathsf{fresh}_2$ means that the procedure $u$, when invoked 
with number $n$
in the present or any future state, will create a cell with content $n$ 
which is fresh
  {\em in the current state}.  For example the following program
  satisfies this assertion (naming it as $u$):
  \begin{equation}\label{AAAAAAANEW}
    f_2 \;\DEFEQ\; \mathtt{let} \ x =\refPrg{0} \ \mathtt{in} \
    \lambda y^{\NAT}.(x:=y;\,x)
  \end{equation}

\mfb{The function returned by (\ref{AAAAAAANEW}) does return a fresh reference upon initial invocation: but from}
  the next time \mfb{this function} returns the same reference cell albeit with the new
  value specified. So it will be fresh with respect to the current
  state (for which we are asserting this formula) but \emph{\martinb{not necessarily}} with respect
  to each initial state of invocation.

\item $\mathsf{fresh}_3$ means that
  if we invoke the procedure $u$ in \mfb{the current}
  state or in any further future state, it will create a cell which is
  fresh in that state.
\end{enumerate}
Then we have: 
\begin{eqnarray}\label{freshentail}
\mathsf{fresh}\LITEQ \mathsf{fresh}_3
\ENTAILS \mathsf{fresh}_2
\ENTAILS \mathsf{fresh}_1
\end{eqnarray}
which we shall prove by the axioms for $\allworlds$ later.  The program
(\ref{AAAAAAANEW}) satisfies $\mathsf{fresh}_1$ and
$\mathsf{fresh}_2$, but does {\em not} satisfy $\mathsf{fresh}$ (nor
$\mathsf{fresh}_3$) since $f_2$ returns the same location.  On the
other hand, $f_1$ satisfies all of $\mathsf{fresh}$,
$\mathsf{fresh}_1$, $\mathsf{fresh}_2$ and $\mathsf{fresh}_3$.  This
example demonstrates that a combination of $\allworlds$ and a
decomposed evaluation formula gives precise specifications in the
presence of the local state.\footnote{Note that in $\mathsf{fresh}$
  and $\mathsf{fresh}_3$, it is essential that we put universal
  quantifications $\forall \TVX$ and $\forall i^\TVXscript$ {\em after}
  $\allworlds$. This has not been possible
  in the two-sided evaluation formulae used in the
  logics for pure and imperative higher-order functions without local
  state in \cite{SHORT1,HY04PPDP,GLOBAL,ALIAS}. See 
(\ref{encoding:A}).} 
\end{enumerate}

\section{Models and Semantics}
\label{sec:models}
\subsection{Models}
\label{model:sound}
\NI 
We introduce the semantics of the logic based on the operational
semantics of programs, using partially hidden stores.
%
Our purpose is to have a \emph{precise} and \emph{clear} correspondence between
programs' operational behaviour (and the induced observational
semantics) and the semantics of assertions. 
\martinb{This is the reason for defining our models operationally. This approach offers}
a simple framework to reason about the semantic effects of hidden
(and/or newly generated) stores on higher-order imperative programs
(for further discussions, see Remark \ref{rem:models:abs} later).  For
capturing local state, our models incorporate hidden locations using
$\new$-binders, suggested by the $\pi$-calculus \cite{MilnerR:calmp1}.
For example,  consider the program $\IncPrg$ from the
introduction.
\begin{equation}
        \IncPrg \DEFEQ \mathtt{let} \ x =\refPrg{0} \ \mathtt{in} \
\lambda ().(x:=!x+1;\,!x)
\end{equation}
Recall that after running $\IncPrg$, we reach a
state where a hidden name stores $0$, to be used by the resulting
procedure when invoked. 
Hence, $\IncPrg$ named $u$, is modelled as:  
\begin{equation} \label{ex:model:local:app}
(\new l)(\ASET{u:\lambda ().(l:=!l+1;\,!l)}, \ \ASET{l\mapsto 0}) 
\end{equation}
which says that the appropriate behaviour 
is at $u$, in addition to a hidden reference  $l$
storing  $0$.


\begin{definition}\rm {\rm (models)} \label{def:models:concrete}
An \emph{open model of type $\Gamma$} 
is a tuple $(\xi, \sigma)$ where:
\begin{enumerate}[$\bullet$]
  
\item $\xi$, called \emph{environment}, is a finite map 
from variables in $\dom{\Gamma}$ to closed values such that, 
for each $x\in\dom{\Gamma}$, 
$\xi(x)$ is typed as $\Gamma(x)$ under $\Gamma$, 
i.e.~$\Gamma\proves \xi(x):\Gamma(x)$. 
        
\item $\sigma$, called \emph{store}, is a finite map from 
labels in $\{l \ | \ l \in \dom{\Gamma} \ \}$ 
to closed values such that 
for each $l\in \dom{\sigma}$, 
$\Gamma(l)$ has type $\refType{\alpha}$, 
then $\sigma(l)$ has type $\alpha$ under $\Gamma$,
i.e.~$\Gamma\proves \sigma(l):\alpha$. 
\end{enumerate} 
When $\Gamma$ includes free type variables,
$\xi$  maps them to closed types, 
with the obvious corresponding typing constraints.
A {\em model} of type $\Gamma$ 
is a structure $(\new \VEC{l})(\xi, \sigma)$ with 
$(\xi, \sigma)$ being an open model of type 
$\Gamma,\Delta$ with $\ASET{\VEC{l}}= \dom{\Delta}$.
$(\new \VEC{l})$ acts as binders. 
$\MMM, \MMM', \ldots$ range over models. 
\end{definition}
\NI 
An open model maps variables and locations to closed values: a model
then specifies part of the locations as ``hidden''. 
For example, 
$(\new l)(x:l\cdot y:l', [l \mapsto 3]\cdot [l' \mapsto 3])$ 
is a model with a typing environment: 
$\Gamma = \{ x:\reftype{\NAT},y:\reftype{\NAT},l':\reftype{\NAT}\}$. 
We often omit $\Gamma$ and a mapping from type variables 
to closed types from $\MMM$. 

Since assertions in the
present logic are intended to capture observable program behaviour, 
the semantics of the logic uses models quotiented by an
observationally sound equivalence, 
which we choose to be the 
standard contextual congruence itself. 


\begin{definition}\rm \label{def:modelequivalence}
Assume $\MMM_i\DEFEQ (\new \VEC{l}_i)(\VEC{x}:\VEC{V}_i, \sigma_i)$
typable under $\Gamma$.
Then we write
$\MMM_1\WB\MMM_2$ if the following clause holds for each
typed context $C[\ \CD\ ]$ which is typable under $\Gamma$ and
in which no labels from $\VEC{l}_{1,2}$ occur:
\begin{equation}
(\new \VEC{l}_1)(C[\ENCan{\VEC{V}_1}], \sigma_1)\converges
\quad\text{iff}\quad
(\new \VEC{l}_2)(C[\ENCan{\VEC{V}_2}], \sigma_2)\converges
\end{equation}
where $\ENCan{\VEC{V}}$ is the $n$-fold pairings of a vector of values.
\end{definition}

\NI Definition \ref{def:modelequivalence} in effect takes
models up to the standard contextual congruence. We could have used a
different program equivalence (for example call-by-value $\beta\eta$
convertibility), as far as it is observationally sound. 
Note that we have
\begin{equation}
(\new \VEC{l})(\xi\CD\AT{x}{V_1}, \sigma\CD l\mapsto W_1)
\ {\WB}\ 
(\new \VEC{l})(\xi\CD\AT{x}{V_2}, \sigma\CD l\mapsto W_2)
\end{equation}
whenever $V_1 \cong V_2$ and $W_1\cong W_2$, where $\cong$  is the
contextual congruence on programs defined in \S~\ref{obscon:models}. 

To see the reason why we take the models up to observational congruence, 
let us consider the following program: 
\begin{equation}
        \IncPrgTwo
\                \DEFEQ \ 
\mathtt{let} \ x =\refPrg{0}, \ y =\refPrg{0} \ \mathtt{in}\ 
        \lambda ().(x:=!x+1;y:=!y+1;\,(!x+!y)/2)
\end{equation}
which is contextually equivalent to $\IncPrg$. 
Then we have the following model for $\IncPrgTwo$. 
\begin{eqnarray}
\label{ex:model:variant:app}
         (\new ll')
         (
                 \{
                         u:\lambda ().(x:=!x+1;y:=!y+1;\,(!x+!y)/2),
             \; x:l,\;\
             y:l'
            \},\ 
 \{             l\mapsto 0,\; l'\mapsto 0
 \})
 \end{eqnarray}
Since the two programs 
originate in the same abstract behaviour, we wish to identify the model in 
(\ref{ex:model:local:app}) and the above model, 
taking them up to the equivalence. 

\begin{remarks}\rm \label{rem:models:abs} (presentation of models)\ 
  The model as given above can be presented algebraically using the
  language of categories \cite{stark:namhof}. One method, which can
  treat hiding as above categorically, uses a class of toposes which
  treat renaming through symmetries~\cite{HondaK:elestr}. We can
  also use the ``swapping''-based treatment of binding based on
  \cite{PittsAM:newaas}.  Note however that the use of such different
  presentations (with respective merits) does {\em not} alter the
  equational and other properties of models and the satisfaction
  relation, as far as we wish to use the standard observational
  semantics (Morris-like contextual congruence) or the equivalent
  models (so-called fully abstract models) as a basis of our
  logic. Another significant point is that the game-based model in
  \cite{AHM98} is the only known model satisfying this (full
  abstraction) criteria, whose morphisms are isomorphic to a class of
  typed $\pi$-calculus processes \cite{Honda02}. The presented
  ``operational'' model is hinted at by, and is close to, the
  $\pi$-calculus presentation of semantics of the target language.
  The present approach allows us to have models which are
  automatically faithful to the standard observational semantics of
  the language, directly capturing the effects of hidden stores by semantics of the logic.
  Other models may as well be used for exploring various 
  aspects of the presented logic.
\end{remarks}

\begin{remarks} (hidden locations)\label{rem:hidden} \ \martinb{Following  standard
  textbooks \cite{GunterCA:semprol,PierceBC:typsysfpl}, we treat
  locations as values (which is natural from the viewpoint of
  reduction). A significant point is that distinctions among
  these values (locations) matter even if they are hidden.} For
  example if we have:
  \begin{equation}
    M\; \DEFEQ\; \PAIR{\refPrg{2}}{\refPrg{2}} 
  \end{equation}
  and evaluate $M$, \mfb{we get a pair of two fresh locations both
  storing $2$. For the denotation of this
  resulting value, it is essential that these two} references are
  distinct. For example the program:
  \begin{equation}
      N \;\DEFEQ\; \ \mathtt{let}\ x=\refPrg{2}\ \mathtt{in}\ \PAIR{x}{x}
  \end{equation}
  has a different observable behaviour, as justified by a context
  $C[\ ]\ \DEFEQ\ \mathtt{if}\ \pi_1[\ ]=\pi_2[\ ]\ \mathtt{then}\ 1 \
  \text{else} \ 2$. Thus distinctions matter,
  even if locations are hidden.
\end{remarks}

\subsection{Semantics of Equality}
\label{sub:sem:equality}
For the rest of this section, we give  semantics to assertions, 
mainly focussing on  key features  concerning local
state and which therefore differ from the previous logics
\cite{ALIAS}.
We start with the semantics of 
equality.


A key example are the programs $\IncShared$ in (\ref{ex:IncPrgshared}) and 
$\IncUnshared$ in (\ref{ex:IncPrgshared2}) from the introduction. 
After the second assignment of (\ref{ex:IncPrgshared}) 
and (\ref{ex:IncPrgshared2}), we consider whether 
we can assert ``$!a \ = \ !b$'' (i.e.~the content of $a$ and $b$ are
equal). For this inquiry, let us first recall the following defining clause for the
satisfaction of equality of two logical terms from
\cite{ALIAS} which follows the standard 
definition of logical equality. First we set, with
$\EXPRESSIONTYPES{\Gamma}{e}{\alpha}$,
$\MODELTYPES{\Gamma}{\MMM}$ and  an open model
$\MMM = (\xi, \sigma)$,
an interpretation of $e$ under $\MMM$ as 
follows.\footnote{Since a model in \cite{ALIAS} does not have
local state, it suffices to consider open models.}
\[
\begin{array}{c}
	\semb{x}_{\xi, \sigma} = \xi(x) 
		\quad 
	\semb{!e}_{\xi, \sigma} = \sigma(\semb{e}_{\xi, \sigma}) 
 		\quad 
	\semb{\PROGRAM{c}}_{\xi, \sigma} = \LOGIC{c} 
		\quad
	\semb{\PROGRAM{op}(\VEC{e})}_{\xi, \sigma}  = \LOGIC{op}(\semb{\VEC{e}}_{\xi, \sigma})
		\\[2mm] 
	\semb{\PAIR{e}{e'}}_{\xi, \sigma} = \PAIR{\semb{e}_{\xi, \sigma}}{\semb{e'}_{\xi, \sigma}}
		\quad 
	\semb{\INJ{i}{e}}_{\xi, \sigma} = \LOGICINJ{}{i}{\semb{e}_{\xi, \sigma}}
\end{array}
\]
which are all standard. Then we define:
\begin{eqnarray}\label{model:eq:simple}
\text{(the definition from \cite{ALIAS})}\qquad
\MMM\models e_1=e_2
\quad
\IFFDEF
\quad
\MAP{e_1}_{\MMM}\WB \MAP{e_2}_{\MMM}
\end{eqnarray}
Note  that (\ref{model:eq:simple}) 
says that $e_1=e_2$ is true under an open model $\MMM$ iff their
interpretations in $\MMM$ are congruent. Now suppose we apply
(\ref{model:eq:simple}) to the question of $!a \ = \ !b$ in
$\IncUnshared$.  Since the two instances of $\IncPrg$ stored in $a$
and $b$ have the identical denotation (or identical behaviour: because
they are exactly the same programs), the equality $!a \ = \ !b$ holds
for $\IncUnshared$ if we use (\ref{model:eq:simple}).  {\em However
this interpretation is wrong:} we observe that, in $\IncUnshared$,
running $!a$ twice and running $!a$ and $!b$ consecutively lead to
different observable behaviours, due to their distinct local states
(which can be easily represented using evaluation formulae).  Hence we
must have $!a \ \not= \ !b$, which says the standard definition
(\ref{model:eq:simple}) is not applicable in the presence of the local
state. On the other hand, running $!a$ and running $!b$ have always
identical observable effects: that is we can always replace the
content of $a$ with the content of $b$ in $\IncShared$, hence the
equality $!a \ = \ !b$ should hold for $\IncShared$.

The reason that the standard equality does not hold is because two
currently identical stateful procedures will in future demonstrate
distinct behaviour.
On the other hand, two identical functions which share
the same local state always show the same behaviour hence in
$\IncShared$ we obtain equality.  

This analysis indicates that we
need to consider programs placed in contexts to compare them precisely, 
leading to 
the following extension for the semantics for the equality,
assuming $\MMM\DEFEQ(\nu \VEC{l})(\xi, \sigma)$: 
\begin{equation} 
\MMM \models e_1 = e_2 \quad \IFFDEF\quad \MMM[u:e_1] \WB \MMM[u:e_2]
\end{equation}
where $\expands{\MMM}{u}{e}$ denotes $(\nu \VEC{l})(\xi\CD
u:\MAP{e}_{\xi, \sigma}, \sigma)$ with $u$ fresh and the variables and
labels in $e$ should be free in $\MMM$. Note that $\MMM[u:e]$ offers
the notion of a ``program-in-context'' when $e$ denotes a program. 
For example let us consider a model for the state immediately
after the assignment $b:=!a$ in $\IncShared$. Then the model
may be written as (taking $a$ and $b$ to be locations):
\begin{equation}
\MMM_{\IncShared} 
        \ = \  (\nu l)\left( 
        \emptyset,
        \quad
        \begin{array}{l}
                a \mapsto \lambda ().(l := !l + 1; !l),
                        \\
                b \mapsto \lambda ().(l := !l + 1; !l), 
                         \\ 
                l \mapsto n
        \end{array}
        \right)
\end{equation}
We obtain (writing the map for $a, b, l$ above as $\sigma$ for brevity):
\begin{equation}
\MMM_{\IncShared} [u:!a]
        \ = \  (\nu l)\left( 
                u: \lambda ().(l := !l + 1; !l),\quad\sigma
        \right)
\end{equation}
Notice that the function assigned to $u$ shares $l$ in the environment:
we are interpreting the dereference $!a$ ``in context''.
Similarly we obtain:
\begin{equation}
\MMM_{\IncShared} [u:!b]
        \ = \  (\nu l)\left( 
                u: \lambda ().(l := !l + 1; !l),\quad\sigma
        \right)
\end{equation}
By which we conclude $\MMM_{{\IncShared}} \models !a=!b$: if the results
of interpreting two terms in context are equal then we know their
effects to the model are equal. We leave it to the reader to check
the inequality between $!a$ and $!b$ for the corresponding
model representing ${\IncUnshared}$. 

The definition of equality above satisfies the standard axioms of 
equality as we shall see  in \S~\ref{sec:axioms}. It is also
accompanied by a notion of \emph{symmetry} which can be used for
checking (in)equality, introduced below.

\begin{definition}[permutation] \label{def:permutation}\rm 
Let $\MMM
\DEFEQ (\new \VEC{l})(\xi\CD\ATb{v}{V}\CD\ATb{w}{W},\;\sigma)$
where $\MMM$ is typed under $\Gamma$ and \mfb{$v, w$ have} the same 
type under $\Gamma$.  
Then, we set:
\begin{equation}
\MMM\piprm{vw}{wv}\; \DEFEQ\;(\new \VEC{l})(\xi\CD\ATb{v}{W}\CD\ATb{w}{V},\;\sigma)
\end{equation}
called a {\em permutation of $\MMM$ \mfb{at $v$ and $w$}}.
We extend the notion to an arbitrary bijection $\rho$ on $\dom{\Gamma}$,
writing $\MMM[\rho]$. A permutation $\rho$ on $\MMM$ is 
a {\em symmetry on $\MMM$} when
$\MMM[\rho]\congmodel \MMM$. 
\end{definition}

\begin{proposition}[symmetries]\label{prop:symmetry}\ 
\begin{enumerate}[\em(1)]
\item \label{prop:symmetry:per:modula}
Given $\MMM_{1,2}$ and a bijection $\rho$ on free variables in 
the domain of $\MMM_{1,2}$, 
we have $\MMM_1\WB\MMM_2$ iff $\MMM_1[\rho]\WB \MMM_2[\rho]$.

\item 
\label{prop:symmetry:syn:modula}
If $\MMM_1\WB\MMM_2$ and $\rho$ is symmetry of $\MMM_1$, 
then $\rho$ is symmetry of $\MMM_2$. 

\end{enumerate}
\end{proposition}
\begin{proof}
Obvious by definition.  
\end{proof}

\NI We illustrate how we can use the result above to model the
subtlety of equality of behaviours with shared local state. Let us
consider the following models $\MMM_1$ and $\MMM_2$, which represent
the situations analogous to $\IncShared$ and $\IncUnshared$ (again
after running the second assignment).  The defining clause for
equality \mfb{gives},
using $\expands{\MMM_1}{u}{v}\WB\expands{\MMM_1}{u}{w}$:
\begin{equation}
\MMM_1 \ = \  (\nu l)\left( 
        \begin{array}{l}
                v : \lambda ().(l := !l + 1; !l),
                        \\
                w : \lambda ().(l := !l + 1; !l),
        \end{array}
        \quad
                l \mapsto 0
        \right)
                \models
        v = w
\end{equation}
On the other hand, we have: 
\begin{equation}
\MMM_2 \ = \   
        (\nu ll')\left( 
        \begin{array}{l}
                v : \lambda ().(l := !l + 1; !l),
                        \\
                w : \lambda ().(l' := !l' + 1; !l'),

        \end{array}
        \quad
        \begin{array}{l}
                l \mapsto 0, \\
                l' \mapsto 0
        \end{array}
        \right)
                \models
        v \neq w
\end{equation}
This is because $\piprm{uv}{vu}$ is a symmetry of $\expands{\MMM_2}{u}{v}$,  
but {\em not} of $\expands{\MMM_2}{u}{w}$. The latter can be examined
by comparing the following two models (writing ``$u, w:V$'' to denote
``$u:V, w:V$''):
\begin{eqnarray}
\expands{\MMM_2}{u}{w}
& \ = \  &
        (\nu ll')\left( 
        \begin{array}{l}
                v : \lambda ().(l := !l + 1; !l),
                        \\
                u, w : \lambda ().(l' := !l' + 1; !l'),
        \end{array}
        \quad
        \begin{array}{l}
                l \mapsto 0, \\
                l' \mapsto 0
        \end{array}
        \right)\\
(\expands{\MMM_2}{u}{w})\piprm{uv}{vu}
& \ = \  &
        (\nu ll')\left( 
        \begin{array}{l}
                u : \lambda ().(l := !l + 1; !l),
                        \\
                v, w : \lambda ().(l' := !l' + 1; !l'),
        \end{array}
        \quad
        \begin{array}{l}
                l \mapsto 0, \\
                l' \mapsto 0
        \end{array}
        \right)
\end{eqnarray}
which  differ semantically when e.g.~$v$ and $w$ are 
invoked consecutively.
Hence 
by Proposition \ref{prop:symmetry} (\ref{prop:symmetry:syn:modula}), 
$\expands{\MMM_2}{u}{v}\not\WB\expands{\MMM_2}{u}{w}$, justifying 
the above inequality $v\not = w$. 
The permutations 
also help to prove the axioms of equality in \S~\ref{sec:axioms}. 

\subsection{Semantics of Necessity and Possibility Operators}
\label{subsec:necessity}
We define, with $u$ fresh,  
\[ 
\expands{\MMM}{u}{N}\converges \MMM'
\quad \mbox{when}\
(N\xi, \sigma)\converges (\new \VEC{l}')(V, \sigma')
\ \mbox{with}\
\MMM=(\new\VEC{l})(\xi,\ \sigma) 
\ \mbox{and}\ 
\MMM'=(\new\VEC{l}\VEC{l}')(\xi\CD u\!:\!V,\ \sigma')
\] 
where we always assume $u$ is fresh
and the variables and labels in $N$ are free in $\MMM$. 
The above definition intuitively means that 
$\MMM$ {\em can} reduce to $\MMM'$ through arbitrary effects on $\MMM$ by an
external program: in other words, $\MMM'$ is a hypothetical future
state (or ``possible world'') of $\MMM$. 
Then we generate $\MMM\EVOLVES\MMM'$ by 
\begin{enumerate}[(1)]
\item $\MMM\EVOLVES\MMM$ 
\item if $\MMM\EVOLVES\MMM_0$ and 
$\MMM_0[u:N]\converges \MMM'$, then  
$\MMM\EVOLVES\MMM'$
\end{enumerate} 
Thus $\MMM\EVOLVES \MMM'$ reads: 
\begin{quote}
{\em 
$\MMM$ may evolve to $\MMM'$ by interaction 
with zero or more typable programs}.
\end{quote}
Note that $\EVOLVES$ is reflexive and transitive. 
If $\MMM\EVOLVES\MMM'$
and $\MMM'$ adds the new domain $\ASET{x_1..x_n}$, then 
$x_1..x_n$ is its {\em increment} and we often explicitly write 
$\MMM\EVOLVESinc{x_1..x_n}\MMM'$.

The semantics of $\allworlds C$ says that for any target of evolution,
$C$ should hold:
\begin{equation}\label{def:allworlds}
  \MMM\models \allworlds C
  \quad\IFFDEF\quad
  \forall \MMM'.(\MMM\EVOLVES \MMM' \ENTAILS \MMM'\models C).
\end{equation} 
Dually we set:
\begin{equation}\label{def:someworld}
  \MMM \models \someworld C
  \quad\IFFDEF\quad
  \exists \MMM'.(\MMM\EVOLVES \MMM' \ANDl \MMM'\models C).
\end{equation} 

\subsection{Semantics of Evaluation Formulae}
The semantics of the evaluation formula is given below: 
\[ 
\MMM \models \ONEEVAL{e}{e'}{x}{C}
  \quad\IFFDEF\quad
\exists \MMM'.(\MMM[x:ee']\converges \MMM'\ANDl \MMM'\models C)
\]
which says that in the {\em current} state, if we apply $e$ to $e'$,
then the return value (named $x$) and the resulting state together
satisfy $C$. 

\martinb{We already motivated  the
 decomposition of  the original evaluation formulae 
\cite{ALIAS} into the simplified evaluation formulae and
the necessity operator from} \S~\ref{sub:ex:local}. 
Let us write the original evaluation formulae
in \cite{GLOBAL,ALIAS} as $\ORGEF{\eval{C}{e}{e'}{x}{C'}}$. Then
we can translate this in the present language as:
$$
\ORGEF{\eval{C}{e}{e'}{x}{C'}}
\;\LOGICEQ\;
\exists f,g.
(f=e\AND g=e'\AND \allworlds \eval{C}{f}{g}{x}{C'})
$$ 
that is, we interpret $e$ and $e'$ in the present state and name
them $f$ and $g$, and assert that, now or in any future state in which
$C$ is satisfied, if we apply $f$ to $g$, then it returns $x$ which,
together with the resulting state, satisfies $C'$.
The original clause says:
\begin{quote}
{\em In any initial hypothetical state which is reachable from the
  present state and which satisfies $C$, 
the application of $e$ to $e'$ terminates and both the result $x$
and the final state satisfy  $C'$.}
\end{quote}
To see the reason why we require $\allworlds$ 
\mfb{in the  specification of functions},  
we set:
\begin{equation}\label{HSnew}
\MMM \quad\DEFEQ\quad (\new l)(u:\lambda ().!l,\ w:\lambda ().l:=!l+1,\ \ l\mapsto 5)
\end{equation}
We can check that the set of all legitimate hypothetical states from this
state (i.e.~$\MMM'$ such that $\MMM[z:N]\converges \MMM'$) can be enumerated by:
\begin{equation}\label{HSnewdash}
  \MMM'/z \quad\DEFEQ\quad (\new l)(u:=\lambda ().!l,\ w:\lambda ().l:=!l+1,\ \ l\mapsto m)
\end{equation}
for each $m\geq 5$ 
(\martinb{since these are  essentially all the models reachable from $\MMM$}, as 
outside programs can create new references). 

Thus we have, for
$\MMM$ in (\ref{HSnew}):
\begin{equation}\label{HSnewres}
\MMM \models \allworlds \ONEEVAL{w}{()}{x}{x\geq 5}
\end{equation}
which says in any {\em future} state where $w$ is invoked, it always returns something no less
than $5$, which is operationally reasonable.

We can use this formula for specifying
the following program: 
\begin{equation}\label{HSnewconc}
\begin{array}{rcl}
L
&\quad\DEFEQ&\quad
\mathtt{let}\ x=\, \REFPRG{5}\ \mathtt{in}\\
&&
\qquad
\mathtt{let}\ u=\, \lambda ().!x\ \mathtt{in}\\ 
&&
\qquad\qquad
\mathtt{let}\ w=\, \lambda ().x:=!x+1\; \mathtt{in}\\
&&
\qquad\qquad\qquad
(fw)\ ;\ \IFTHENELSE{x\geq      5}{\mathtt{t}}{\mathtt{f}}
\end{array}
\end{equation}
When the application $fw$ takes place, some unknown computation occurs 
which may change the value of $x$: but as far as 
$fw$ terminates, it always returns $\true$. 
To reach (\ref{HSnewres}), we need to 
consider {\em all possible} $\MMM'$ with the effect from the outside. 
Since such $\MMM'$ satisfies (\ref{HSnewdash}), 
we can conclude the program $L$ always 
returns $\mathtt{t}$ (if $fw$ terminates).


\subsection{Semantics of Universal and Existential Quantification}
\label{subsec:sem:universal}
The universal and existential quantifiers also need 
to incorporate local state. 
We need one definition to identify a set of terms which 
do not change the state of any models. Below $\MMM^{\Gamma}$
indicates that $\MMM$  is typable under $\Gamma$.

\begin{definition}[Functional Terms]\label{def:functional}
\rm We define the set of {\em functional terms} of type $\Gamma$, denoted
${\mathcal{F}}^{\Gamma}$, or often
simply 
${\mathcal{F}}$ leaving its typing implicit, as:
\[ {\mathcal{F}}\ \DEFEQ \ \{ N \ | \ \forall \MMM^{\Gamma}. 
(\MMM[u:N]\Downarrow \MMM' \ENTAILS \MMM \cong \MMM'/u)\}
\]
where 
$\MMM/u \DEFEQ (\new \VEC{l})(\xi,\sigma)$ if 
$\MMM=(\new \VEC{l})(\xi\cdot\AT{u}{V},\sigma)$; 
and $\MMM/u\DEFEQ\MMM$ when $u\not\in \FV{\MMM}$. 
We write $L,L',...$ for functional terms, often leaving their types implicit.
\end{definition}

\NI Above $\MMM \cong \MMM'/u$ ensures that $L$ does not affect $\MMM$
during evaluation of $L$ in $\MMM$.  Note that values are always functional
terms.  In a context of reasoning for object-oriented languages, a
similar formulation (called strong purity) is used in \cite{Naumann05}
for justifying the semantics of method invocations whose evaluation has
no effect on the state of existing objects.

Now we define: 
\begin{eqnarray}
\MMM \models \forall x.C 
\quad \IFFDEF \ 
\forall L\in \mathcal{F}.(\expands{\MMM}{x}{L}\converges \MMM' \ENTAILS 
\MMM'\models C)  
\end{eqnarray}
Dually, we have: 
\begin{eqnarray}
\MMM \models \exists x.C 
\ \IFFDEF \ 
\exists L\in \mathcal{F}.(\expands{\MMM}{x}{L}\converges \MMM' \ANDl 
\MMM'\models C)  
\end{eqnarray}
If we restrict $L$ above to a value,
then the definition coincides with 
the original one in \cite{ALIAS}. We need to extend values to
functional terms so that 
a term can read information from 
hidden locations \martinb{(cf.~the semantics of 
equality $e_1 = e_2$). }
As a simple example, consider: 
\[ 
\MMM\DEFEQ (\new l_1,l_2)(y:l_1, \ l_1\mapsto l_2, l_2\mapsto 2)
\]
Under this model, we wish to say 
$\MMM \models \exists x.x = !y$. But if we only allow $x$ to range
over values, this standard tautology does not hold for $\MMM$.
Using the functional term $!y\in \mathcal{F}$, we can expand the entry $x$ 
with $!y$, and we have: 
\[ 
\MMM[x\, :\, !y] \Downarrow 
(\new l_1l_2)(x:l_1\cdot y:l_1, \ l_1\mapsto l_2, l_2\mapsto 2)
\DEFEQ \MMM' \ANDl \MMM'\models x=y 
\] 
\martinb{Thus using a functional term $L$ instead of a value $V$ 
for a quantified variable
is necessary
for reasons similar to those that required modifying the semantics
of equality.
Universal and existential quantifiers satisfy the standard 
axioms familiar from first-order logic,  some of which are studied later.}

\subsection{Semantics of Hiding}
\label{subsec:sem:hiding}
The universal hiding-quantifier has the following semantics.
\begin{equation} \label{ACACAC}
\MMM \models \HIDEa x.C
\quad
\IFFDEF
\quad
\forall \MMM'.
((\new l)\MMM'\WB \MMM\,\ENTAILS\,\ \MMM'[x:l]\models C)
\end{equation}
where $l$ is fresh, i.e. $l\not\in\FL{\MMM}$ where $\FL{\MMM}$ denotes
free labels in $\MMM$.
The notation $(\new l)\MMM'$ 
denotes addition of the hiding of $l$ to $\MMM'$, 
as well as indicating that $l$ occurs free in $\MMM'$. 
$\MMM[x:l]$ adds $x:l$ to the environment part of $\MMM$. 

Dually, with $l$ fresh again:
\begin{equation} \label{AEAEAE}
\MMM \models \HIDEe x.C
\quad
\IFFDEF
\quad
\exists \MMM'.
((\new l)\MMM'\WB \MMM\ANDl \MMM'[x:l]\models C)
\end{equation}
\martinb{which says that $x$ denotes a hidden reference, say $l$,
and the result of taking it off from $\MMM$ satisfies $C$.
}

As an example of satisfaction, let:
\begin{equation} \label{AEAEAE2}
\MMM\DEFEQ
(\new l)(\ASET{u:\lambda ().(l:=!l+1;!l)},\; \ASET{l\mapsto 0})
\end{equation}
then we have:
\begin{equation} 
\MMM
\models
 \HIDEe x.C 
\end{equation}
with 
\begin{equation} 
C \ \DEFEQ \ 
!x=0 \ANDl 
\allworlds \forall  i.\ASET{!x=i}\APP{u}{()}=z\ASET{z=!x\ANDl\, !x=i+1}
\end{equation}
using the definition in (\ref{AEAEAE}) above. To see this holds, let
\begin{equation}
\MMM'\DEFEQ (\ASET{u:\lambda ().(l:=!l+1;!l)}, \; \ASET{l\mapsto 0})
\end{equation} 
We have
$(\new l)\MMM'\LOGICEQ \MMM$ and $\MMM'[x:l]\models C$.
Here $\MMM$ represents a situation where 
$l$ is hidden and $u$ denotes a 
function which increments and returns the content of $l$;
whereas $\MMM'$ is the result of taking off this hiding,
exposing the originally local state, cf.~\cite{CairesL:spalogcI}.

Despite  $x$'s type being  a reference, 
$\forall x.C$ differs substantially  from 
$\HIDEforall x.C$. The former says that for any reference $x$, which can be
either (1) an existing free reference;
(2) an existing hidden reference reachable through dereferences; or
(3) a fresh reference with  arbitrary content, 
the model satisfies $C$. On the other hand, the latter
means that for any reference $x$ which is
hidden in the present model, $C$ should hold: in this case
$x$ cannot be a free reference name hence (1) is not included.
Similarly for their dual existential versions.


\subsection{Semantics of Content Quantification}
\label{subsec:sem:content}
Next we define the semantics of the content quantification. 
Let us write 
$\updates{\MMM}{x}{V}$ 
for $(\nu \VEC{l})(\xi, \sigma[l\mapsto V])$
with $\MMM = (\nu \VEC{l})(\xi, \sigma)$
and $\xi(x)=l$.
In \cite{ALIAS} (without local state), 
$\MMM \models \allCon{x}{C}$ is
defined as $\forall V. \updates{\MMM}{x}{V} \models C$
which means that for all content of $x$, $C$ holds. 
In the presence of the local state, we simply extend the use of values
to the use of functional terms in the sense of 
Definition \ref{def:functional} as follows: 
\begin{eqnarray}
  \MMM \models \allCon{e}{C} \quad \IFFDEF \quad 
\forall L\in {\mathcal{F}}.
\updates{\MMM}{e}{L} \models C
\end{eqnarray}
where 
we write
$\updates{\MMM}{e}{L}$ 
for 
$(\nu \VEC{l})(\xi, \sigma[l'\mapsto V])$,
assuming
$\MMM=(\nu \VEC{l})(\xi, \sigma)$,
$\MAP{e}_{\xi, \sigma}=l'$,
$(\nu \VEC{l})(L\xi , \sigma)\converges \MMM'$ and $\MMM'\WB (\nu \VEC{l})(V, \sigma)$.
Thus we consider an update 
through the assignment of an
external functional term $L$ to a location in $\MMM$ 
under local names.  
With this definition, all  the axioms and  invariant rules 
in \cite{ALIAS} stay unchanged. 

\subsection{Semantics of Reachability}
\label{subsec:sem:reachability}
We now define the semantics of reachability.
Let $\sigma$ be a store
and $S\subset\dom{\sigma}$. Then the {\em label
closure of $S$ in $\sigma$}, written $\ncl{S}{\sigma}$, is the minimum
set $S'$ of locations such that: (1) $S\subset S'$ and (2) If $l\in S'$ then
$\FL{\sigma(l)}\subset S'$.
The label closure satisfies the following natural properties. 

{\begin{lemma} \label{lem:ncl}
For all $\sigma$, we have: 
\begin{enumerate}[\em(1)]
\item $S\subset \ncl{S}{\sigma}$;  
$S_1\subset S_2$ implies 
      $\ncl{S_1}{\sigma} \subset \ncl{S_2}{\sigma}$;  and 
$\ncl{S}{\sigma}=\ncl{\ncl{S}{\sigma}}{\sigma}$
\item 
$\ncl{S_1}{\sigma} \cup \ncl{S_2}{\sigma} = \ncl{S_1\cup S_2}{\sigma}$

\item  
$S_1 \subset \ncl{S_2}{\sigma}$ and 
$S_2 \subset \ncl{S_3}{\sigma}$, 
then $S_1 \subset \ncl{S_3}{\sigma}$

\item 
there exists $\sigma'\subset \sigma$ such that 
$\ncl{S}{\sigma} = \FL{\sigma'}= \dom{\sigma'}$.  

\end{enumerate}
\end{lemma}}
\NI 
\begin{proof}
(1,2) are direct from the definition. (3)  follows immediately
from (1,2). For (4), take $\sigma'=\cup_{l\in
\ncl{S}{\sigma}}[l\mapsto \sigma(l)]$. Then 
obviously $\sigma'\subset \sigma$ and  
$\ncl{S}{\sigma} = \FL{\sigma'}= \dom{\sigma'}$.  
\end{proof}

For reachability, we define: 
\[
\begin{array}{c}
\text{
\begin{tabular}{ll}
$\MMM\models e_1\reachable e_2$ 
\ \ \ \ if &
$\MAP{e_2}_{\xi, \sigma}\in \ncl{\FL{\MAP{e_1}_{\xi,
      \sigma}}}{\sigma}$
\ for each $(\new \VEC{l})(\xi, \sigma)\WB \MMM$
\end{tabular}
}
\end{array}
\]
The clause says
the set of all reachable locations from $e_1$ includes $e_2$
modulo $\WB$.

For the programs  in \S~\ref{sub:ex:local} ({\ref{ex:funcreachable}}),
we can check
$f_w\reachable x$, $f_r\reachable x$ and  $f_c\reachable z$
hold under 
$
f_w:\lambda ().(x:=1), \
f_r: \lambda ().!x, \
f_c: \mathtt{let} \ x=
\REFPROG{z} \ \mathtt{in} \ \lambda ().x$ (regardless of the store part). 

The following characterisation of $\notreach$ is often useful for
justifying  fresh name axioms.  Below $\sigma = \sigma_1\uplus
\sigma_2$ indicates that $\sigma$ is the union of $\sigma_1$ and $\sigma_2$,
assuming $\dom{\sigma_1}\cap \dom{\sigma_2} = \emptyset$.

\begin{proposition}[partition] 
\label{pro:partition}
$\MMM \models x\notreach u$ if and only if for some $\VEC{l}$, $V$, $l$ and $\sigma_{1,2}$,
we have $\MMM\WB (\new \VEC{l})(\xi\CD u:V\CD x:l,\ \sigma_1\uplus \sigma_2)$ 
such that 
$\ncl{\FL{V}}{\sigma_1\uplus \sigma_2}=\FL{\sigma_1}=\dom{\sigma_1}$ and
$l\in\dom{\sigma_2}$.
\end{proposition}
\begin{proof}
For the only-if direction,  assume
$\MMM \models x\notreach u$.
By the definition of (un)reachability,
we can set  (up to $\WB$)
$\MMM\DEFEQ (\new \VEC{l}')(\xi\CD u:V\CD x:l,\ \sigma)$ 
such that
$l\not\in \ncl{\FL{V}}{\sigma}$.  
Now take $\sigma_1$ such that
$\ncl{\FL{V}}{\sigma}=\ncl{\FL{V}}{\sigma_1}=\FL{\sigma_1}=\dom{\sigma_1}$
by Lemma \ref{lem:ncl}. Note by definition
$l\not\in \dom{\sigma_1}$.
Now let
$\sigma=\sigma_1\uplus \sigma_2$.
Since $l\in \dom{\sigma}$,  we know 
$l\in \dom{\sigma_2}$, hence done.
The if-direction is obvious by definition of reachability.
\end{proof}
\NI The characterisation
says that if $x$ is unreachable from $u$ then, up to $\WB$, the store
can be partitioned into one covering all reachable names from $u$ and
another containing $x$.

Now we give the full definition of the satisfaction relation. 
For readability, we first list the auxiliary definitions 
many of which have already been stated before.

\begin{notation}\hfill
\begin{enumerate}[(a)]
\item $\expands{\MMM}{u}{e}$ denotes $(\nu \VEC{l})(\xi\CD
u:\MAP{e}_{\xi, \sigma}, \sigma)$ where we always assume $u$ is fresh
and the variables and labels in $e$ are free in $\MMM$.

\item
$\MMM/u$ denotes $(\new \VEC{l})(\xi,\sigma)$ if 
$\MMM=(\new \VEC{l})(\xi\cdot\AT{u}{V},\sigma)$; 
and if $u\not\in \FV{\MMM}$
we set $\MMM/u=\MMM$. 

\item
$\expands{\MMM}{u}{N}\converges \MMM'$ when 
$(N\xi, \sigma)\converges (\new \VEC{l}')(V, \sigma')$ and
$\MMM'=(\new\VEC{l}\VEC{l}')(\xi\CD u\!:\!V,\ \sigma')$
with $\MMM=(\new\VEC{l})(\xi,\ \sigma)$ 
where we always assume $u$ is fresh
and the variables and labels in $N$ are free in $\MMM$. 

\item
$\MMM\EVOLVES \MMM'$ is generated by: 
(1) $\MMM\EVOLVES\MMM$; and  
(2) if $\MMM\EVOLVES\MMM_0$ and 
$\MMM_0[u:N]\converges \MMM'$, then  
$\MMM\EVOLVES\MMM'$. 

\item
We write
$\updates{\MMM}{e}{V}$ 
for 
$(\nu \VEC{l})(\xi, \sigma[l\mapsto V])$
with 
$\MMM=(\nu \VEC{l})(\xi, \sigma)$ and 
$\MAP{e}_{\xi, \sigma}=l$. 

\item
We write
$\MMM\CDs\AT{\TVX}{\alpha}$ 
for 
$\MMM=(\nu \VEC{l})(\xi\CDs\AT{\TVX}{\alpha}, \sigma)$ 
with 
$\MMM=(\nu \VEC{l})(\xi, \sigma)$ where  
$\TVX$ is not in $\MMM$ and $\alpha$ is closed.

\end{enumerate}
\end{notation}

\begin{definition}[Satisfaction]
\rm 
\label{def:satisfaction}
The semantics of the assertions follows.  All omitted cases are by de
Morgan duality.  
\begin{enumerate}[(1)]
\item 
\label{model:eq}
$\MMM \models e_1 = e_2$ if $\MMM[u:e_1] \congmodel \MMM[u:e_2]$.

\item 
\label{model:and}
$\MMM \models C_1\AND C_2$ if $\MMM \models C_1$ and $\MMM \models C_2$.

\item 
\label{model:neg}
$\MMM \models \neg C$ if not $\MMM \models C$.

\item 
$\MMM\models \allworlds C$ 
if 
  $\forall \MMM'.(\MMM\EVOLVES \MMM' \ENTAILS \MMM'\models C)$.

\item 
\label{model:forall}
$\MMM \models \forall x.C$ 
if 
$\forall L\in \mathcal{F}.(\expands{\MMM}{x}{L}\converges \MMM' 
\ENTAILS \MMM'\models C)$  
\item 
\label{model:hidea}
$\MMM \models \HIDEa x.C$ if 
$\forall \MMM'.
((\new l)\MMM'\WB \MMM\,\ENTAILS\,\ \MMM'[x:l]\models C)$ 

\item 
\label{model:poly}
$\MMM \models \forall \TVX.C$ if 
for \mfb{all closed types} $\alpha$, $\MMM\CDs\AT{\TVX}{\alpha}\models C$.

\item 
\label{model:allcon}
$\MMM \models \allCon{e}{C}$ if 
for each $\forall L\in {\mathcal{F}}. 
\updates{\MMM}{e}{L}\models C$.

\item 
\label{model:reach}
$\MMM\models e_1\reachable e_2$
if for each $(\new \VEC{l})(\xi, \sigma) \WB \MMM$,
$\MAP{e_2}_{\xi, \sigma}\in \ncl{\FL{\MAP{e_1}_{\xi,
      \sigma}}}{\sigma}$. 

\item \label{oneeval}
$\MMM \models \ONEEVAL{e}{e'}{z}{C}$ if
$\exists \MMM'.(\MMM[x:ee']\converges \MMM'\ANDl \MMM'\models C)$. 

\item  \label{locatedassertion}
$\MMM \models \ONEEVAL{e}{e'}{z}{C}@\VEC{w}$ if
\[
\begin{array}{ll}
\exists \MMM'.(& \MMM[z:ee']\converges\MMM' \ANDl \MMM' \models
C'\ANDl \\
& 
\forall \MMM''.(\MMM[z:\LET{\VEC{x}}{!\VEC{w}}{\LET{y}{ee'}{\VEC{w}:=\VEC{x}}}]
\converges \MMM'' \;\ENTAILS\; \MMM''\WB
\MMM[z:()]))
\end{array}
\]
\end{enumerate}
In the defining clauses above, 
we assume
$\FV{e, e_{1,2}, e'}\subset\FV{\MMM}$,
$\FL{e, e_{1,2}, e'}\subset\FL{\MMM}$,
$\FV{L}\subset\FV{\MMM}$ and
$\FL{L}\subset\FL{\MMM}$,
as well as well-typedness of models and formulae.
\end{definition}

\NI In
Definition 
\ref{def:satisfaction},
(\ref{model:and}) and 
(\ref{model:neg}) are standard. 
(\ref{model:poly}) is from \cite{HY04PPDP}.  Others have already been
explained.  In 
(\ref{locatedassertion}),
the program
$\LET{\VEC{x}}{!\VEC{w}}{\LET{y}{ee'}{\VEC{w}:=\VEC{x}}}$
first keeps the content of $\VEC{w}$ in $\VEC{x}$
and executes the application $ee'$; then finally
restores the original content in $\VEC{w}$. 
By $\MMM''\WB \MMM[z:()]$
the resulting model $\MMM''$ has no state change w.r.t.
the original model $\MMM$, this means $ee'$ only
updates at $\VEC{w}$ up to $\WB$.


This concludes the introduction of the satisfaction relation for
the present logic. 
The properties of models are explored further in the 
rest of this section and in \S~\ref{sec:axioms}.

\subsection{Thin and Stateless Formulae}
\label{subsec:formulae}

In this subsection, we introduce two kinds of formulae which play a
key role in the reasoning principles of the present logic, in
particular the proof rules discussed in the next section.

The first definition introduces formulae in which the thinning of
unused variables from models can be done as in first-order logic. 

\begin{definition}[Thin Formula]\rm
\label{def:thin}
Let $\Gamma\proves C$ and $y\in\dom{\Gamma}$
such that $y\notin\FV{C}$.
Then we say that $C$ is {\em thin with respect to $y$ }
if for each $\MMM$ typable under $\Gamma$,
$\MMM\models C$ implies $\MMM/y\models C$. 
We say $C$ is {\em thin} if under each typing and for each $y\not\in\FV{C}$,
$C$ is thin w.r.t. $y$. 
\end{definition}

\NI In a thin formula $C$, reference names which do not appear in $C$
do not affect the meaning of $C$.  There are formulae which are not
thin (we see some examples below) but they are of a very special kind.
In our experience they never appear in practical reasoning
including our reasoning examples in \S~\ref{sec:example}.

As examples of formulae which are not thin, when an evaluation formula
occurs negatively, formulae may cease to be thin. Consider the
following satisfaction:
\[ 
(\new ll')(u:\lambda ().!l', \ x:l, l\mapsto l', l'\mapsto 1)
\ \models 
\ \someworld \ONEEVAL{u}{()}{z}{z=2}
\]
which means that $u$ is a function which might return 2 {\em someday}
since a value stored in $l'$ can be changed via $x$ (for example, by
the command $!x:=2$). When we delete $x$ from the above model, the
behaviour of $u$ will change as follows.
\[ 
(\new l')(u:\lambda ().!l',l':1)
\ \models \ \allworlds \ONEEVAL{u}{()}{z}{z=1}
\]
since now $u$ {\em always} returns $1$ when it is invoked. 
The above judgement entails:
\[ 
(\new l')(u:\lambda ().!l',l':1)
\ \not\models \ \someworld \ONEEVAL{u}{()}{z}{z=2}
\] 
Hence $\someworld \ONEEVAL{u}{()}{z}{z=2}$ is not thin.  Similarly
$\someworld\allworlds \APP{u}{()}=z\{z=0\}$ is not a thin formula.

As noted, formulae which are not thin hardly appear in reasoning; 
all formulae appearing in \S~\ref{sec:example}
are thin; the proof rules always generate thin formulae from thin
formulae. We shall however work with general formulae since many
results hold for none-thin formulae too.

The following syntactic characterisation of thin formulae is useful.

\begin{proposition}[Syntactically Thin Formula]\label{pro:thin} 
\begin{enumerate}[\em(1)]
\item 
\label{pro:thin_wrt_y}
If $\Gamma \proves C$, 
$\Gamma \proves y:\alpha$ and  
$\alpha\in \ASET{\UNIT,\BOOL,\NAT}$,
then $C$ is thin with respect to $y$. 

\item 
\label{pro:thin:base}
$e=e',e\not=e',\REACH{e}{e'}$ and $e\notreach e'$ are thin. 

\item 
\label{pro:thin:ih}
If $C,C'$ are thin w.r.t. $y$, 
then 
$C \AND C'$, $C \OR C'$, 
$\forall x^\alpha.C$ for all $\alpha$, 
$\exists x^\alpha.C$ with $\alpha\in \{ \UNIT,\BOOL,\NAT\}$, 
$\exists\TVX.C$, 
$\forall\TVX.C$, 
$\HIDEe x.C$, 
$\HIDEa x.C$, 
$\allworlds{C}$,  
$\allCon{x}{C}$ and 
$\ONEEVAL{e}{e'}{x}{C'}$ 
are thin w.r.t. $y$. 
\end{enumerate}
\end{proposition}
\begin{proof}
(1,2) are immediate. For (3),  
suppose
$C$ and $C'$ are thin w.r.t. $y$, $y\not\in\FV{C, C'}$ and 
$\MMM\models C \AND C'$.  Then
$\MMM\models C$ hence
$\MMM/y\models C$, similarly for $C'$, 
hence 
$\MMM/y\models C \AND C'$. Similarly for other cases.
Next let $C$ be thin w.r.t. $y$ and  
$\MMM\models \HIDEe x.C$.  Then
there exists $\MMM'$ such that 
$(\new l)\MMM'\WB \MMM$ and $\MMM'[x:l]\models C$.
Then $(\new l)\MMM'/y\WB \MMM/y$. By assumption, 
$\MMM'[x:l]/y\models C$, and hence $\MMM/y\models \HIDEe x.C$, as desired. 
Next let $C$ be thin w.r.t. $y$. 
Suppose $\MMM\models\ONEEVAL{e}{e'}{z}{C}$,
i.e. $\MMM[z:ee']\converges \MMM'$ and $\MMM'\models C$.
Then we have
$\MMM/y[z:ee']\converges \MMM'/y$.
Since $C$ is thin w.r.t. $y$, we have $\MMM'/y\models C$,
as required. 
\end{proof}
\martinb{The next set of formulae are {\em stateless formulae} whose validity
does not depend on the state part of the model, cf.~stateless formulae in 
\cite{GLOBAL,ALIAS}.
}

\begin{definition}[Stateless Formula]
\label{def:stateless}\rm 
$C$ is {\em stateless} iff $C \ENTAILS \allworlds C$ is valid. 
We let $A, B, A',B',\ldots$ range over stateless formulae.
\end{definition} 

\begin{proposition}[Stateless Formulae]\label{pro:stateless}
\begin{enumerate}[\em(1)]
\item For all $C$, $\allworlds C$ is stateless. 
\item If $C$ is stateless then 
$C \LITEQ \allworlds C \LITEQ \allworlds \allworlds C$.
\end{enumerate}
\end{proposition}
\begin{proof}
Both are immediate from the definition,
see also \S~\ref{subsec:axiom:allworlds} for further related results. 
\end{proof}

\NI The above proposition says that if $C$ is stateless then $C$
holds in any future state starting from the present state.  
The following generalisation of this notion says that the validity of
a formula does not depend on the stateful part of models \emph{except at
specific locations}. This notion is used by the axioms for local
invariants later. 
  
\begin{definition}[Stateless Formula Except $\VEC{x}$]
\label{def:stateless_except}\rm 
We say that $C$ is {\em stateless except $\VEC{x}$} if, whenever 
$\MMM\models C$ and $\MMM\EVOLVES \MMM'$ such that $\MMM$ and $\MMM'$ 
coincide in their content at $\VEC{x}$ of reference types, i.e.
\begin{enumerate}[(1)]
\item
$\MMM\WB(\new \VEC{l}_0)(\xi,\  \sigma)$;
\item
$\MMM'\WB (\new \VEC{l}_0\VEC{l}_1)(\xi\CD\xi',\ \sigma')$; and
\item
$\sigma(\xi(x_i))=\sigma'(\xi(x_i))$ 
for each $x_i\in\ASET{\VEC{x}}$,
\end{enumerate}
then $\MMM'\models C$.
\end{definition} 

\NI 
Definition \ref{def:stateless_except}
uses the internal representation of models. 
Alternatively we may define an {\em $\VEC{x}$-preserving term}
which has the shape:
\begin{equation}
\LET{y_1}{!x_1}{...\LET{y_n}{!x_n}{\LET{z}{N'}{(x_1:=y_1;...;x_n:=y_n;z)}}}
\end{equation}
then say $C$ is stateless except $\VEC{x}$ 
if whenever $\MMM\models C$ and $\MMM[u:N]\converges \MMM'$
where $N$ is a $\VEC{x}$-preserving term we have
$\MMM'\models C$.

Note if $\VEC{x}$ is empty in Definition
\ref{def:stateless_except} then the third clause is vacuous: hence in
this case the definition means that for each $\MMM$ such that
$\MMM\models C$ we have $\MMM\EVOLVES \MMM'$ implies $\MMM'\models C$,
that is $C$ is stateless.

It is convenient to be able to check the statelessness of formulae
(relative to references) syntactically. For an inductive
characterisation, we introduce the following notion. As always 
we assume the standard bound name convention.

\begin{definition}[Tame Formulae]\rm 
The set of \emph{tame formulae}
is generated by the following rules: 
\begin{enumerate}[$\bullet$]

\item $e_1 = e_2$ and $e_1 \not= e_2$ are tame.

\item $e_1 \reach e_2$ and $e_1 \noreach  e_2$ are tame.

\item For any $C$, $\allworlds C$ is tame.

\item if $C$ is tame then
$\forall y^\alpha.C$, 
$\exists y^\alpha.C$,
$\exists\TVX.C$, 
$\forall\TVX.C$, 
$\allCon{y}{C}$ and  
$\someCon{y}{C}$
are all tame.

\item if $C,C'$ are tame
then $C \AND C'$ and $C \OR C'$ are tame. 
\end{enumerate}
We say that $!x$ is an {\em active dereference} in $C$
if $C$ is tame and $!x$ (with $x$ being free or bound) occurs neither
in the scope of $\allworlds$, $[!x]$ nor
$\ENCan{!x}$. 
\end{definition}

\NI The following result (though not used in the present work) is notable
for carrying over  reasoning techniques from the logic for
aliasing \cite{ALIAS}.

\begin{proposition}[Decomposition]
\label{pro:cq:elimination}
Suppose $C$ is tame. Then there is tame $C'$ such that
$C\LITEQ C'$ and $C'$ does not contain content quantifications
except under the scope of $\allworlds$.
\end{proposition}
\begin{proof}
  The proof follows precisely that of \cite[\S 6.1,
  Theorem 1]{ALIAS}.
\end{proof}

\NI We can now introduce syntactic stateless formulae.

\begin{definition}[Syntactic Stateless Formulae]
\label{def:syntactic:stateless}
  \rm We say \emph{$C$ is syntactically stateless except $\VEC{x}$}
  if $C$ is tame and only names from $\VEC{x}$ are among the active
  dereferences in $C$.
\end{definition}
\begin{proposition}\label{pro:stateless:B}\hfill
  \begin{enumerate}[\em(1)]
  \item If $C$ is syntactically stateless except $\VEC{x}$ then 
    $C$ is stateless except $\VEC{x}$.
  \item 
If
$\allCon{\VEC{x}}{C}$ is syntactically stateless
    then $C$ is stateless except $\VEC{x}$.
  \end{enumerate}
\end{proposition}
\begin{proof}
   (1) is by induction of the generation of tame formulae.  Base cases and
  $\allworlds C$ are immediate.  Among the inductive cases the only
  non-trivial case is quantifications of references.  Suppose $C$ is
  tame and contains active dereferences at $\VEC{x}y$.
\begin{enumerate}[$\bullet$]
\item If the validity of $C$ relies on $y$ (i.e. 
  for some $\MMM_{1,2}$ which differ only at $y$ we have  
  $\MMM_1\models C$ and 
  $\MMM_2\not\models C$)
then $\forall
  y^\alpha.C$ is falsity: if not $\forall y^\alpha.C$ and $C$ are
  equivalent.  In either case we know $C$ is stateless except
  $\VEC{x}$. 
\item If validity of $C$ relies on $y$ then $\exists
  y^\alpha.C$ is truth: if not $\exists y^\alpha.C$ and $C$ are
  equivalent. The rest is the same.  
\item If validity of $\allCon{y}{C}$ relies on the content of $y$ then
  $\allCon{y}{C}$ is falsity: the rest is the same. 
Similarly for $\someCon{y}{C}$.

\end{enumerate}
The cases of $C \AND C'$ and $C \OR C'$ are immediate by induction.
(2) is an immediate corollary of (1).
\end{proof}



\section{Proof Rules and Soundness}\label{sec:rules}
\subsection{Hoare Triples}\label{subsec:triple}
\NI 
This subsection summarises judgements and  proof rules for local state.
The main judgement consists of a program and a pair of formulae following 
Hoare \cite{HOARE}, augmented with a fresh name called \emph{anchor}
\cite{HY04PPDP,SHORT1,GLOBAL}, 
\[
        \ASSERT{C}{M}{u}{C'}
\] 
which says: 
\begin{quote}
{\em If we evaluate $M$ in the initial state satisfying $C$,
then it terminates with a value, name it $u$, and a final state,
which together satisfy $C'$.}
\end{quote}
Note that our judgements are about total correctness.
Sequents have identical shape as those in \cite{GLOBAL,ALIAS}: the
computational situations described is however quite different, in that
both $C$ and $C'$ may now describe behaviour and data structures with
local state.  

The same sequent is used for both validity and provability. If we wish
to be specific, we prefix it with either $\proves$ (for provability) or
$\models$ (for validity). 
We assume that judgements are 
well-typed in the sense that, in
$\ASSERT{C}{M}{u}{C'}$ with $\Gamma \proves M:\alpha$,  
  $\Gamma, \Delta\proves C$ and $\ATb{u}{\alpha}, \Gamma,
  \Delta \proves C'$ for some $\Delta$ such that $\dom{\Delta}\cap(\dom{\Gamma}\cup\ASET{u})=\emptyset$.

In $\ASSERT{C}{M}{u}{C'}$,  
the name $u$ is the \emph{anchor} of the judgement, which should
\emph{not} be in $\dom{\Gamma} \cup \FV{C}$; and
$C$ is the \emph{pre-condition} and $C'$ is the {\em post-condition}.
The \emph{primary 
names} are $\dom{\Gamma}\cup\ASET{u}$, while the
\emph{auxiliary names} (ranged over by $i,j,k,...$) 
are those free names in $C$ and $C'$ which are
not primary. 
An anchor is used for naming the value from $M$ and for specifying
its behaviour. 
We use the abbreviation 
$\ASET{C}\,M\ASET{C'}$ to denote 
$\ASET{C}\,M:_u\ASET{u = () \AND C'}$.  


\begin{myfigure}
{\small
\vspace{-2.4mm}
\[
\begin{array}{c}
        [\mbox{\em Var}]
        \,
        \infer
        {-}
        {
          \ASSERT{C\MSUBS{x}{u}}{x}{u}{C}
        }
        \quad
        [\mbox{\em Const}] 
        \infer
        {-}
        {
          \ASSERT{C\MSUBS{\mathsf{c}}{u}}{\mathtt{c}}{u}{C}
        }
\quad 
        [\mbox{\em Add}]
                \,
                \infer
        {
            \ASET{C}\,M_1:_{m_1}\ASET{C_0}
            \ 
            \ASET{C_0}\,M_2:_{m_2}\ASET{\,C'\MSUBS{m_1+m_2}{u}\,}
        }
        {
          \ASET{C}\, M_1+M_2\, :_u\,\ASET{C'}
        }
\\[4mm]
        [\mbox{\em In}_1]
                \,
        \infer
        {
                \ASSERT{C}{M}{v}{C'\MSUBS{\LOGICINJ{}{1}{v}}{u}}
        }
        {
                \ASSERT{C}{\INJ{1}{M}}{u}{C'}
        }
                \quad 
        [\mbox{\emph{Case}}]
                \,
        \infer
        {
          \ASSERT{C^{\minus \VEC{x}}}{M}{m}{C_0^{\minus \VEC{x}}}
          \quad
          \ASSERT{C_0\MSUBS{\LOGICINJ{}{i}{x_i}}{m}}{M_i}{u}{{C'\,}^{\minus \VEC{x}}} 
        }
        {
          \ASSERT{C}{\CASE{M}{x_i}{M_i}}{u}{C'}
        }    
        \\[4mm]
        [\mbox{\em Proj}_1]
                \,
        \infer
        {
                \ASSERT{C}{M}{m}{C'\MSUBS{\pi_1(m)}{u}}
        }
        {
                \ASSERT{C}{\pi_1(M)}{u}{C'}
        }    
        \quad 
        [\mbox{\em Pair}]
                \,
        \infer
        {
          \begin{array}{l}
            \ASSERT{C}{M_1}{{m}}{C_0}
            \quad
            \ASSERT{C_0}{M_2}{{n}}{C'\MSUBS{\ENCan{m,n}}{u}}
          \end{array}
        }
        {
                \ASSERT{C}{\ENCan{M_1,M_2}}{u}{C'}
        }
       \\[4mm]
        [\mbox{\emph{Abs}}]
                \,
        \infer
        {
          \ASSERT{A^{\minus x\VEC{i}}\AND C}{M}{m}{C'}\  
        }
        {\ASSERT{A}{\lambda x.M}{u}{\allworlds \forall
            x\VEC{i}.(\ASET{C}\APP{u}{x}=m\ASET{C'})}}
\quad  
        [\mbox{\emph{App}}]
                \,
        \infer
        {
        \begin{array}{l}
          \ASSERT{C}{M}{m}{C_0}\  
          \ASSERT{C_0}{N}{n}{\APP{m}{n}=u\ASET{C'}\,}              
        \end{array}
        }
        {\ASSERT{C}{MN}{u}{C'}}
                \\[4mm]
   [\mbox{\emph{If}}]
             \,
            \infer
       {
                 \ASSERT{C}{M}{b}{C_0}
                 \quad
                 \ASSERT{C_0\MSUBS{\true}{b}}{M_1}{u}{C'}
                 \quad
                 \ASSERT{C_0\MSUBS{\false}{b}}{M_2}{u}{C'}
        }
        {
                \ASSERT
                     {C}
                     {\IFTHENELSE{M}{M_1}{M_2}}
                     {u}
                     {C'}
        }
        \\[4mm]
        \hspace{-3mm}
        [\mbox{\emph{Deref}}]
                \,
        \infer
        {
          \ASSERT{C}{M}{m}{C'\MSUBS{!m}{u}}
        }
        {
          \ASSERT{C}{!M}{u}{C'}
        }
        \quad 
        [\mbox{\emph{Assign}}]
                \,
        \infer
        {
          \ASSERT{C}{M}{m}{C_0}
          \quad
          \ASSERT{C_0}{N}{n}{C'\LSUBS{n}{\,!m}}
        }
        {
          \ASET{C}\ M := N\ \ASET{C'}
        }
        \\[4mm]
        [\mbox{\em Rec-Ren}]
                \,
        \infer
        {
          \ASSERT{A^{\minus f}}{\lambda x.M}{u}{B}
	}
        {
          \ASSERT{A}{\mu f.\lambda x.M}{u}{B\MSUBS{u}{f}}
        }
        \quad 
        \LEFTONEPREMISERULENAMED
        {Ref}
        {
                \ASSERT{C}{M}{m}{C'}
        }
        {
                \ASSERT
                {C}
                {\REFPROG{M}}
                {u}
                {
                {\HIDEe x.(C'\SUBST{!u}{m}\AND u\noreach i^{\TVXscript}\AND u=x)}
                }
        }
\\[4mm]
\rulenameB{Aux${}_\forall$V}
\infer
{
\ASET{C^{\minus i}}\ V:_u\, \ASET{C'}
}
{
\ASET{C}\ V:_u\, \ASET{\forall i.C'}
}
\quad 
\rulenameB{Aux${}_\forall$}
\infer
{
\ASET{C^{\minus i}}\ M:_u\, \ASET{C'}
\quad \alpha \ \text{is base type}
}
{
\ASET{C}\ M:_u\, \ASET{\forall i^\alpha.C'}
}
        \\[4mm]
         [\mbox{\emph{Conseq}}]
        \infer
         {
           C \ENTAILS C_0
           \;
           \ASSERT{C_0}{\!M}{u\!}{C_0'}
           \;
           C_0' \ENTAILS C'
         }
         {
           \ASSERT{C}{M}{u}{C'}
         }
        \quad\quad 
\mbox{[{\em Subs}]}\ 
        \infer 
        {
                \ASSERT{C}{M}{u}{C'}
                \quad u\not\in \PFN{e}
        }
        {
                \ASSERT{C\SUBST{e}{i}}{M}{u}{C'\SUBST{e}{i}}
        } 
       \\[6mm]
       [\mbox{\emph{Cons-Eval}}] \ \
\infer
{
\begin{array}{l}
\ASSERT{C_0}{M}{m}{C'_0}
\quad 
\ x\ \text{fresh};
\;\ \text{$\VEC{i}$ auxiliary}\\
\always\forall \VEC{\TVX}.\forall \VEC{i}.\EVAL{C_0}{x}{()}{m}{C'_0}
\ \ENTAILS\
\always\forall \VEC{\TVX}.\forall \VEC{i}.\EVAL{C}{x}{()}{m}{C'}
\end{array}
}
{\ASSERT{C}{M}{m}{C'}}
\end{array}
\]

\NI We require $C'$ is thin w.r.t. $m$ in [{\em Case}] and [{\em Deref}], and 
$C'$ is thin w.r.t. $m,n$ in [{\em App}, {\em Assign}].   
\caption{Proof Rules}\label{figure:rules}
}
\end{myfigure}

\subsection{Proof Rules}
\label{subsec:rules}
The full compositional proof rules and new structure rules are given in Figure
\ref{figure:rules}. 
In each proof rule, we assume all occurring judgements to be
well-typed and no primary names in the premise(s) to occur as
auxiliary names in the conclusion. 
We write $C^{\minus \VEC{x}}$ to indicate $\FV{C}\cap\ASET{\VEC{x}}=\emptyset$.
Despite our 
semantic enrichment, all compositional proof rules in the base logic
\cite{ALIAS} (\mfb{and} [{\em Rec}-{\em Ren}] from \cite{completeness}) 
 syntactically stay as they are, 
except for: 
\begin{enumerate}[$\bullet$]

\item \mfb{adding a rule for } the reference generation,

\item revising [{\em Abs}] and [{\em App}] \mfb{so they use one-sided
evaluation formulae},

\item adding the thinness condition in the post-condition of the
conclusion in [{\em Case}], [{\em App}], [{\em Assign}] and [{\em
Deref}]

\end{enumerate}
The thinness condition is required when the anchor names used in the
premise contribute to $C'$ in the conclusion.  The reason for this
becomes clearer when we \mfb{prove  soundness}. This condition does not
jeopardise the completeness of our logic. All reasoning examples 
we have explored meet this condition including 
those in \S~\ref{sec:example}.

Note that in [{\em Add}], since $C'$ is always thin with respect to
$m_i$ by Proposition \ref{pro:thin} (\ref{pro:thin_wrt_y}), we do not
have to state this condition explicitly. Similarly for [{\em If}]
since $C'$ is always thin with respect to $b$.

[{\em Assign}] uses {\em logical substitution} which is built with content
quantification to represent  substitution of content of a possibly
aliased reference \cite{ALIAS}.
\begin{equation*}
C\LSUBS{e_2}{!e_1}
\ \; 
\DEFEQ
\ \;
\forall m.(m=e_2\ENTAILS \allCon{e_1}{(!e_1=m \ENTAILS C)}).
\end{equation*}
with $m$ fresh (\mfb{we have a dual} characterisation by $\ENCan{!e_1}$). 
Intuitively $C\LSUBS{e_2}{!e_1}$ describes the situation where a model
satisfying $C$ is updated at a memory cell referred to by $e_1$ 
(of a reference type) with a value $e_2$ (of its content type),  
with $e_{1,2}$ interpreted in the current model. 
 
In rule [{\em Ref}], $u\noreach i$ indicates that the newly generated
cell $u$ is unreachable from any $i$ of arbitrary type $\TVX$ in the
initial state: then the result of evaluating $M$ is stored in that
cell.\footnote{One may write the conclusion of this rule 
  as $ \ASSERT {C} {\REFPROG{M}} {u}
  {(C'\SUBST{!u}{m}\AND u\noreach i^{\TVXscript})}$ which may be useful for readability.  In this
  paper however we intentionally do not introduce this or other
  abbreviations for the sake of clarity.} Here $i$ is a(ny) fresh
variable denoting an arbitrary datum which already exited in the
pre-state. Just as the standard auxiliary variable in Hoare-like
logics, this $i$ is semantically bound at the sequent level.  In a
large proof, we may want each instance of [{\em Ref}] to use a fresh and
distinct variable, even though in practice we usually
apply the substitution rule discussed below to instantiate this
``bound'' variable into an appropriate expression so name clash may
not occur.\footnote{The treatment of a fresh 
variable as an input binder in [{\em Ref}] is useful for 
mechanisation of reasoning, just like auxiliary variables in Hoare triples.}

For the structural rules (i.e.~those which only manipulate
assertions), those given in \cite[\S 7.3]{ALIAS} for the base logic
stay valid except that the universal abstraction rule
$\rulenameB{Aux${}_\forall$}$ in \cite[\S 7.3]{ALIAS} needs to be
weakened as $\rulenameB{Aux${}_\forall$}$ and 
$\rulenameB{Aux${}_\forall$V}$
in Figure \ref{figure:rules}. 
\martinb{Note that the original structural rule $\rulenameB{Aux$_{\forall}$}$, which does not}
have this condition, is not valid in the presence of new reference
generation. For example we can take:
\begin{equation}\label{newrefex}
\ASET{\truth}\ \REFPRG{3}:_u\, \ASET{u\noreach i\AND !u=3}
\end{equation}
which is surely valid. But without the side condition, we can infer
the following from (\ref{newrefex}).
$$
\ASET{\truth}\ \REFPRG{3}:_u\, \ASET{\forall i.(u\noreach i\AND !u=3)}
$$ 
which does not make sense (just substitute $u$ for $i$). This is
because 
$i$ cannot range over newly generated names: such an interplay with new name generation is
not possible if the target program is a value, or if $i$ is of  base
type.

We also have two useful structural rules added in the present logic.
The first rule is [{\em Subs}] in Figure \ref{figure:rules},
which can be used  to instantiate the fresh variable $i$
in [{\em Ref}] with an arbitrary datum.  The rule uses the following
set of reference names.

\begin{definition}[Plain Name] \label{def:plain}
We write $\PFN{e}$ for the set of
\emph{free plain names} of $e$, defined as:
$\PFN{x} = \{x\}$, $\PFN{\mathsf{c}} =
\PFN{!e} = \emptyset$, 
$\PFN{\ENCan{e, e'}} = \PFN{e}\cup\PFN{e'}$,
and 
$\PFN{\injection_i(e)} = \PFN{e}$.
\end{definition}

\NI In brief, the set of free plain names of $e$ contains reference names in
$e$ that do not occur dereferenced, as first described in Definition
\ref{def:plain}.  As we shall see later, the side condition for [{\em
Subs}] using $\PFN{e}$ is necessary for  soundness.

As an example usage of [{\em Subs}], consider:
\begin{equation}\label{BBBBB}
  \ASET{!z=2} \refPrg{2}:_m \ASET{!m=2 \AND i\noreach m} 
\end{equation}
where we take off $\HIDEe$ by an axiom later. We can then use
[{\em Subs}] to show:
\begin{equation}\label{CCCCC}
  \ASET{!z=2} \refPrg{2}:_m \ASET{!m=2 \AND z\noreach m} 
\end{equation}
Note $m \in \PFN{m}$: hence we {\em cannot} use $m$ instead of $z$ in (\ref{CCCCC}), 
which is obviously unsound. As another use of 
[{\em Subs}], consider a judgement:
\begin{equation}\label{EEEEEE}
  \ASET{\truth} \PAIR{\refPrg{2}}{\refPrg{2}}:_m 
  \ASET{!\pi_1(m)=2 \AND !\pi_2(m)=2 \AND \pi_1(m)\neq\pi_2(m)}
\end{equation}
In order to derive (\ref{EEEEEE}), we simply combine (\ref{BBBBB}) with
the following judgement:
\begin{equation}
  \ASET{!m=2 \AND i\noreach m} \refPrg{2}:_n \ASET{!m=2 \AND !n=2\AND j\noreach n}
\end{equation}
where we use a different fresh variable $j$. We can now replace $j$
with $m$ using [{\em Subs}], and via [{\em Cons}] we obtain:
\begin{equation}
  \ASET{!m=2 \AND i\noreach m} \refPrg{2}:_n \ASET{!m=2 \AND !n=2\AND m\neq n}
\end{equation}
from which we can infer (\ref{EEEEEE}) by pairing, combined with
(\ref{BBBBB}). 

Another significant additional rule is [{\em Cons-Eval}], also given
in Figure \ref{figure:rules}. This is a strengthened version of the
standard consequence rule, and is used when incorporating the local
invariant axiom of the evaluation formula for derivations of the
examples in \S~\ref{sec:example}.  Technically, this is a consequence
of (a) having a proof system by which we can compositionally build
proofs; and (b) representing fresh generation of references by
disjointness from fresh variables.  We shall see in examples that it
is useful in reasoning.


The full list of  structural rules can be found in Appendix
\ref{app:rules}.

\subsection{Located Judgements} 
\label{subsec:derived}
\martinb{Proof rules which contain an explicit 
 effect set (similar to
 located evaluation formulae)  were  introduced in \cite{ALIAS}  and  
are of substantial help in reasoning about programs. }
Located Hoare triples take the form: 
\[ 
\ASET{C} M :_u\!\ASET{C'}@\VEC{e}
\] 
where each $e_i$ is of a reference type and does not contain
\mfb{(sub)expressions of the form $!g$}.\footnote{\label{footnote:located} This
restriction is for a simplification of the
interpretation, and can be taken off if $\VEC{e}$ is interpreted in
the pre-condition.} 
$\VEC{e}$ is called {\em effect set}. 
We prefix it with either $\proves$ (for provability) or
$\models$ (for validity) 
if we wish to be specific. 
 
The full rules are listed in
Figure \ref{fig:rules:compositional:located} (proof rules) and Figure
\ref{fig:rules:structural:located} (structure rules) in Appendix \ref{app:rules}.
All rules come from \cite{ALIAS}
except for the new name generation rule and the universal
quantification rule, both corresponding to the new rules in the basic
proof system. 
The structures rules are also revised along the lines of 
Figure \ref{figure:rules}. 


\subsection{Invariance Rules for Reachability}
\label{sub:invariant}
Invariance rules are useful for modular reasoning. A simple form 
is when there is no state change:
\[ 
[\mbox{\emph{Inv-Val}}]
        \,
        \infer
        {
          \ASSERT{C}{V}{m}{C'}
        }
        {
          \ASSERT{C\AND C_0}{V}{m}{C'\AND C_0}
        }
\] 
Alternatively if a formula is stateless it continues to hold 
irrespective of state change.
\[ 
[\mbox{\emph{Inv-Stateless}}]
        \,
        \infer
        {
          \ASSERT{C}{M}{m}{C'}
        }
        {
          \ASSERT{C\AND \allworlds C_0}{M}{m}{C'\AND \allworlds C_0}
        }
\] 
When it is formulated with (un)reachability predicates, however, one needs some
care.  Since reachability is a stateful property, it is generally {\em
  not} invariant under state change. For example, suppose $x$ is
unreachable from $y$; after running $y:=x$, $x$ becomes reachable from
$y$. Hence the following rule is unsound.
\[ 
\mbox{[{\em Unsound-Inv}]}
 \infer
{
\ASSERT{C}{M}{m}{C'}
}
{
\ASSERT{C\AND e\notreach e'}{M}{m}{C'\AND e\notreach e'}
}
\qquad
\mbox{(unsound)}
\]
From the following general invariance rule [{\em Inv}],  
we can derive an invariance rule for $\noreach$. 
\[ 
\begin{array}{c}
\mbox{[{\em Inv}]}
         \infer
         {
           \ASSERT{C}{M}{m}{C'}@\VEC{w}
	  \qquad
          C_0\mbox{ is tame}
         }
         {
         \ASSERT
         {C\AND \allCon{\VEC{w}}{C_0}}
         {M}{m}
         {C'\AND \allCon{\VEC{w}}{C_0}}@\VEC{w}
         }
\end{array}
\]
In [{\em Inv}], the effect set $\VEC{w}$ gives the minimum information
by which the assertion we wish to add, $C_0$, can be stated as an
invariant since $[!\VEC{w}]C_0$ says that $C_0$ holds regardless of
the content of $\VEC{w}$. Thus $C_0$ can stay invariant after
execution of $M$.  Unlike the existing invariance rules as found in
 standard Hoare logic or in Separation Logic \cite{REYNOLDS}, we
need no side condition ``$M$ does not modify stores mentioned in
$C_0$'': $C$ and $C_0$ may even overlap in their mentioned references,
and $C$ does not have to mention all references $M$ may read or write.
 
The  following instance of 
[{\em Inv}] is useful. 
\[
\begin{array}{c}
[\mbox{\emph{Inv-$\noreach$}}]
        \,
        \infer
        {
          \ASSERT{C}{M}{m}{C'}@x \quad \text{no dereference occurs in $\VEC{e}$}
        }
        {
        \ASSERT{C\AND x \noreach \VEC{e}}{M}{m}{C'\AND x \noreach \VEC{e}}@x
        }
\end{array}
\]
In $[\mbox{\emph{Inv-$\noreach$}}]$, 
we note $\allCon{x}{x \noreach \VEC{e}}\LITEQ x \noreach \VEC{e}$
is always valid if $\VEC{e}$ contains
no dereference $!e$, cf.~Proposition \ref{pro:notreach} 3-(5) 
later.
Hence $x \noreach \VEC{e}$ is stateless except at $x$. 
The side condition is indispensable:
consider $\ASET{\truth} x:=x\ASET{\truth}@x$ (which is typable with
recursive types), which does not imply
$\ASET{x\noreach !x} x:=x\ASET{x\noreach !x}@x$. 

One of the important aspects of these invariance rules is that the
effect set of a located judgement or assertion can contain a hidden name --
a name which has been created and which is (partially) accessible. 
For example, we can infer (using $\rulenameB{LetRef}$ in 
\S~\ref{subsec:newvar}):\footnote{Since $!y$ is stated in the
pre-condition, we can also write 
$\ASSERT{\truth}{\LETREF{x}{2}{!y=x}}{u}{\new x.(!x=2 \AND !!y=x \AND x\noreach i)}@!y$, cf.~footnote \ref{footnote:located}.
} 
\[ 
\ASSERT{!y=h}{\LETREF{x}{2}{!y=x}}{u}{\new x.(!h=x \AND !x=2 \AND !y=h
\AND x\noreach i)}@h
\]


\subsection{Soundness}
\label{subsec:soundness}
\NI 
Let $\MMM$ be a model $(\new \VEC{l})(\xi,\sigma)$ of type
$\Gamma$, and $\Gamma\proves M:\alpha$ with $u$
fresh.  
Then \emph{validity}  $\models
\ASET{C}M:_u\ASET{C'}$ is given by 
(with $\MMM$ including all variables in $M$, $C$ and $C'$ except $u$): 
\[
\models \ASET{C}M:_u\ASET{C'}
\;\ \IFFDEF\;\ 
\forall \MMM.
(
\MMM\models C\ \ENTAILS \ (\expands{\MMM}{u}{M}\converges\MMM'
\ANDl \MMM' \models C'
))
\]
where the notation $\expands{\MMM}{u}{N}\converges\MMM'$ 
appeared in Definition \ref{def:satisfaction}(c). 
This is equivalent to, with $V\DEFEQ \lambda ().M$:
\begin{equation}\label{above}
\forall \MMM.
(\MMM[m:V]\models\allworlds \EVAL{C}{m}{()}{u}{C'})
\end{equation}
Similarly the semantics of the located judgement:
\begin{equation}\label{here}
\models\ASET{C}\ M:_u\,\ASET{C'}@\VEC{x}
\end{equation}
is given through the corresponding 
located assertion, 
using the following term (let $z$ be fresh):
\begin{equation}\label{therethereagain}
V
\;\DEFEQ\;
\LETREF{z}{0}{\lambda ().\IFTHENELSE{!z=0}{\LET{m}{M}{(z:=!z+1;m)}}{\Omega}}
\end{equation}
where $\Omega$ is a diverging closed term (in fact any closed program
works). The use of $z$ is to prevent leakage of information
from $m$ after the evaluation: after evaluation $m$ can never reveal
any information thus it is the same thing as evaluating $M$ once. 

With this $V$ we set the definition of
(\ref{here}) as follows:
\begin{equation}\label{below}
\forall \MMM.
(\MMM[m:V]\models
\allworlds \EVAL{C}{m}{()}{u}{C'}@\VEC{x})
\end{equation}
Among the proof rules the only non-trivial addition from the preceding
systems (in fact the only difference) is the rule for reference
generation.  For its soundness we use the free plain names 
as defined in Definition \ref{def:plain} (recall
$\PFN{e}$ is the set of reference names in $e$ 
that do not occur dereferenced).  For free plain names we note:

\begin{lemma} \label{lem:fresh}
Let $u \notin \PFN{e}$. Then for all $M$, with $u$ fresh, we have:
$\expands{\MMM}{u}{\REFPROG{M}} \converges \MMM'$ 
implies $\MMM' \models u \noreach e$. 
\end{lemma}
\begin{proof}
Suppose 
$\MMM=(\nu \VEC{l})(\xi, \sigma)$ and  
$\expands{\MMM}{u}{\REFPROG{M}} \converges \MMM'$. Then $\MMM'
=(\nu \VEC{l}l)(\xi \cdot u : l, 
\sigma \cdot [l\mapsto V])$ with 
$u\not\in \FV{\xi}$, 
$l\not\in \FL{\sigma,\xi}$ and 
$
(\nu \VEC{l}_0)(M\xi, \sigma_0)
\Downarrow 
(\nu \VEC{l}_0)(V, \sigma)$.  
Then one can check 
$\MAP{i}_{\xi\cdot u:l,\sigma \cdot [l\mapsto V]}=
\MAP{i}_{\xi,\sigma}$ and $\MAP{i}_{\xi,\sigma}\not \in 
\ncl{l}{\sigma \cdot [l\mapsto V]}=\ncl{l}{[l\mapsto V]}$. 
\end{proof}

\NI We can now establish:

\begin{theorem}[Soundness]\label{thm:sound}
$\proves \ASET{C}M:_u\ASET{C'}$ implies $\models
\ASET{C}M:_u\ASET{C'}$.
\end{theorem}
\begin{proof} 
Except [{\em Ref}], all rules
precisely follow \cite[\S 8.2]{ALIAS} (except for the use of thinness
which allows the same reasoning as in \cite[\S 8.2]{ALIAS} to go through).
For [{\em Ref}], we have, with $l$ fresh:  
\begin{reasoning}
\MMM\models C 
&\;\THENs\;&\;
\expands{\MMM}{m}{M}\converges \MMM' \ANDl \MMM'\models C'
&\quad \text{Hypothesis}\\
&\;\THENs\;&\;
\expands{\expands{\MMM}{m}{M}}{u}{\refPrg{m}}\converges
(\new l)\MMM''\ANDl 
\MMM'' \models C'\ANDl !u=m
\\
& & \mbox{with} \ \MMM''
\DEFEQ \MMM'[u:l][l\mapsto V] \\
&\;\THENs\;&\; 
\expands{\MMM}{u}{\refPrg{M}}\converges
(\new l)\MMM''/m \ANDl \MMM''/m \models C'\SUBST{!m}{u}\ANDl
\NOTREACH{u}{i}
\quad 
&\quad \text{Lemma \ref{lem:fresh}}\\
&\;\THENs\;&\;
\MMM''/m [x:l] \models 
C'\SUBST{!m}{u}\ANDl \NOTREACH{u}{i}\ANDl x=u\\
&\;\THENs\;&\;
(\new l)\MMM''/m \models 
\new x.(C'\SUBST{!m}{u}\ANDl \NOTREACH{u}{i} \ANDl x = u) 
\end{reasoning}
See Appendix \ref{app:soundness} for the full proofs. 
\end{proof}

\begin{theorem}[Soundness]\label{thm:soundlocated}
$\proves \ASET{C}M:_u\ASET{C'}@\VEC{e}$ implies $\models
\ASET{C}M:_u\ASET{C'}@\VEC{e}$.
\end{theorem}
\begin{proof}
As above (and for remaining rules  as in \cite[\S 8.2]{ALIAS}). 
See Appendix \ref{app:soundness} for [{\em Ref}] and the invariant rules. 
\end{proof}

\section{Axioms and Local Invariants}
\label{sec:axioms}
This section studies the basic axioms for the logical constructs, 
including those for local state. 

\subsection{Axioms for Equality}
\label{sub:axiom:eq}
Equality, logical connectives and quantifiers satisfy the standard
axioms (quantifications need a modest use of thinness, see Proposition \ref{pro:hiding}
later). For logical connectives, this is direct from the
definition. For equality and quantification, however, this is not
immediate, due to the non-standard definition of their semantics.

First we check the equality indeed satisfies the standard axioms for
equality. We start from the following lemmas.
$C\MMSUBS{u}{v}$ denotes a simultaneous substitution. 

\begin{lemma}\label{lemma:29410} 
Let $\MMM$ have type $\Gamma$.
\begin{enumerate}[\em(1)]

\item\label{lem:sat:renaming} (injective renaming)\
  Let $u, v\in\dom{\Gamma}$. Then
        $\MMM\models C$
        iff 
        $\MMM\MMSUBS{u}{v}\models C\MMSUBS{u}{v}$.

\item\label{lem:sat:perm} (permutation)\
   Let $u, v\in\dom{\Gamma}$. Then we have
        $\MMM\models C$
        iff 
        $\piprm{uv}{vu}\MMM\models C\MMSUBS{u}{v}$.

\item\label{lemma:29410:5} (exchange)\
\label{lem:sat:exchange}
        Let $u, v\not\in\FV{e, e'}$. Then we have
        $\expands{\expands{\MMM}{u}{e}}{v}{e'}\models C$
        iff 
        $\expands{\expands{\MMM}{v}{e'}}{u}{e}\models C$.

\item\label{lemma:29410:6}\ (partition and monotonicity)\ 
Let
$\MMM=(\new \VEC{l})(\xi, \sigma)$ be of type $\Gamma$ and
$\MMM'=(\new \VEC{l}\VEC{l}')(\xi\CD\xi', \sigma\CD\sigma')$ be 
  such that $(\FL{\sigma'}\cup\FL{\xi'})\cap\ASET{\VEC{l}}=\emptyset$.
Further let $\Gamma\proves C$.
Then $\MMM\models C$ iff $\MMM'\models C$.
In particular with
$u\not\in\fv{C}$ 
we have
$\MMM \models C$ 
iff $\expands{\MMM}{u}{V} \models C$.

\item\label{lemma:29410:1}\ (symmetry)\
        $\MMM \models e_1 = e_2$ iff for fresh and distinct $u, v$:
$\expands{\expands{\MMM}{u}{e_1}}{v}{e_2} \congmodel
\expands{\expands{\MMM}{u}{e_2}}{v}{e_1}$.


\item\label{lemma:29410:3}
\label{lemma:29410:4}
\ (substitution)\ 
 $\expands{\expands{\MMM}{u}{x}}{v}{e} \congmodel
\expands{\expands{\MMM}{u}{x}}{v}{e\SUBST{u}{x}}$; and 
        $\expands{\expands{\MMM}{u}{e}}{v}{e'}
                \congmodel
        \expands{\expands{\MMM}{u}{e}}{v}{e'\SUBST{e}{u}}$.
\end{enumerate}
\end{lemma}
\begin{proof}
All are elementary, mostly by (simultaneous) induction on $C$. 
\end{proof}

\NI In (\ref{lemma:29410:6}) above, note that the extended part in
$\MMM'$ on the top of $\MMM$ may refer to free labels of $\MMM$ but (since
$\MMM$ is a model) no labels in $\MMM$ can ever refer to (free or bound) labels
in $\MMM'$.

We are now ready to establish the standard axioms for equality.

\begin{lemma}\label{lemma:0248} {\rm (axioms for equality)}\ \ 
For any model $\MMM$ and $x$, $y$, $z$ and $C$:
\begin{enumerate}[\em(1)]

\item $\MMM \models x = x$, \ \ 
$\MMM \models x = y \IMPLIES y = x$\ \ 
and 
$\MMM \models (x = y \AND y = z) \IMPLIES x = z$.

\item $\MMM \models (C(x, y) \AND x = y) \IMPLIES C(x, x)$.

\end{enumerate}
where $C(x, y)$ 
indicates $C$ together with some of the occurrences of
$x$ and $y$, while
$C(x, x)$  is the result of substituting $x$ for $y$, i.e.~$C(x, y)\MSUBS{x}{y}$ 
see \cite[\S 2.4]{MENDELSON}.
\end{lemma}
\begin{proof}
For the first clause, reflexivity is because
$\expands{\MMM}{u}{x}\congmodel\expands{\MMM}{u}{x}$,
while symmetry and transitivity are from those of $\congmodel$. 
For the second clause, we proceed by induction on $C$.
We show the case where $C$ is $e_1 = e_2$. The case 
$C$ is $e_1 \reachable e_2$ is straightforward 
by definition. Other claims are by induction on $C$. 

It suffices to prove 
$\MMM \models x=y$ and $\MMM\models C$ imply
$\MMM \models C\MSUBS{x}{y}$.
\begin{align}
        \MMM \models x = y
                &\Rightarrow
        \expands{\expands{\MMM}{u}{x}}{v}{y}
                \congmodel
        \expands{\expands{\MMM}{u}{y}}{v}{x}
                \label{lemma:0248:1}
                \\
                &\Rightarrow
        \expands{\expands{\expands{\MMM}{u}{x}}{v}{y}}{w}{e_i}
                \congmodel
        \expands{\expands{\expands{\MMM}{u}{y}}{v}{x}}{w}{e_i}
                \label{lemma:0248:2}
\end{align}
Here (\ref{lemma:0248:1}) is by Lemma \ref{lemma:29410}.\ref{lemma:29410:1} and (\ref{lemma:0248:2}) follows from
the congruency of $\WB$.

\begin{align*}
        \expands{\expands{\expands{\MMM}{u}{x}}{v}{y}}{w}{e_i}
                &\congmodel
        \expands{\expands{\expands{\MMM}{u}{x}}{v}{y}}{w}{e_i\SUBST{v}{y}}
                \tag{\text{Lem.}~\ref{lemma:29410}(\ref{lemma:29410:3})}
                \\
                &\congmodel
        \expands{\expands{\expands{\MMM}{u}{y}}{v}{x}}{w}{e_i\SUBST{v}{y}}
                \tag{\ref{lemma:0248:1}}
                \\
                &\congmodel
        \expands{\expands{\expands{\MMM}{u}{y}}{v}{x}}{w}{e_i\SUBST{v}{x}\SUBST{v}{y}}
                \tag{\text{Lem.}~\ref{lemma:29410}(\ref{lemma:29410:3})}
                \\
                &\congmodel
        \expands{\expands{\expands{\MMM}{u}{y}}{v}{x}}{w}{e_i\SUBST{x}{x}\SUBST{x}{y}}
                \tag{Lem.~\ref{lemma:29410}(\ref{lemma:29410:4})}
                \\
                &\congmodel
        \expands{\expands{\expands{\MMM}{w}{e_i\SUBST{x}{x}\SUBST{x}{y}}}{u}{y}}{v}{x}
                \tag{Lem.~\ref{lemma:29410}(\ref{lemma:29410:5})}
\end{align*}

\begin{align*}
        \MMM \models e_1 = e_2
                &\Rightarrow
        \expands{\expands{\MMM}{u}{x}}{v}{y} \models e_1 = e_2
                \tag{Lem.~\ref{lemma:29410}(\ref{lemma:29410:6})}
                \\
                &\Rightarrow
        \expands{\expands{\expands{\MMM}{u}{x}}{v}{y}}{w}{e_1}
                \congmodel
        \expands{\expands{\expands{\MMM}{u}{x}}{v}{y}}{w}{e_2}
\end{align*}
Thus we get
\begin{align*}
        \expands{\expands{\expands{\MMM}{w}{e_1\SUBST{x}{x}\SUBST{x}{y}}}{u}{y}}{v}{x}
                &\congmodel
        \expands{\expands{\expands{\MMM}{u}{x}}{v}{y}}{w}{e_1}
                \\
                &\congmodel
        \expands{\expands{\expands{\MMM}{u}{x}}{v}{y}}{w}{e_2}
                \\
                &\congmodel
        \expands{\expands{\expands{\MMM}{w}{e_2\SUBST{x}{x}\SUBST{x}{y}}}{u}{y}}{v}{x}
\end{align*}
This allows to conclude to: 
\[
                \expands{\MMM}{w}{e_1\MSUBS{x}{x}\MSUBS{x}{y}}
                        \congmodel
                \expands{\MMM}{w}{e_2\MSUBS{x}{x}\MSUBS{x}{y}}
\]
which is equivalent to $\MMM \models C(x, x)$, 
as required.
\end{proof}

\subsection{Axioms for Necessity Operators}
\label{subsec:axiom:allworlds}
We list basic axioms for Necessity and Possibility Operators. 
Below recall that $\someworld C \LOGICEQ \neg (\allworlds \neg C)$.

\begin{proposition}[Necessity Operator]
\label{pro:necessity}\hfill
\begin{enumerate}[\em(1)]
\item $\allworlds (C_1\entails C_2) \ENTAILS \allworlds C_1\entails
\allworlds C_2$; 
$\allworlds C \ENTAILS C$; 
$\allworlds \allworlds C \LITEQ \allworlds C$; 
$C \entails \someworld C$. Hence $\allworlds C \entails
\someworld C$.
\item {\em (permutation and decomposition)}
\begin{enumerate}[\em(a)]
\item $\allworlds e_1 = e_2 \LITEQ e_1 = e_2$ 
and $\allworlds e_1 \not= e_2 \LITEQ e_1 \not= e_2$
if $e_i$ does not contain dereference. 
\item $\allworlds (C_1 \AND C_2) \LITEQ \allworlds C_1 \AND \allworlds C_2$. 

\item $\allworlds  C_1 \OR \allworlds C_2 \ENTAILS \allworlds (C_1 \OR C_2)$. 

\item $\allworlds \forall x.C \ENTAILS \forall x.\allworlds C$   
and $\allworlds \forall x. \allworlds C \LITEQ 
\allworlds \forall x.C$. 

\item $\exists x.\allworlds C \ENTAILS \allworlds \exists x.C$
and $\allworlds \exists x^\alpha.C \LITEQ \exists x^\alpha.\allworlds C$
with $\alpha\in \ASET{\UNIT,\BOOL,\NAT}$. 

\item 
$\allworlds \HIDEa x.C \LITEQ \HIDEa x.\allworlds C$ 
and 
$\HIDEe x.\allworlds C \ENTAILS \allworlds \HIDEe x.C$.

\item $\allworlds \exists\TVX.C \LITEQ \exists\TVX.\allworlds C$; and 
$\allworlds \forall\TVX.C \LITEQ \forall\TVX.\allworlds C$. 

\item 
$\allworlds \allCon{x}{C} \LITEQ \allCon{x}{\allworlds C} \LITEQ \allworlds C$
and  
$\someCon{x}{\allworlds C}\LITEQ \allworlds C
\ENTAILS \allworlds \someCon{x}{C}$.
\end{enumerate}
\end{enumerate}
\end{proposition}
\begin{proof}
See Appendix \ref{subsec:axioms:necessity}. 
\end{proof}
By the second axiom in (d), 
we can derive $\mathsf{fresh}\LITEQ \mathsf{fresh}_3$ in the last example 
of \S~\ref{sub:ex:local}.

The following proposition clarifies the interplay between
$\allworlds C$ and evaluation formulae, and is useful in 
many examples. Recall below that
$\APP{e}{e'}\Uparrow$ (defined in Notation \ref{con:assertions}) means the application leads to the
divergence.

\begin{proposition}[Perpetuity]
\label{pro:perpetuity}
With $z$ fresh, 
$
\allworlds C 
\; \LITEQ\;
\forall \TVX, \TVY.f^{\TVX\FS\TVY}.x^{\TVX}.
(f\bullet x\Downarrow
\ENTAILS \ONEEVAL{f}{x}{z}{\allworlds C})$.
Again with $z$ fresh, 
$
\allworlds C 
\; \LITEQ\;
\forall \TVX, \TVY.f^{\TVX\FS\TVY}.x^{\TVX}.
(f\bullet x\Downarrow
\ENTAILS \ONEEVAL{f}{x}{z}{C})
$.
\end{proposition}
\begin{proof}
  Throughout we use $\allworlds C \LITEQ \allworlds \allworlds C$.
  For the first equivalence suppose $\MMM\models \allworlds C$ and
  $\MMM[f:L][x:L'][z:fx]\converges \MMM'$.  Then step by step we reach
  $\MMM'\models \allworlds C$ by the definition of $\allworlds C$. For
  the other direction, suppose $\MMM\models \allworlds C$ and for all
  $N,N'$, we have $\MMM[u:N][w:N']\converges \MMM'$.  By assumption
  $\MMM[f:\lambda ().N][z:f()][w:N']\converges \MMM'$ such that
  $\MMM'\models C$ with $\MMM \EVOLVES \MMM'$.  Since $\MMM[u:N][w:N']
  \models C$, we have $\MMM[u:N] \models \allworlds C$, as required.
  For the second equivalence, the ``only-if'' direction is immediate,
  while the ``if'' direction is proved as in the previous ``if''
  direction, observing that we can combine an arbitrary number of
  applications into a single one.
\end{proof}

\NI The first logical equivalence of Proposition \ref{pro:perpetuity}
allows us to say that if
$\allworlds C$ holds and if a procedure is executed and if the
evaluation terminates then $\allworlds C$ (hence in particular $C$)
holds again. In essence, this is why a specification using $\allworlds C$
(or the equivalent) is useful: it allows us specify a behaviour which
holds regardless of execution of other procedures and resulting
state change. The second logical equivalence shows that, in addition,
we can in fact define $\allworlds C$ via evaluation formulae (which
in fact directly corresponds to the semantics of $\allworlds C$
in \S~\ref{sec:models}).

Next, the following proposition says that located assertions are
derived constructs, definable by combining non-located assertions and
content quantification.

\begin{proposition}[Decomposition of Located Evaluation Formula]
\label{pro:located_decompose}
\ \ \ $\ONEEVAL{x}{y}{z}{C}@\VEC{w}\ \LITEQ$\\ 
$\forall
u^{\UNIT\Rightarrow\UNIT}.(\APP{u}{()}\Downarrow\; \ENTAILS\;
\ONEEVAL{x}{y}{z}{C\ \AND\ \someCon{\VEC{w}}{\APP{u}{()}\Downarrow}})$
\end{proposition}
\begin{proof}
In the following discussions we consider $\VEC{w}$ to be a singleton $w$ for simplicity. 
First assume the left-hand side holds for a model say $\MMM$.
Then the application only  changes the content of $w$, hence
if $\APP{u}{()}\Downarrow$ then by restoring the content of $w$
we again have $\APP{u}{()}\Downarrow$. Secondly assume 
the right-hand side holds but the left-hand side does not.
Then there must be some $u$ which uses this difference at $w$
to change its diverging behaviour, hence a contradiction.
\end{proof}

\NI This decomposition uses content quantification to define located
evaluation formulae where the effect set is restricted to specified
finite locations. We can generalise located assertions to those 
which can specify the range of effects by formulae, which is
sometimes useful. Such formulae can also be decomposed in the same way
using an extended form of content quantification.

\subsection{Axioms for Hiding}
\label{subsec:axiom:hiding}
Next we list basic axioms for hiding quantifiers. 
The most convenient axiom is 
about the elimination of  hiding quantifiers, introduced 
by reference generation. To formulate this, 
we need some preparation. 

\begin{definition}[Monotone/Anti-Monotone Formulae]
\label{def:monotone}
$C$ is {\em monotone} if $\MMM\models C$ and
$l\not\in \FL{C}$ imply $(\new l)\MMM\models C$. 
$C$ is {\em anti-monotone} if $\neg C$ is monotone. 

\end{definition}
The proof of the following proposition is similar 
to Proposition \ref{pro:thin}.

\begin{proposition}[Syntactic Monotone/Antimonotone
    Formulae]\label{pro:mono}
\hfill 
\begin{enumerate}[\em(1)]
\item 
\label{pro:mono:base}
$T$, $F$, 
$e=e'$, $e\not=e'$, $\REACH{e}{e'}$ and 
${e}\noreach{e'}$ are monotone. 

\item 
\label{pro:mono:ih}
If $C,C'$ are monotone, 
then $C \AND C'$, $C \OR C'$, 
$\forall x^\alpha.C$ for all $\alpha$, 
$\exists x^\alpha.C$ with $\alpha\in \{ \UNIT,\BOOL,\NAT\}$, 
$\exists\TVX.C$, 
$\forall\TVX.C$, 
$\HIDEe x.C$, 
$\HIDEa x.C$, 
$\allworlds{C}$,  
$\allCon{x}{C}$, and  
$\ONEEVAL{e}{e'}{x}{C'}$ 
are monotone. 

\item
The conditions  exactly dual to 1 and 2 give antimonotone formulae.
\end{enumerate}
\end{proposition}



\begin{proposition}[Axioms for $\forall$, $\exists$ 
and $\HIDEe$] \label{pro:hiding}
Below we assume there is no capture of variables in types and formulae.
\begin{enumerate}[\em(1)]
\item {\rm (introduction)}
$C \ENTAILS \HIDEe x.C$ if $x\not\in\FV{C}$ 

\item {\rm (elimination)}
$\HIDEe x.C \LITEQ C$
if $x\not\in\FV{C}$ and $C$ is monotone. 

\item 
For any $C$ we have
$C \ENTAILS \exists x.C$.
Given $C$ such that $x\not\in \FV{C}$
and $C$ is thin with respect to $x$, we have
$\exists x.C \ENTAILS C$. 

\item 
For any $C$ we have $\forall x.C \ENTAILS C$.
For $C$ such that $x\not\in \FV{C}$ and $C$ is thin with respect to $x$,
we have $C \ENTAILS \forall x.C$. 


\item $\HIDEe x.(C_1 \AND C_2) \ENTAILS \HIDEe x. C_1 \AND \HIDEe x. C_2 $. 
\item $\HIDEe x.(C_1 \OR C_2) \LITEQ \HIDEe x.C_1 \OR \HIDEe x.C_2$. 

\item $\HIDEe y.\forall x.C \ENTAILS \forall x.\HIDEe y.C$   

\item $\exists x.\HIDEe y. C \ENTAILS \HIDEe y. \exists x.C$  
and $\exists x^\alpha.\HIDEe y. C \LITEQ \HIDEe y.\exists x^\alpha.C$
with $\alpha\in \ASET{\UNIT,\BOOL,\NAT}$. 

\item 
$\HIDEe y.\HIDEa x.C \ENTAILS \HIDEa x.\HIDEe y.C$; 
and 
$\HIDEe y.\HIDEe x.C \LITEQ \HIDEe x.\HIDEe y.C$.    

\item $\HIDEe y.\exists\TVX.C \LITEQ \exists\TVX.\HIDEe y.C$; and 
$\HIDEe y.\forall\TVX.C \ENTAILS \forall\TVX.\HIDEe y.C$. 
\label{strange}

\item 
$\HIDEe y.\allCon{x}{C} \ENTAILS \allCon{x}{\HIDEe y.C}$
and  
$\HIDEe y.\someCon{x}{C} \ENTAILS \someCon{x}{\HIDEe y. C}$
\end{enumerate}
\end{proposition}
\begin{proof}
See Appendix \ref{app:pro:hiding}. 
\end{proof}

\NI For (1) and (2), it is notable that we do {\em not} generally have
$C\ENTAILS \HIDEe x.C$ even if $C$ is thin.  Neither $\HIDEe x.C
\ENTAILS C$ with $x\not\in\FV{C}$ holds generally. 

For the counterexample of $C\ENTAILS \HIDEe x.C$ without the side
condition, let $\MMM\;\DEFEQ\; (\ASET{x:l,\ x':l},\; \ASET{l\mapsto
  5})$.  Then $ \MMM\models x=x' $ but we do {\em not} have
$\MMM\models \HIDEe y.y=x'$ since $l$ is certainly not hidden ($x$ is
renamed to fresh $y$ to avoid confusion).

For the counterexample of $\HIDEe x.C \ENTAILS C$ with
$x\not\in\FV{C}$, let $\MMM\;\DEFEQ\; (\new l)(\ASET{u:\lambda
  ().!l},\; \ASET{l\mapsto 5})$.  Then we have:
\[
\begin{array}{l}
\MMM\models \allworlds \APP{u}{()}=z\ASET{z=5}
\end{array}
\] 
Also we have: 
\[
\begin{array}{l}
\MMM\models (\new x)\someworld \APP{u}{()}=z\ASET{z=0}
\end{array}
\] 
with $\MMM[x:l]\models \someworld \APP{u}{()}=z\ASET{z=0}$.
If we apply $\HIDEe x.C \LITEQ C$ to the above formula, we have 
$\MMM\models \someworld \APP{u}{()}=z\ASET{z=0}$, 
which contradicts $\MMM\models 
\allworlds\ASET{\truth}\APP{u}{()}=z\ASET{z=5}$.   

Note this shows that integrating \mfb{these quantifiers with the standard
universal and existential quantifiers lets the latter loose their
standard axioms, motivating the introduction of the $\nu$-operator}: from
Proposition \ref{pro:hiding} (1,2,3), either $\exists x.C \ENTAILS
\HIDEe x.C$ or $\HIDEe x.C \ENTAILS \exists x.C$ (with $x$ typed by a
reference type) does not hold in general (if $x\not\in \FV{C}$ and $C$ is
thin, then $\exists x.C \ENTAILS C \ENTAILS \HIDEe x.C$; and 
if  $x\not\in \FV{C}$ and $C$ is monotone, then $\HIDEe x.C \ENTAILS C \ENTAILS
\exists x.C$). 
 
The content quantifiers  also have  useful axioms.  
Appendix \ref{subsec:axioms:cq} lists a selection.


\subsection{Axioms for Reachability}
\label{subsec:axiom:reachablity}
\NI We start from  axioms for reachability. Note that
our types include recursive types. 
\begin{proposition}[axioms for reachability] \label{pro:notreach}
The following assertions are valid.
\begin{enumerate}[\em(1)]
\item
{\em (1)} $x\reachable x$; \ \ 
{\em (2)} $x\reachable y\AND y\reachable z\;\ENTAILS\; x\reachable z$;
\item 
{\em (1)} $y \notreach x^{\alpha}$
\ with $\alpha\in\ASET{\UNIT,\NAT,\BOOL}$;
{\em (2)} $x\noreach y \THEN x \not = y$;
{\em (3)} $x\notreach w\AND w\reachable u\;\ENTAILS\;x\notreach u$.
\item
{\em (1)} $\PAIR{x_1}{x_2}\reachable y\, \;\LITEQ\;\, x_1\reachable y\OR
x_2\reachable y$;
{\em (2)} $\INJ{i}{x}\reachable y\; \LITEQ\; x\reachable y$;
{\em (3)} $x\reachable y^{\refType{\alpha}} \;\ENTAILS\; x\reachable !y$; 
{\em (4)} 
$x^{\refType{\alpha}}\!\reachable y \AND x\neq y \; \ENTAILS\;
!x\reachable y$; 
{\em (5)} $[!x]y\reach x\ \LITEQ\ y\reach x$ and $[!x]x\noreach y\ \LITEQ\ x\noreach y$.
\end{enumerate}
\end{proposition}
\begin{proof}
1, 2 and 3.(1--4) are direct from the definition (e.g.~for 3-(2) we
observe $l\in\FL{\INJ{i}{V}}$ iff $l\in\FL{V}$).  For 3-(5), suppose
$\MMM\models y\reach x$, and take
$\MMM'$ which only differs from $\MMM$ in the stored value at (the
reference denoted by) $x$. Since $\MMM\models y\reach x$ holds,
there is a shortest sequence of connected references from $y$ to $x$ 
which, by definition, does not include $x$ as its 
intermediate node. Hence this sequence also exists in $\MMM'$, 
i.e.~$\MMM'\models y\reach x$,
proving 
$[!x]y\reach x\ \LITEQ\ y\reach x$.
Similarly, we can prove $[!x]x\noreach y\ \LITEQ x\noreach y$. 
\end{proof}
\NI 
3-(5) says that altering the content of $x$ does not affect
reachability {\em to} $x$.
Note $[!x]y\noreach x\LITEQ y\noreach x$
is not valid at all.
3-(5) was already used for deriving $[\mbox{\emph{Inv-$\noreach$}}]$ in
\S \ref{subsec:rules} (notice that we cannot substitute $!x$ for $y$ in
$[!x]x\noreach y$ to avoid name capture \cite{ALIAS}).

Let us say $\alpha$ is {\em finite} if it does not contains an arrow
type or a type variable.  We say $\REACH{e}{e'}$ is {\em finite} if
$e$ has a finite type. 

\begin{theorem}[elimination] 
\label{pro:elimination}
\label{theorem:elimination}
Suppose all reachability predicates in $C$ are finite. 
Then there exists $C'$ such that $C\LITEQ C'$ and 
no reachability predicate occurs in $C'$. 
\end{theorem}
\begin{proof}
By Proposition \ref{pro:notreach}. See Appendix \ref{app:elimination}. 
\end{proof}

\NI The elimination of reachability predicates crucially uses type
information in logical terms: as a simple example consider
$\REACH{x}{y}$ where $x$ has type $\refType{\refType{\NAT}}$ and $y$
has type $\refType{\NAT}$.  Then we have $\REACH{x}{y} \LITEQ !!x=y$.
The precise inductive elimination rules are given in Appendix
\ref{app:elimination}.


For analysing reachability with function types, it is useful to 
define the following ``one-step'' reachability predicate.
Below $e_2$ is of a reference type.
\begin{equation}
\begin{array}{c}
\text{
\begin{tabular}{ll}
$\MMM \models e_1\directreach e_2$
\ \ \ \ if &
$\MAP{e_2}_{\xi, \sigma}\in \FL{\MAP{e_1}_{\xi, \sigma}}$
\ for each $(\new \VEC{l})(\xi, \sigma)\WB \MMM$ 
\end{tabular}
}
\end{array}
\end{equation}
The predicate $f\directreach l'$ 
means $l'$ occurs in any $\cong$-variant of the
program $f$. 

The following is straightforward from the definition. 
\begin{proposition}[Support]
$(\new \VEC{l})(\xi, \sigma) \models
x\directreach l'$ iff $l'\in\bigcap \ASET{\FL{V}\ | \ V\cong
\xi(x)}$. 
\end{proposition}
The latter says that $l'$ is in the support 
\cite{stark:namhof,GabbayM:newapprasib,PittsAM:nomlfo}
of $x$. 

We set $x\directreach^n y$ for $n\geq 0$ by:
\begin{eqnarray*}
  x\directreach^0 y & \LITEQ & x=y\\ 
  x\directreach^1 y & \LITEQ & x\directreach y\\ 
  x\directreach^{n+1} y & \LITEQ & \exists z.(x\directreach z \ANDl
  !z\directreach^n y)\quad (n\geq 1)
\end{eqnarray*}
By definition, 
we immediately observe:

\begin{proposition}\label{pro:decomposition}
$x\reach y\ \IFF\ \exists n.(x\directreach^n y)\LITEQ
(x=y\, \OR\, x\directreach y\, \OR\, \exists z.(x\directreach  z\AND z \neq y \AND
z\reach y))$.
\end{proposition}

\NI Proposition \ref{pro:decomposition}, combined with Theorem
\ref{theorem:elimination}, suggests that if we can clarify one-step
reachability at function types then we will be able to clarify the reachability
relation as a whole. Unfortunately this relation is inherently
intractable.

\begin{proposition}[undecidability of $\directreach$ and $\reach$]
\label{pro:notax}
{\rm (1)} $\MMM \models f^{\alpha\FS\beta}\directreach x$ 
is undecidable.\ \ \ 
{\rm (2)} $\MMM \models f^{\alpha\FS\beta}\reach x$ is undecidable.
\end{proposition}
\begin{proof}
For (1), let $V \DEFEQ \lambda ().\IFTHENELSE{M=()}{l}{\REF{0}}$ 
with a closed \PCFv-term $M$ of type $\UNIT$.
Then $f:V,\,x:l\models f\directreach x$ iff $M\converges$,
reducing the satisfiability to the halting problem of
\PCFv-terms.  For (2), take the same $V$ so that the type of $l$ and
$x$ is $\REF{\NAT}$ in which case $\directreach$ and $\reach$ coincide.
\end{proof}
%
The same result holds for call-by-value
$\beta\eta$-equality. 
Proposition \ref{pro:notax} indicates inherent intractability of $\directreach$
and $\reach$. 

However Proposition \ref{pro:notax} does not imply that we cannot obtain useful
axioms for (un)reacha\-bi\-li\-ty at function types. Next, we discuss a
collection of  axioms with function types.
First, the following axiom says that if $x$ is unreachable from $f$, $y$ and
$\VEC{w}$, then the application of $f$ to $y$ with the effect set
$\VEC{w}$ never exports $x$.

\begin{proposition}[unreachable functions] 
\label{pro:notreachfunc} For an arbitrary $C$, the following 
is valid with $i$ and $\TVX$ fresh: 
%
$$
\allworlds \ASET{C \AND x\noreach fy\VEC{w}}\APPs{f}{y}\!=\!z\ASET{C'}@\VEC{w}
\; \ENTAILS\; 
\allworlds \forall \TVX, i^{\TVXscript}.
 \ASET{C \AND x\noreach fiy\VEC{w}}\APPs{f}{y}\!=\!z
 \ASET{C'\AND x\noreach fiyz\VEC{w}}@\VEC{w}$$
\end{proposition}
\begin{proof} 
See Appendix \ref{app:notreachfunc}. 
\end{proof}



\subsection{Local Invariants}
\label{subsec:localinvariant}
\NI We now introduce an axiom for local invariants. 
Let us first consider a
function which writes to a local reference of 
 base type.  Even programs of this kind pose fundamental difficulties in
reasoning, as shown in \cite{MeyerAR:towfasflv}.  Take the following program:
\begin{equation}\label{ex:invprg}
\mathtt{compHide}\;\DEFEQ\;  \mathtt{let} \ 
x=\refPrg{7} \ \mathtt{in} \ \lambda y.(y>!x)
\end{equation}
The program behaves as a pure function $\lambda y.(y>7)$. Clearly,
 the obvious local invariant $!x=7$ is preserved. We demand this assertion to 
survive under arbitrary invocations of $\mathtt{compHide}$: thus (naming the
function $u$) we arrive at the following invariant:
\begin{equation}\label{ex:inv:a}
C_0\;\LOGICEQ\;\ !x=7\;\AND\;\allworlds\forall y.\ASET{!x=7}\APP{u}{y}=z\ASET{!x=7}@\emptyset
\end{equation}
Assertion (\ref{ex:inv:a}) says: (1) the invariant $!x=7$ holds now;
and that (2) once the invariant holds, it continues to hold for ever (note
$x$ can never be exported due to the type of $y$ and $z$, 
so that only $u$ will touch $x$). 
Using this assertion,
$\mathtt{compHide}$ satisfies the following
with $i$ fresh:
\begin{eqnarray}\label{ex:inv:c}
& \ASET{\truth}\mathtt{compHide}:_u\ASET{\new x.(x \noreach i^{\TVXscript}
 \ANDl C_0\ANDl C_1)} & \\
\label{ex:inv:ccc}
& C_1\LOGICEQ \allworlds \forall
y.\ASET{!x=7}\APP{u}{y}=z\ASET{z=(y>7)}@\emptyset. &
\end{eqnarray} 
Thus, noting $C_0$ is only about the content of $x$ (in fact
it is syntactically stateless except $x$ in the sense of
Definition \ref{def:syntactic:stateless}, 
we can conclude $C_0$ continues to hold automatically over
any future computation by any programs.
Hence we cancel
$C_0$ together with $x$:
\begin{equation}\label{ex:inv:e}
\ASET{\truth}\mathtt{compHide}:_u
\ASET{\allworlds\forall y.\APP{u}{y}=z\ASET{z=(y>7)}}
\end{equation}
which describes a purely functional behaviour. 

Now we leave the example and move to the general case,  
stipulating the underlying reasoning principle as an axiom. 
Let $y, z$ be fresh. We define:
\begin{equation}\label{ex:inv:FO} 
\LocalInv{u}{C_0}{\VEC{x}}
\ \;\LOGICEQ\;\ 
C_0 \ANDl (\allworlds\forall yi.\{C_0\}\APP{u}{y}\Downarrow\ENTAILS 
\allworlds\forall yi.\eval{C_0}{u}{y}{z}{C_0 \ANDl \VEC{x}\notreach z})
\end{equation}
where $C_0\ENTAILS \VEC{x}\notreach iy$.
$\LocalInv{u}{C_0}{x}$ says that currently $C_0$ holds; and
that if $C_0$ holds, applying $u$ to $y$ results in, if it
ever converges, $C_0$ again and the returned $z$ is disjoint from
$\VEC{x}$.
The axiom also uses:
\begin{equation} 
x\reachable^\bullet \VEC{y}
\;\LOGICEQ\;
\forall z.
(x\reachable z\ENTAILS z\in\ASET{\VEC{y}})
\end{equation} 
Thus $x\reachable^\bullet \VEC{y}$ says that 
all references reachable from $x$ are inside
$\ASET{\VEC{y}}$. We write $\VEC{x}\reachable^\bullet \VEC{y}$
for the conjunction $\AND_i x_i\reachable^\bullet \VEC{y}$.
The axiom follows.


\begin{proposition}[axiom for information hiding] 
\label{pro:localinv:firstorder}
\label{pro:localinv:higherorder}
\label{pro:localinv}
Assume 
$C_0\!\LITEQ C'_0\AND\VEC{x}\notreach iy\AND \VEC{g}\reachable^\bullet \VEC{x}$,\
$C'_0$ is stateless except $\VEC{x}$,\    
$C$ is antimonotone,
$C'$ is monotone,
$i,m$ are fresh and 
$\{\VEC{x},\VEC{g}\} \cap (\FV{C, C'}\cup\ASET{\VEC{w}})=\emptyset$.  
Then the following is valid: 
\[
\begin{array}{ll}
\text{\AIHax}\quad \quad 
&
\forall \TVX.\forall i^\TVX.
\APP{m}{()}\!=\!u\ASET{
(\new\VEC{x}.\exists \VEC{g}.E_1)\AND E} 
\; \ENTAILS \; 
\forall \TVX.\forall i^\TVX.
\APP{m}{()}\!=\!u\ASET{E_2\AND E}
\end{array}
\]
with
\begin{enumerate}[$\bullet$]
\item 
$E_1 \LOGICEQ \LocalInv{u}{C_0}{\VEC{x}}
\AND 
\allworlds \forall yi.\eval{C_0\ANDs C}{u}{y}{z}{C'}@\VEC{w}\VEC{x}$\ 
\item   
$E_2 \LOGICEQ 
\allworlds 
\forall y.\eval{C}{u}{y}{z}{C'}@\VEC{w}$ and 
\item 
$E$ is an arbitrary formula. 
\end{enumerate}
\end{proposition}
\begin{proof}
See Appendix \ref{app:proof:AHI}.
\end{proof}

\NI 
{\AIHax} is used 
with the refined consequence rule 
$[\mbox{\emph{Cons-Eval}}]$
(cf.~Figure \ref{figure:rules}) 
to simplify from $E_1$ to $E_2$, eliminating hidings.
Its validity is proved using Proposition \ref{pro:partition}. 
The axiom\footnote{We believe that the monotonicity of $C'$ and anti-monotonicity of $C$
are unnecessary in Proposition \ref{pro:localinv}, though the present proof uses them.} says: 
\begin{quote}
{\em if a function $u$ with  fresh reference $x_i$ is generated, and if
it has a local invariant $C_0$ on the content of $x_i$,  
then we can cancel $C_0$ together with $x_i$.} 
\end{quote}
Note that: 
\begin{enumerate}[$\bullet$]

\item 
  The statelessness of $C_0$ except $\VEC{x}$ 
ensures that satisfaction of $C_0$ is not affected by
state change except at $\VEC{x}$; and  


\item \mfb{The quantification $\exists \VEC{g}.E_1$ of $\VEC{g}$ in (\textsf{AIH})
allows the invariant to contain free variables,  }
extending applicability of
the axiom, for example in the presence of circular references 
as we shall use in \S \ref{sec:example}
for $\mathtt{safeEven}$. $\VEC{g}\reachable^\bullet\VEC{x}$ ensures
that $\VEC{g}$ are contained in the $\VEC{x}$-hidden part of the model.
\end{enumerate}
Coming back to $\mathtt{compHide}$, we take,
for {\AIHax}: 
\begin{enumerate}[(1)]
\item $C'_0$ to be $!x=7$ which is syntactically stateless except $x$;
\item $C_0$ to be $C'_0\AND x\noreach i$; 
\item $\VEC{s}$ and $\VEC{w}$ empty, 
\item both $C$ and $E$ to be $\truth$
(which is anti-monotonic by 
Proposition \ref{pro:mono}, 
and 
\item $C'$ to be $z=(y>7)$ (which is monotonic by 
the same proposition),
\end{enumerate}
thus arriving at the desired assertion.

{\AIHax} eliminates $\HIDEe$ from the post-condition based on local
invariants.  The following axiom also eliminates $\HIDEe x$, this time
solely based on freshness and disjointness of $x$.

\begin{proposition}[$\HIDEe$-elimination]\ 
\label{pro:nuelim}
Let $x\not\in \FV{C}$ and $m, i, \TVX$ be fresh. Then the following is valid:
\begin{equation}\label{nuelimformula}
\forall \TVX, i^\TVXscript.\:
\APP{m}{()}\!=\!u
\ASET{\HIDEe \VEC{x}.([!\VEC{x}]C\AND \VEC{x}\noreach ui^\TVXscript)}
\;\; \ENTAILS\;\; 
\APP{m}{()}\!=\!u
\ASET{C}
\end{equation}
\end{proposition}
\begin{proof}
See Appendix \ref{proof:pro:nuelim}.
\end{proof}


\NI This proposition says that if a hidden (and newly created) location $x$ in the post-state is
disjoint from any asserted data including
the used function itself and those in the pre-state, then we can
safely neglect it (in this sense it is a garbage collection rule
when we are not concerned with newly created variables).
The following axiom stipulates how an invariant can be {\em
transferred} by a function (caller) which uses another function (callee) when
the latter only affects a set of references unreachable
from the former.

\begin{proposition}[invariant by application]\label{pro:transfer}
\label{pro:localinv:transfer}
Assume $C_0$ is stateless except at $\VEC{x}$,
$C_0\ENTAILS \VEC{x}\noreach y$ and
$y\not\in\FV{C_0}$. Then the following is valid. 
\[
(\allworlds \forall y.\ASET{C_0}\APP{f}{y}\!=\!z\ASET{C_0}@\VEC{x}
\ANDl
\allworlds \ASET{C}\APP{g}{f}\!=\!z\ASET{C'})
\ \ENTAILS\
\allworlds \ASET{C\AND C_0\AND \VEC{x}\noreach g}
\APP{g}{f}=z
\ASET{
  C_0
  \AND 
  C'
}
\]
\end{proposition}
\begin{proof}
See Appendix \ref{proof:pro:transfer}. 
\end{proof}

\NI The axiom says that 
the result of applying a function $g$ disjoint from each local reference
$x_i$ in $\VEC{x}$, 
to the argument function $f$ which satisfies a local invariant 
exclusively at $\VEC{x}$, again
preserves that local invariant.  

Proposition \ref{pro:transfer} may be considered as a higher-order
version of Proposition \ref{pro:notreachfunc} and in fact is closely 
related in that both depend on localised effects of a function
at references.



\section{Reasoning Examples}
\label{sec:example}
\label{sec:examples:1}
\label{sec:example:1}

\NI This section demonstrates the usage of the proposed logic through
concrete reasoning examples. 

\subsection{New Reference Declaration}
\label{subsec:newvar}
We first show a useful derived rule given by the combination
of ``let'' and new reference generation. 
\[
        [\mbox{\em LetRef}]
                \,
        \infer
        {
                \ASSERT{C}{M}{m}{C_0}
                \quad 
                \ASSERT{C_0\MSUBS{!x}{m} \AND 
		  \NOTREACH{x}{\VEC{e}}}{N}{u}{C'}
                \quad 
                x \notin \PFN{\VEC{e}}
        }
        {
                \ASSERT{C}{\LET{x}{\REFPROG{M}}{N}}{u}{\nu x.C'}
        }
\]
where $C'$ is thin w.r.t. $m$.  
Above $\PFN{e}$ denotes the set of
\emph{free plain names} of $e$ which are reference names in $e$ 
that do not occur dereferenced, given in Definition \ref{def:plain}.
The meaning of  $\NOTREACH{x}{\VEC{e}}$ was given  in Notation
\ref{con:assertions} 
in \S~\ref{sub:ex:local}.
The rule reads:
\begin{quote}
{\em Assume {\rm (1)} executing $M$ with precondition $C$ leads to $C_0$,
with the resulting value named $m$; 
and {\rm (2)} running $N$ from 
$C_0$ with $m$ as the content of $x$ together with the assumption $x$ is unreachable from 
each $e_i$, leads to $C'$ with the resulting value named $u$.
Then running $\LET{x}{\REFPROG{M}}{N}$  from $C$
leads to $C'$ whose $x$ is fresh and hidden.} 
\end{quote}
The side condition $x\not\in\PFN{e_i}$ is essential for
consistency (e.g. without it, we could assume $x\noreach x$, 
i.e. $\falsity$); and $\new x.C'$ cannot be strengthened to $x\noreach
i \ANDl C'$ since
$N$ may store $x$ in an existing reference.
The use of general $\tilde{e}$ is also essential since the
we can start from total
disjointness (separation) and reach possibly partial disjointness
in the conclusion. For this purpose we need to have explicit
$x\noreach \tilde{e}$ initially, which may possibly be weakened in the
post-condition $C$ through the actions in $N$.

The rule directly gives a proof rule for
new reference declaration \cite{PittsAM:realvo,REYNOLDS,MeyerAR:towfasflv}, 
$\LETNEW{x}{M}{N}$, which has the same operational behaviour as 
$\LET{x}{\REFPROG{M}}{N}$. 

We can derive  $[\text{\em LetRef}]$ as follows. Below $i$ is fresh. 
{
\begin{DERIVATION}
     \LINE
     {1.\ }
     { 
     \ASSERT{C}{M}{m}{C_0} 
     }
     {
     \quad (premise)
     }
     \LINE
     {2.\ }
     { 
     \ASSERT{C_0\SUBST{!x}{m} \AND \NOTREACH{x}{\VEC{e}}}{N}{u}{C'}
     \quad\text{with}\quad x \notin \PFN{\VEC{e}}
     }
     {
     \quad (premise)
     }
     \LINE
     {3.\ }
     { 
                 \ASSERT
                 {C}
                 {\REFPROG{M}}
                 {x}
                 {
                 \nu {y}.(C_0\SUBST{!x}{m}\AND \NOTREACH{x}{i} \AND x = y)
                 }
     }
     {
     \quad (1,Ref)
     }
     \LINE
     {4.\ }
     { 
              \ASSERT
              {C}
              {\REFPROG{M}}
              {x}
              {
                 \nu {y}.(C_0\SUBST{!x}{m}\AND \NOTREACH{x}{\VEC{e}} \AND x = y)
              }
     }
     {
     (Subs $n$-times)
     }
     \LINE
     {5.\ }
     { 
                 \ASSERT
                 {C_0\SUBST{!x}{m}\AND \NOTREACH{x}{\VEC{e}} \AND x = y}
                 {N}
                 {u}
                 {
                 C'\AND x = y
                 }
     }
     {
     (2, Invariance)
     }
     \LINE
     {6.\ }
     { 
                \ASSERT
                 {C}
                 {\LET{x}{\REFPROG{M}}{N}}
                 {u}
                 {
                 \nu y.(C'\AND x = y)
                 }
     }
     {
     (4,5,LetOpen)
     }
     \LASTLINE
     {7.\ }
     { 
                \ASSERT
                 {C}
                 {\LET{x}{\REFPROG{M}}{N}}
                 {u}
                 {
                 \nu x.C'
                 }
     }
     {
     (Conseq)
     }
\end{DERIVATION}}
[{\em LetOpen}] is the rule for let to open the
scope: 
\[ 
\mbox{[{\em LetOpen}]}\ 
        \infer 
        {
                \ASSERT{C}{M}{x}{\nu \VEC{y}.C_0}@\VEC{e}_1\quad 
                \ASSERT{C_0}{N}{u}{C'}@\VEC{e}_2
        }
        {
                \ASSERT{C}{\LET{x}{M}{N}}{u}{\nu \VEC{y}.C'}@\VEC{e}_1\VEC{e}_2
        } 
\]
where $C'$ is thin w.r.t. $x$. 
[{\em LetOpen}] and [{\em Subs}] (both rules being for located judgements)
are found in Figure \ref{figure:rules:derived} in Appendix \ref{app:rules}, and
their soundness is proved in Appendix \ref{sub:letopen_subs}.

\subsection{Shared Stored Function}
\label{ex:incshared}
We present a
simple example of hiding-quantifiers and unreachability using
$\IncShared$ in (\ref{ex:IncPrgshared}) from
\S~\ref{sec:introduction}. 

\begin{equation*}\label{ex:IncPrgsharedAgain}
\IncShared
\;\DEFEQ\;
a\!:=\!\IncPrg;
b\!:=\!!a;
c_1\!:=\!(!a)();
c_2\!:=\!(!b)();
(!c_1+!c_2)
\end{equation*}
with $\IncPrg \DEFEQ \mathtt{let} \ x =\refPrg{0} \ \mathtt{in} \
\lambda ().(x:=!x+1;\,!x)$. 
Naming it $u$, the assertion
$\HIDEe x.\INCSPECb{u,x,n}$ (defined below) captures the behaviour of $\IncPrg$:
\begin{eqnarray*}
\INCSPEC{x,u} 
&\LOGICEQ &
\allworlds \forall j.\EVAL{!x = j}{u}{()}{j+1}{!x = j+1}@ x.\\ 
\INCSPECb{u,x,n} 
& \LOGICEQ &
!x = n \AND \INCSPEC{x,u}.
\end{eqnarray*}
The following derivation for $\IncShared$ sheds light on how shared
higher-order local state can be transparently reasoned in the present
logic.  For brevity we
work with the implicit global assumption that $a,b,c_1,c_2$ are
pairwise distinct and safely omit an anchor from judgements when
the return value is of unit type.
\begin{DERIVATION}
        \LINE
        {1.\ }
        {
        \ASSERT{\TRUTH}{\IncPrg}{u}{\nu x.\INCSPECb{u,x,0}}
        }
	{}
        \LINE
        {2.\ }
        {
        \SHORTASSERT{\TRUTH}{a := \IncPrg}{\nu x.\INCSPECb{!a,x,0}}{}
        }
        {(1, Assign)}
        \LINE
        {3.\ }
        {
                \SHORTASSERT
                {
                \INCSPECb{!a, x, 0}
                }
                {b := !a}
                {\INCSPECb{!a, x, 0} \AND \INCSPECb{!b, x, 0}}
                {}
        }
        {(Assign)}
        \LINE
        {4.\ }
        {
                \SHORTASSERT
                {
                        \INCSPECb{!a, x, 0}
                }
                {c_1 := (!a)()}
                {\INCSPECb{!a, x, 1} \AND !c_1 = 1}
                {}
        }
        {(Assign)}
        \LINE
        {5.\ }
        {
                \SHORTASSERT
                {
                \INCSPECb{!b, x, 1}
                }
                {c_2 := (!b)()}
                {\INCSPECb{!b, x, 2} \AND !c_2 = 2}
                {}
        }
	{(App etc.)}
        \LINE
        {6.\ }
        {
        \ASSERT{!c_1 = 1 \AND !c_2 = 2}{(!c_1) + (!c_2)}{u}{u= 3}
	}
        {(Deref etc.)} 
        \LINE
        {7.\ }
        {
      \ASSERT{\TRUTH}{\IncShared}{u}
       {\nu x .u = 3}
        }
       {(2--6, LetOpen)}
        \LASTLINE
        {8.\ }
        {
	  \ASSERT{\TRUTH}{\IncShared}{u}
		 {u = 3}
        }
        {(Conseq)}
\end{DERIVATION}
Line 1 is by [{\em LetRef}]. 
Line 8 uses 
Proposition \ref{pro:hiding}(2), 
$\HIDEe x.C\ENTAILS C$.

To shed light on how the difference in sharing
is captured in inferences, 
we 
list the inference for a program which assigns
{\em distinct} copies of $\IncPrg$ to $a$ and $b$,  
\begin{equation*}\label{ex:IncPrgUnshared}
\IncUnshared
\;\DEFEQ\;
 a\!:=\!\IncPrg;
 b\!:=\!\IncPrg;
 c_1\!:=\!(!a)();
 c_2\!:=\!(!b)();
 (!c_1+!c_2)
\end{equation*}
This program assigns to $a$ and $b$ two
separate instances of $\IncPrg$.  This lack of sharing 
between $a$ and $b$ 
in $\IncUnshared$ 
is captured by the following derivation:
 \begin{DERIVATION}
         \LINE
         {1.}
         {
        \ASSERT{\TRUTH}{\IncPrg}{m}{\nu x.\INCSPEC{u,x,0}}
        }{}
        \LINE
         {2.}
         {
        \SHORTASSERT{\TRUTH}{a := \IncPrg}{\nu x.\INCSPEC{!a,x,0}}{}
        }
         {}
        \LINE
         {3.}
         {
         \SHORTASSERT{\INCSPEC{!a,x,0}}{b := \IncPrg}
         {\nu y.\INCSPECb{0,0}}{}
         }
         {}
        \LINE
         {4.}
         {
                 \SHORTASSERT
                 {
                   \INCSPECb{0, 0}
                 }
                 {z_1 := (!a)()}
                 {\INCSPECb{1,0}  \AND !z_1 = 1}
                 {}
         }
         {}
        \LINE
         {5.}
         {
                 \SHORTASSERT
                 {
                         \INCSPECb{1,0}
                 }
                 {z_2 := (!b)()}
                 {\INCSPECb{1,1}  \AND !z_2 = 1}
                 {}
         }
         {}
         \LINE
         {6.}
         {
   \ASSERT{!z_1 = 1 \AND !z_2 = 1}{(!z_1) + (!z_2)}{u}{u = 2}{}
         }{}
        \LINE
         {7.}
         {
         \ASSERT{\TRUTH}{\IncUnshared}{u}{\nu xy. u = 2}
         }{}
         \LASTLINE
         {8.}
         {
         \ASSERT{\TRUTH}{\IncUnshared}{u}{u = 2}
         }{}
 \end{DERIVATION}
Above 
$\INCSPECb{n, m}\LOGICEQ \INCSPEC{!a, x,  n}\AND
\INCSPEC{!b, y, m}\AND x\neq y$.  Note $x\not = y$ is
guaranteed by [{\em LetRef}].  This is in contrast to the 
derivation for $\IncShared$, 
where, in Line 3, $x$ is automatically shared after
``$b:=!a$'' which leads to scope extrusion.

\subsection{Memoised Factorial} 
\label{sub:factorial}
\NI 
Next we treat the memoised factorial 
(\ref{memfact}) (from \cite{PittsAM:operfl}) in the introduction. 
\[
\begin{array}{l}
\mathtt{memFact}
\;\DEFEQ \;
\mathtt{let}
\ a =\refPrg{0}, \ 
\ b =\refPrg{1} 
\ \mathtt{in}\\ 
\qquad\; \; \;
\lambda x.\IFTHENELSE{x=!a}{!b}{(a:=x\,;\,b:=\mathtt{fact}(x)\,;\,!b)}
\end{array}
\]
Above $\mathtt{fact}$ is the standard factorial function. 

Our target assertion specifies the behaviour of a pure factorial.
\[ 
Fact(u)\LOGICEQ
\allworlds \forall x.\APP{u}{x}=y\ASET{y=\fAct{x}}@\emptyset.
\] 
The following inference starts from the $\mathtt{let}$-body of 
$\mathtt{memFact}$, 
which we name $V$. We set:\ 
\[
\begin{array}{lll}
E_{1a} & \LOGICEQ & 
\allworlds \forall xi.\ASET{C_0}\APP{u}{x}\!=\!y\ASET{C_0\AND ab\noreach y}@ab\\
E_{1b} & \LOGICEQ &  
\allworlds \forall xi.\ASET{C_0\AND C}\APP{u}{x}\!=\!y\ASET{C'}@ab
\end{array}
\]
and we set 
$C_0$
to be  
$ab\noreach ix \ANDl !b\!=\!\fAct{(!a)}$, 
$C$ to be $\truth$, and $C'$ to be $y=\fAct{x}$. 
Note that $!b\!=\!\fAct{(!a)}$ is stateless except $ab$ 
by Proposition~\ref{pro:notreach}(5);
and that, by the type of $x$ and $y$ being $\NAT$ 
and Proposition \ref{pro:notreach} 2-(1),
we have $ab\noreach x\LITEQ ab\noreach y\LITEQ \truth$.

We can now reason:
\begin{DERIVATION}
        \LINE
        {1.}
        {
        \ASSERT{\truth}{0}{a}{a=0}@\emptyset
        }
        {
        (Const)  
        }
        \LINE
        {2.}
        {
        \ASSERT{a=0}{1}{b}{b=\fAct{a}}@\emptyset  
        }
        {
        (Const)  
        }
        \LINE
        {3.}
        {
          \ASSERT{\TRUTH}{V}{u}
		 {
		   \allworlds \forall xi.
		   \ASET{C_0}
		   \APP{u}{x}\!=\!y
 		   \ASET{C_0\ANDl C'}
		 }@\emptyset
	}
	{(Abs)}
        \LINE
        {4.}
        {
	  \ASSERT
	      {\truth}
	      {V}
	      {u}
	      { 
		E_{1a}
		\AND
		E_{1b}		   
	      }@\emptyset
	}
	{(3, Conseq)}
        \LINE
        {5.}
        {
	  \ASSERT
	      {ab\noreach i \AND !b\!=\!\fAct{(!a)}}
	      {V}
	      {u}
	      {
ab\noreach i
\AND !b\!=\!\fAct{(!a)} 
\AND 
		E_{1a}
		\AND
		E_{1b}		   
	      }@\emptyset
	}
	{(4, Inv-$\noreach$, Inv-Val in \S~\ref{sub:invariant})}
        \LINE
        {6.}
        {
	  \ASSERT{\TRUTH}{\mathtt{memFact}}{u}
		 {
		   \new ab.
		   (
                   C_0 \AND 
		   E_{1a}
		   \AND
		   E_{1b}
		   )
		 }@\emptyset
	}
	{(1,2,4, LetRef in \S~\ref{subsec:newvar})}
        \LINE
        {7.}
        {
	  \APP{m\!}{\!()}\!\!=\!\!u
	  \ASET{\new ab.(C_0 \AND E_{1a}\!\AND E_{1b}\!)}		   
	  \ENTAILS
	  \APP{m\!}{\!()}\!\!=\!\!u
	  \ASET{Fact(u)}
	}
	{($\star$)}
        \LASTLINE
        {8.}
        {
	  \ASSERT{\TRUTH}{\mathtt{memFact}}{u}
		 {Fact(u)}@\emptyset
	}
	{(6,7,ConsEval)}
\end{DERIVATION}
Line 4 uses 
the axiom 
$\ASET{C}\APP{f}{x}\!\!=\!\!y\ASET{C_1\ANDs C_2}@\VEC{w}
\ENTAILS
\AND_{i=1,2}\ASET{C}\APP{f}{x}=y\ASET{C_i}@\VEC{w}$ (in \cite{ALIAS}).
Line 7 uses {\sf (AIH)}. 

\subsection{Information Hiding (2): Stored Circular Procedures}
We next consider 
stored higher order functions which mimic stored procedures. 

We start with a simple one,  
$\mathtt{circFact}$ from \cite{GLOBAL},  
which uses a self-recursive higher-order local store.
\[
\begin{array}{rcl}
\mathtt{circFact}
& \DEFEQ &
x\,:=\,\lambda z.\IFTHENELSE{z=0}{1}{z \times (!x)(z-1)}\\
\mathtt{safeFact}
& \DEFEQ  &
\mathtt{let} \ x = \refPrg{\lambda y.y} 
\ \mathtt{in} \ (\mathtt{circFact};!x)
\end{array}
\]
In \cite{GLOBAL}, we have derived the following judgement.
\begin{eqnarray}
\label{der:circfact}
\ASET{\truth}
\mathtt{circFact}:_u\,
\ASET
{
CircFact(u, x)
}@x
\end{eqnarray}
where 
$$
CircFact(u, x) 
\;\LOGICEQ\;
\allworlds 
\forall n.\ASET{!x=u}\APP{!x}{n}=z\ASET{z=\factorial{n}\AND !x=u}@\emptyset
\ANDl
\, !x=u\ \ 
$$  
which says:
\begin{quote}
{\em After executing
the program, $x$ stores a procedure which would calculate
a factorial if $x$ stores that behaviour, and that
$x$ does store the behaviour.}  
\end{quote}
We now show
$\mathtt{safeFact}$  
named $u$ satisfies $Fact(u)$.
Below we use:
\[
\begin{array}{rcl}
CF_a &  \LOGICEQ & \allworlds \forall n.\ASET{!x=u}\APP{!x}{n}=z\ASET{!x=u}@\emptyset\\
CF_b & \LOGICEQ & \allworlds \forall n.\ASET{!x=u}\APP{!x}{n}=z\ASET{z=\fAct{n}}@\emptyset
\end{array}
\]
(note that $x\noreach z \LITEQ \truth$ and $x\noreach n \LITEQ \truth$
by Proposition \ref{pro:notreach} (2)-1). 
\begin{DERIVATION}
        \LINE
        {1.}
        {
        \ASET{\truth}
        \lambda y.y:_m 
        \ASET{\truth}@\emptyset
        }{}
        \LINE
        {2.}
        {
        \ASET{\truth}
        \mathtt{circFact}\,;\,!x:_u\,
        \ASET
        {
        CircFact(u, x)
        }@x
        }
        {}
        \LINE
        {3.}
        {
        \ASET{\truth}
        \mathtt{circFact}\,;\,!x:_u\,
        \ASET
        {
        !x=u \ANDl CF_a\ANDl CF_b
        }@x
        }
        {(2, Conseq)}
        \LINE
        {4.}
        {
          \ASET{\NOTREACH{x}{i}}
          \mathtt{circFact}\,;\,!x:_u\,
          \ASET
              {
		\NOTREACH{x}{i}\AND !x=u \AND CF_a\AND CF_b
              }@x
	}
        {(3, Inv-$\noreach$)}
        \LINE
        {5.}
        {
        \ASET{\truth}
        \mathtt{safeFact}:_u
        \ASET{\nu {x}.(C_0 \AND CF_a\AND CF_b)}@\emptyset
        }
        {(4, LetRef)}
        \LINE
        {6.}
        {
	  \APP{m\!}{\!()}\!\!=\!\!u
	  \ASET{\new x.(C_0 \AND CF_{a}\!\AND CF_{b}\!)}		   
	  \ENTAILS
	   \APP{m\!}{\!()}\!\!=\!\!u
	  \ASET{Fact(u)}
	}
	{($\star$)}
        \LASTLINE
        {7.}
        {
        \ASET{\truth}
        \mathtt{safeFact}:_u
        \ASET{Fact(u)}@\emptyset
        }
        {(5, 6, ConsEval)}
\end{DERIVATION}
Line 1 is immediate. Line 2 is (\ref{der:circfact}). 
Line 6, $(\star)$ is by 
{\AIHax}, Proposition \ref{pro:localinv:firstorder}, 
setting 
$C_0 \LOGICEQ x\noreach i \ANDl !x=u$, 
$C\LOGICEQ E\LOGICEQ \truth$ and $C'\LOGICEQ  y=\fAct{x}$. 

\subsection{Mutually Recursive Stored Functions}
Now we investigate the program from (\ref{safeOdd}) in the introduction. 
The reasoning easily extends to programs which use multiple locally
stored, and mutually recursive, procedures.

We first verify 
the following $\mathtt{mutualParity}$ (the $\mathtt{let}$-body). 
\begin{equation}
\begin{array}{rcl}
\mathtt{mutualParity}  \DEFEQ
&
x  :=  \lambda n.\IFTHENELSE{n\!=\!0}{\mathtt{f}}{\mathtt{not}((!y)(n\!-\!1))};\\
&
y  :=  \lambda n.\IFTHENELSE{n\!=\!0}{\mathtt{t}}{\mathtt{not}((!x)(n\!-\!1))}
\end{array}
\end{equation}
Then we have:  
\begin{equation}\label{mutualCheck}
\ASET{\truth}
\mathtt{mutualParity}:_u\
\ASET
{
  \exists gh.IsOddEven(gh,!x!y,xy,n)
}
\end{equation}
where, with
$Even(n) \LITEQ \exists x.(n\!=\!2\times x)$
and  
$Odd(n) \LITEQ Even(n\!+\!1)$:
\[
\begin{array}{lll}
IsOddEven(gh,wu,xy,n)  \LOGICEQ   
(IsOdd(w, gh,n,xy)\;\AND\; IsEven(u,gh,n,xy)
\;\AND\; !x=g \;\AND\; !y=h)\\[1mm]
IsOdd(u,gh,n,xy)  \LOGICEQ  
\allworlds \EVALEFFECT{!x = g\ANDl !y = h}{u}{n}{z}{z=Odd(n)
\;\AND\; !x=g \;\AND\; !y=h }{xy}\\[1mm]
IsEven(u,gh,n,xy)  \LOGICEQ  
\allworlds \EVALEFFECT{!x = g\ANDl !y = h}{u}{n}{z}{z=Even(n)
\;\AND\; !x=g \;\AND\; !y=h
}{xy}
\end{array}
\]
The detailed derivations are given in Appendix \ref{app:mutualfact}. 
Above $IsOdd(u,gh,n,xy)$ says that
\begin{quote}
{\em $!x$  and $!y$ remain unchanged, 
and that $u$ checks if its argument is odd.}
\end{quote}
Similarly for $IsEven(u,gh,n,xy)$. 
Then above $IsOddEven(gh,wu,xy,n)$ says that 
\begin{quote}
{\em $x$ stores a procedure which 
checks if its argument is odd if $y$ stores a procedure which 
does the dual, and $x$ does store the behaviour; and 
dually for $y$}.
\end{quote}
Note that $IsOdd$ and $IsEven$, the effect set is $xy$ since 
$x$ and $y$ are free and assigned to the abstractions in $\mathtt{mutualParity}$. 

Our aim is to derive the judgement 
for $\mathtt{safeEven}$ given below: 
\begin{align}
\mathtt{safeEven} \ & \ \DEFEQ  & 
\mathtt{let} \ x= \refPrg{\lambda n.\mathtt{t}}, 
\ y = \refPrg{\lambda n.\mathtt{t}} 
\ \mathtt{in} \ (\mathtt{mutualParity};!y)
\end{align}

We start from (\ref{mutualCheck})
(the case for
$\mathtt{safeOdd}$ is symmetric).
\begin{eqnarray*}\label{safeMutualEven}
\ASET{\truth}
\mathtt{safeEven} :_u\,
\ASET
{
\forall n.
\allworlds \ONEEVAL{u}{n}{z}{z=Even(n)}@{\emptyset}
}
\end{eqnarray*}

We first identify the local invariant:
\[
C_0\; \LOGICEQ\;\ !x=g \ANDl\ !y=h \ANDl IsEven(h,gh,n,xy) \ANDl
xy\noreach ijn \ANDl
gh\reachable^\circ xy
\] 
Note we have a free variable $h$.
Since $C_0$ only talks about $g$, $h$ and the content of $x$ and $y$,
we know 
$!x=g \ANDl\ !y=h \ANDl IsEven(h,gh,n,xy)$ 
is stateless except $x,y$; and 
$xy\noreach n\LITEQ xy\noreach z \LITEQ \truth$ by Proposition \ref{pro:notreach} (2)-1. 

Let us define: 
\[
\begin{array}{rcl}
ValEven(u) & \LOGICEQ & 
\allworlds \forall n.\EVALEFFECT{\truth}{u}{n}{z}{z=Even(n)}{\emptyset}\\
Even_a 
&\LOGICEQ&
\allworlds \forall n.\ASET{C_0}\APP{u}{n}\!=\!{z}\ASET{C_0}@xy\\
Even_b 
&\LOGICEQ&\allworlds 
\forall n.\ASET{C_0}\APP{u}{n}\!=\!{z}\ASET{z\!=\!Even(n)}@xy
\end{array}
\]
The derivation is given as follows. 
\begin{DERIVATION}
        \LINEs
        {1.}
        {
        \ASET{\truth}
        \lambda n.\mathtt{t}:_m 
        \ASET{\truth}@\emptyset
        }{}
        \LINEs
        {2.}
        {
        \ASET{\truth}
        \mathtt{mutualParity}\,;\,!y:_u\,
       \ASET
        {
        \exists gh.IsOddEven(gh,gu,xy,n)
        }@xy
        }
        {}
        \LINEs
        {3.}
        {
        \ASET{\truth}
        \mathtt{mutualParity}\,;\,!y:_u
          \{\exists gh.(!x=g \AND !y=h \AND IsOdd(g,gh,n,xy) \AND Even_a \ANDl Even_b)\}@xy
        }
        {}
        \LINEs
        {4.}
        {
        \ASET{\NOTREACH{xy}{ij}}
        \mathtt{mutualParity}\,;\,!y:_u
    \ASET
        {
        \exists gh.(C_0 \ANDl 
Even_a \ANDl Even_b)
        }@xy
        }
        {}
        \LINEs
        {5.}
        {
        \ASET{\truth}
        \mathtt{safeEven}:_u
      \ASET{\new xy.\exists gh.
        (C_0 \ANDl  
Even_a \ANDl Even_b)
        }
        @\emptyset
        }
        {}
        \LINEs{6.}
        {
	  \ASET{\truth}
	  \APP{m\!}{\!()}\!\!=\!\!u
	  \ASET{\new xy\exists gh.(C_0\ANDl 
           Even_{a}\!\AND Even_{b}\!)}
	  \ENTAILS
	  \ASET{\truth}
	  \APP{m\!}{\!()}\!\!=\!\!u
	  \ASET{ValEven(u)}
	  \qquad
	  \text{(by {\AIHax})}
	}
	{}
        \LASTLINEs
        {7.}
        {
        \ASET{\truth}
        \mathtt{safeEven}:_u
        \ASET{ValEven(u)}
        @\emptyset
        }
        {}
\end{DERIVATION}
As we can see, the derivation follows the same pattern as that of 
$\mathtt{memoFact}$ and $\mathtt{safeFact}$. 

\subsection{Higher-Order Invariant}
\label{sub:profile}
We move to a program 
(from \cite[p.104]{stark:namhof}) 
whose invariant behaviour depends on another function. 
\martinb{The program instruments a program with 
simple profiling, counting the number of invocations. }
\[
\begin{array}{rcl}
\mathtt{profile}
&\quad\DEFEQ &\quad
\mathtt{let} \ x=\refPrg{0} \ \mathtt{in} \ \lambda y.(x:=!x+1; f y)\\
\end{array}
\]
Since $x$ is never exposed, this program should behave precisely as
$f$.  Thus our aim is to derive:
\begin{equation}
\label{profilejudgement}
\ASET{\allworlds \forall y.\ASET{C}\APP{f}{y}=z\ASET{C'}@\VEC{w}}\ 
\mathtt{profile}:_u\,\ASET{\allworlds\forall y.\ASET{C}\APP{u}{y}=z\ASET{C'}@\VEC{w}}
\end{equation} 
with $x\not\in\fv{C, C'}$
(by the bound name condition) 
and arbitrary anti-monotonic $C$ and monotonic $C'$. 

This judgement says: 
\begin{quote}
{\em if $f$ satisfies the specification
$E\LOGICEQ \allworlds \forall y.\ASET{C}\APP{f}{y}=z\ASET{C'}@\VEC{w}$, 
then $\mathtt{profile}$ satisfies the same specification $E$}.
\end{quote}
To derive
(\ref{profilejudgement}),
we first set 
$C_0$, 
the invariant, to be 
$x\noreach fiy\VEC{w}$. 

As with the previous derivations, we use two subderivations. 

First we derive:
\[
\begin{array}{llllll}
& E   & \LOGICEQ & 
\allworlds \forall y.\ASET{C}\APP{f}{y}=z\ASET{C'}@\VEC{w}\\
\ENTAILS & E_0 & \LOGICEQ & 
\allworlds \forall yi.\ASET{C\ANDl x \noreach fiy\VEC{w}}\APP{f}{y}\!=\!z
   \ASET{C'}@\VEC{w}x
          & \mbox{Axiom {\sf (e8)} in \cite{GLOBAL}}\\
\ENTAILS & E_1  & \LOGICEQ & \allworlds 
\forall yi.\ASET{C\ANDl x \noreach fiy\VEC{w}}\APP{f}{y}=z
       \ASET{C'\ANDl \NOTREACH{x}{zfiy\VEC{w}}}@\VEC{w}x
\quad \quad & \mbox{Axiom {\sf (e8)} in \cite{GLOBAL}}
\end{array}
\]
where Axiom {\sf (e8)} in \cite{GLOBAL} is given as: 
\[ 
(C \ENTAILS C_0 \ANDl \EVAL{C_0}{x}{y}{z}{C_0'} 
\ANDl C_0' \ENTAILS C)
\quad \ENTAILS  \EVAL{C}{x}{y}{z}{C'} 
\]
we use the 
first axiom in Proposition \ref{pro:necessity} (1). 
We also let 
\[
E_2   \LOGICEQ  \allworlds \forall yi.\ASET{\allCon{x}{C}\ANDl C_0}
\APP{f}{y}=z
       \ASET{C'\ANDl C_0}@\VEC{w}x
\]
The inference follows.
\begin{DERIVATION}
        \LINE
        {1.}
        {
        \ASET{\truth}
	x:=!x+1
        \ASET{\truth}@x
        }
	{(Assign)}
        \LINE
        {2.}
        {
        \ASET{[!x]C\AND E\AND x\noreach fiy\VEC{w}}
	\ x:=!x+1\ 
        \ASET{C\AND E\AND x\noreach fiy\VEC{w}}@x
	\qquad	\qquad
	}
	{(Inv-$\noreach$, Conseq)}
        \LINE
        {3.}
        {
        \ASET{C\AND E\AND C_0}
	\ fy:_z\,
          \ASET{C'\AND C_0}@\VEC{w}x
	}
	{(App, Conseq)}
        \LINE
        {4.}
        {
          \ASET{[!x]C\AND E\AND C_0}
	  \, x:=x+1;fy\, :_z\,
          \ASET{C'\AND C_0}@x\VEC{w}
	}
        {(2, 3, Seq)}
        \LINE
        {5.}
        {
        \ASET{E}\
	\lambda y. (x:=x+1;fy):_u
	\,
        \ASET
	    {
	      E_2
	    }
        @ \emptyset 
        }
        {(4, Abs, Inv)}
        \LINE 
        {6.}
        {
        \ASET{E}\
	\lambda y. (x:=x+1;fy):_u
	\,
        \ASET
	    {
	    \LocalInv{u}{C_0}{x}  
	    }
        @ \emptyset 
        }
        {(Similar to 1-5 from $E_2$)}
        \LINE
        {7.}
        {	
        \ASET{E}
	\mathtt{profile}
        \ASET{\new x.(\LocalInv{u}{C_0}{x}\ANDl E_2)}	    
        @ \emptyset 
	}
	{(5, 6, LetRef in \S~\ref{subsec:newvar})}
        \LINE
        {8.}
        {	
	 \APP{m}{()}=u\ASET{\new x.(\LocalInv{u}{C_0}{\VEC{x}}\ANDl E_2)}
	  \ENTAILS
	 \APP{m}{()}=u\ASET{E}
	}
	{($\star$)}
        \LASTLINE
        {9.}
        {
        \ASET{E}
        \mathtt{profile}:_u
        \ASET{E}@\emptyset
        }
        {(7, 8, ConsEval)}
\end{DERIVATION}
Above in Line 2, we note  
$E$ is tame (because of $\allworlds$) 
and equivalent to $\allCon{x}{E}$, hence   
[{\em Inv}] becomes applicable.
Line 8 is inferred by Proposition \ref{pro:localinv:firstorder}. 

\subsection{Nested Local Invariant from [34,27] }
\label{sub:GF}
The next example 
uses a function with local state as an argument to another function.
Let $\DIVPROG\DEFEQ \mu f.\lambda ().(f())$. Below $even(n)$ tests for evenness of $n$. 
\[
\begin{array}{lll}
\mathtt{MeyerSieber}
& \ \DEFEQ \ & 
\mathtt{let}\ x=\, \REFPROG{0}\ 
\mathtt{in}\ 
 \mathtt{let}\ f=\, \lambda ().x:=!x+2\\
& & \quad\quad \quad\quad \mathtt{in}\ 
(g\, f\ ;\ \IFTHENELSE{even(!x)}{()}{\DIVPROG()})
\end{array}
\]
Note
$\DIVPROG()$ immediately diverges. 
Since $x$
is local, and because $g$ will have no way to access $x$ except by
calling $f$, the local invariant that $x$ stores an even number is
maintained. Hence $\mathtt{MeyerSieber}$  satisfies the judgement:
\begin{equation}\label{MSassert}
\ASET{E \AND C}\;\mathtt{MeyerSieber}\;\ASET{C'}
\end{equation}
where, with $x, m\not\in\FV{C, C'}$:
\[ 
E\LOGICEQ\forall f.
(
\allworlds \APP{f}{()}\ASET{\truth}@\emptyset
\ \ENTAILS \ 
\allworlds \ASET{C}\APP{g}{f} \ASET{C'}
)
\]
(anchors of type $\UNIT$ are omitted following Notation \ref{con:assertions}(6)). 
The judgement (\ref{MSassert}) says that: 
\begin{quote}
{\em if feeding $g$ with 
a total and effect-free $f$ always 
satisfies $\ASET{C}\APP{g}{f}\ASET{C'}$, then 
$\mathtt{MeyerSieber}$ starting from $C$ also
terminates with the final state $C'$}. 
\end{quote}
Note such $f$ behaves as $\SKIP$.

For the derivation of (\ref{MSassert}), from the axiom for reachability
in Proposition \ref{pro:transfer}, 
we can derive
$E\ENTAILS E'$
where 
\[ E'\LOGICEQ
\forall f.
(
\allworlds \APP{f}{()}\ASET{\truth}@x
\ \ENTAILS \ 
\allworlds \ASET{[!x]C\AND x\noreach g}\APP{g}{f} \ASET{[!x]C'}
)
\]
Further
$\lambda ().x:=!x+2$ named $f$ satisfies
both:
\[
A_1\LOGICEQ \allworlds \ASET{\truth}\APP{f}{()}\ASET{\truth}@x
\quad\mbox{and}
\quad 
A_2\LOGICEQ \allworlds \ASET{Even(!x)}\APP{f}{()}\ASET{Even(!x)}@x
\]
Then from $A_1$ and $E'$, we obtain
$A'_1 \LOGICEQ \allworlds \ASET{[!x]C\AND x\noreach g}\APP{g}{f} \ASET{[!x]C'}$.

Using Proposition~\ref{pro:transfer}, $A'_1$ and $A_2$ we obtain:  
\[
\ASET{Even(!x) \AND [!x]C \AND E \AND  x\noreach gi}
M \ASET{[!x]C' \AND  x\noreach i}
\]
with 
$M\DEFEQ \mathtt{let}\ f=\, \lambda ().x:=!x+2
\ \mathtt{in}\ 
(gf\ ;\ \IFTHENELSE{even(!x)}{()}{\DIVPROG()})$. 

We then apply a variant of [{\em LetRef}] 
(replacing $C_0\MSUBS{!x}{m}$ in the premise of
[{\em LetRef}] in \S \ref{subsec:rules}
with $[!x]C_0 \AND\, !x=m$) to obtain 
\[ 
\ASET{E \AND C}
\ \mathtt{MeyerSieber}\ 
\ASET{\HIDEe x.([!x]C' \AND  x\noreach i)}
\]
Finally by Proposition~\ref{pro:nuelim} (noting the returned value has a
base type, cf.~Proposition \ref{pro:notreach} 2-(1)),
we reach
$
\ASET{E \AND C}
\ \mathtt{MeyerSieber}\ 
\ASET{C'}
$.
The full derivation is given in 
Appendix \ref{app:MS}.

\subsection{Information Hiding (5): Object} 
\label{sub:object}
As  final example, we treat information hiding for a
program with state, a small object encoded in imperative higher-order
functions, taken from \cite{KoutovasWand06}
(cf.~\cite{ComparingObjectEncodings,PierceTurner:Object,PierceBC:typsysfpl}).
The following program generates a simple object each time
it is invoked. 
\begin{eqnarray*}
\mathtt{cellGen}
&
\DEFEQ
&
\lambda z.
\left(\!\!
\!\!\! 
\begin{array}{l}
\ 
\mathtt{let}\ x_{0,1}=\, \REFPROG{z}\ \mathtt{in}\
\mathtt{let}\ y=\, \REFPROG{0}\ \mathtt{in}\!\!\\
\left(
\begin{array}{l}
\lambda ().\IFTHENELSE{even(!y)}{!x_0}{!x_1},\\
\lambda w.(y:=!y+1\ ;\ x_{0, 1}:=w)
\end{array}
\right)\!\!
\end{array}
\right)\!\!\!
\end{eqnarray*}
The object has a getter and a setter method. Instead of having one local
variable, it uses two with the same content, of which one is read at
each odd-turn of the ``read'' requests, another at each even-turn.
When writing, it writes the same value to both. Since having two
variables in this way does not differ from having only one
observationally, we expect the following judgement to hold
$\mathtt{cellGen}$:
\begin{equation}\label{cellgen}
\ASET{\truth}\ \mathtt{cellGen}:_u\,\ASET{CellGen(u)}
\end{equation}
where we set:
\begin{eqnarray*}
CellGen(u) 
& \LOGICEQ&
\allworlds \forall zi.\APP{u}{z}=o\ASET{\new x.(Cell(o, x)\AND !x=z \ANDl 
o\noreach i \ANDl x=o)}@\emptyset\\
Cell(o, x)
& \LOGICEQ&
\allworlds \forall v.\ASET{!x=v}\APP{\pi_1(o)}{()}=z\ASET{z=v \ANDl !x=v}@\emptyset\;\AND
\allworlds \forall w.\APP{\pi_2(o)}{w}\ASET{!x=w}@x
\end{eqnarray*}
$Cell(o, x)$ says that 
$\pi_1(o)$, the getter of $o$, 
returns the content of a
local variable $x$; and 
$\pi_2(o)$, the setter of $o$, writes 
the received value to $x$. 
Then $CellGen(u)$ says that, when $u$ 
is invoked with a value, say $z$, 
an object is returned with its initial fresh local state initialised
to $z$.
Note both specifications only mention a single local variable.
A straightforward derivation of
(\ref{cellgen}) 
uses $!x_0=!x_1$ as the
invariant to erase $x_1$: then we $\alpha$-converts $x_0$ to $x$ to
obtain  the required assertion $Cell(o, x)$.
See Appendix \ref{app:object} for full inferences.

\section{Extension, Related Work and Future Topics}
\label{sec:conclusion}
\NI For lack of space, detailed comparisons with existing program
logics and reasoning methods, in particular with Clarke's
impossibility result, Spatial Logic
\cite{CairesL:spalogcI} (which contain a hiding quantifier used in a
concurrency setting),
as well as other logics such as 
LCF, Dynamic logic,
higher-order logic and specification logic 
are left to 
our past papers
\cite{SHORT1,HY04PPDP,GLOBAL,ALIAS}.  Below we focus on directly
related work that treats locality and freshness in higher-order
languages.

\subsection{Three Completeness Results}

\label{subsec:complete}
We discuss  completeness properties of the proposed logic.  A
strong completeness property called {\em descriptive completeness} is
studied in \cite{completeness}. Descriptive completeness means that
characteristic assertions are provable for each program (i.e.~an
assertion characterising a program's behaviour uniquely up to
observational congruence).  We have shown \cite{completeness} that
this property implies two other completeness properties in our base
logic, {\em relative completeness} (which says that provability and
validity of judgements coincide, i.e.~completeness relative to an
oracle which can decide the validity of formulae in the assertion
language) and {\em observational completeness} (which says that
validity precisely characterises the standard contextual equivalence).

For lack of space, we only state the latter, which
we regard as a basic semantic property of the logic.




Write $\cong$ for the standard contextual congruence for programs
\cite{PierceBC:typsysfpl}; further write $M_1\congLogic M_2$ to mean
($\models \ASET{C}M_1:_u\ASET{C'}$ iff 
$\models \ASET{C}M_2:_u\ASET{C'}$). We have:

\begin{theorem}[observational completeness]\ \label{thm:obscon} For
  each $\Gamma;\Delta\proves M_{i}:\alpha$ ($i=1,2$), we have
  $M_1\cong_{\CAL L}M_2$ iff $M_1\cong M_2$.
\end{theorem}
\NI The proofs of descriptive, observational and relative completeness
follow \cite{completeness} and are detailed in
\cite{descriptivecompleteness}.


\subsection{Local Variables in Hoare Logic}
To our knowledge, Hoare and Wirth \cite{PASCAL} are the first to
present a rule for local variable declaration. In
our notation, their rule is written as follows.
\[
        [\mbox{\em Hoare-Wirth}]
                \,
        \infer
        {
                \ASET{C\AND x\neq\VEC{y}}\, P\, \ASET{C'}
		\;\
		x\!\not\in\!\fv{C'}\!\cup\!\ASET{\VEC{y}}
        }
        {
                \ASET{C\MSUBS{e}{!x}}\ \LETNEW{x}{e}{P}\ \ASET{C'}
        }
\]
Because this rule assumes references are never exported  beyond
their original scope, there is no need to have $x$ in $C'$. Since
aliasing is not permitted in \cite{PASCAL} either, we can also
dispense with $x\neq\VEC{y}$ in the premise.  $[\mbox{\em LetRef}]$ in
\S~\ref{ex:incshared} differs from [\mbox{\em Hoare-Wirth}] in that the former
can treat aliased references, higher-order procedures and new
references generation extruded beyond their original scope.
[\mbox{\em Hoare-Wirth}] is derivable from $[\mbox{\em LetRef}]$,
$[\mbox{\em Assign}]$ and $\HIDEe$-elimination in Prop. \ref{pro:nuelim}.

Among the studies of verification methods for ML-like languages
\cite{MilnerR:defostaML,CAML}, {\em Extended ML}
\cite{Sannella-Tarlecki85b} is a formal development framework for
Standard ML. {A specification is given by combining module
signatures and algebraic axioms.}
Correctness of an implementation w.r.t.~a
specification is verified by incremental syntactic
transformations. {\em Larch/ML} \cite{Wing-Rollins-Zaremski92} is a
design proposal of a Larch-based interface language for ML.
Integration of typing and interface specification is the main focus of
the proposal in \cite{Wing-Rollins-Zaremski92}. These two works do not
(aim to) offer a program logic with compositional proof rules; nor do
either of these works treat specifications for functions with
dynamically generated references.

\subsection{Related Work and Future Topics}

\paragraph*{\bf Reasoning Principles for Functions with Local State.}
There is a long tradition of studying equivalences over higher-order
programs with local state.  Meyer and Sieber \cite{MeyerAR:towfasflv}
present examples and reasoning principles based on denotational
semantics. Mason, Talcott and others \cite{MT92a,MT92b,HMST} 
investigate equational
axioms for an untyped version of the language treated in the present
paper, including local invariance.  Pitts and Stark
\cite{stark:namhof,PittsAM:realvo,PittsAM:operfl} present powerful
operational reasoning principles for the same ML-fragment considered
here, including reasoning principle for local invariance at
higher-order types \cite{PittsAM:operfl}. 
Our axioms for information hiding in \S~\ref{sec:axioms}, which
capture a basic pattern of programming with local state, are
closely related with these reasoning principles.  Our logic
differs in that its aim is to offer a method for describing and
validating properties of programs beyond program equivalence.  
Equational and logical approaches are complimentary:
Theorem \ref{thm:obscon} 
offers a basis for integration.  For example, we may consider
deriving a property of the optimised version $M'$ of $M$: if we can
easily verify $\ASET{C}M:_u\ASET{C'}$ and if we know $M\cong M'$, we
can conclude $\ASET{C}M':_u\ASET{C'}$, which is useful if $M$ is
better structured than $M'$.  

\paragraph*{\bf Separation Logic.}
The approach by Reynolds et al.~\cite{REYNOLDS} 
represents fresh data generation by relative spatial disjointness from
the original datum, using a sub-structural separating conjunction. This method
captures a significant part of  program properties. 
The proposed logic represents freshness as temporal disjointness through
generic (un)reachability from arbitrary data in the initial
state.  
The presented approach enables uniform treatment of known
data types in verification, including product, sum, reference, 
closure, etc., through the use of anchors,
which matches the observational semantics precisely:
we have examined this point through
several examples, including
objects from \cite{KoutovasWand06}, 
circular lists from \cite{LahiriQadeer06}, 
and tree-, dag- and graph-copy from \cite{SPACE04}.  
These results will be reported in future
expositions.
Reynolds \cite{REYNOLDS}
criticises the use of reachability for describing data structures,
taking  in-place reversal of a linear list as an example. 
Following  \S~\ref{sec:examples:1},
tractable reasoning is
possible for such examples using reachability combined with
$\mbox{[{\em Inv}]}$ and located assertions, see \cite{localexample}. 


Birkedal et al.~\cite{Birkedal2005a} present a ``separation logic
typing'' for a variant of Idealised Algol where types are constructed
from formulae of disjunction-free separation logic.  The typing system
uses  subtyping calculated via categorical semantics, 
the focus of  their study.  The work \cite{Birkedal2005b} extends 
 separation logic with higher-order predicates 
(higher-order frame rule), and demonstrates
how the extension helps modular reasoning about priority queues. Both
works consider neither exportable fresh reference generation nor
higher-order/stored procedures in full generality, so
it would be difficult to represent assertions and validate the
examples  in \S~\ref{sec:examples:1}. 
Examining the  
use of higher-order predicate abstraction in the present logic
 is an interesting future topic.

\paragraph*{\bf Other Hoare Logics.}
Names have been used in Hoare logic since early work by Kowaltowski
\cite{KowaltowskT:axiatseagj}, and are found in the work by von Oheimb
\cite{OheimbD:hoalfjiih}, Leavens and Baker \cite{LEAVENS:BAKER} and
Abadi and Leino \cite{ABADILEINO}, for treating parameter passing and
return values. These works do not treat higher-order procedures and
data types, which are uniformly captured in the present logic along
with parameters and return values through the use of
names. This generality comes from the fact that a large class
of program behaviour is faithfully represented as name passing
processes which interact at names: our assertion language offers a
concise way to describe such interactive behaviour in a logical
framework.

Nanevski et al.~\cite{ALEKS,ALEKS07} study Hoare Type Theory (HTT) which
combines dependent types and Hoare triples with anchors based on
a monadic understanding of computation.  HTT aims to provide an
effective general framework 
which unifies
standard static checking techniques with logical verification. Their
system emphasises the clean separation between static
validation and assertions. In their later work \cite{ALEKS07},
the integration of programs and specifications
in HTT is further pursued by introducing local state. 
Because of their basis in type theory, 
one interesting aspect is that their ``Hoare Triple''
of the form 
``$\ASET{P} x:A \ASET{Q}$''
is in fact a {\em type} 
and that $A$ can contain an arbitrary complex specification. 
Note that the use of type theory does prohibit potentially useful
assertions about circular data structures and references
(this is called a ``smallness'' condition).
The use of monad 
in their logic poses a question whether if we equip the underlying
programming language with monad what reasoning principles we may
obtain as a refinement of the present program logic.

Reus and Streicher
\cite{ReusS05} present a Hoare logic for a simple language
with higher-order stored procedures, extended in
\cite{ReusSchwinghammer}, 
with primitives for the
dynamic allocation and de-allocation of references.
Soundness is proved
with denotational methods, but completeness
is not proved. Their assertions contain quoted
programs, which is necessary to handle recursion via stored functions.  
Their language does not allow procedure parameters and 
general reference creation. 

No work mentioned in this section studies local invariance in
the context of program logics.


\paragraph*{\bf Dynamic and Evaluation Logics.}
\label{par:evaluation} 
Dynamic Logic \cite{DL}, introduced by Pratt \cite{PRATT} and studied
by Harel and others \cite{HAREL}, uses programs and predicates on them
as part of formulae, facilitating detailed specification of various
programs properties such as (non-)termination, or more intensional
features. As far as we know, higher-order procedures and local state
have not been treated in Dynamic Logic, even though we believe part of
the proposed method to treat higher-order functions would work
consistently in their framework.

Evaluation Logic, introduced by Pitts \cite{pitts90evaluation} and
studied by Moggi \cite{Moggi94,Moggi95}, is a typed logic
for higher-order programs based on the metalanguage for computational
monads which permits statements about the evaluation of programs to
values using evaluation modalities. Recently Mossa\-kowski et. al
\cite{MSG08} studied a generic framework for reasoning about purity
\cite{Naumann05} and effects based on a monad-based dynamic logic
which is similar to Evaluation Logic. Evaluation logic is closely
related to the present logic in that it is based on the decomposition
of semantic points into values and computation and that it captures
applications as part of the logic even though the approach of 
Evaluation Logic is based on denotations. Evaluation Logic has
uniformity in that it does not use separate judgements such as Hoare
triples. Evaluation Logic also includes expressions involving
applications as part of terms. Thus its assertion language already
includes  judgements for  programs.

The logic studied in the present paper distinguishes formulae for
evaluation in the logical language (evaluation formulae) from
judgements for  programs (pre/post conditions together with an anchor).
This is motivated by our wish to have the assertion language separate (independent) from
programs, which we believe to fit such engineering purposes as
design-by-contract (where one wishes to have abstract description of
behaviour {\em before} we construct programs). This aspect of the
present logic is closely related with its compositionality: we wish to
build assertions for a program from those for its subprograms, and if
one of its subprograms, say $M$, allows the same assertion as another
program, say $M'$, then we can {\em replace} $M$ by $M'$ and still
obtain the same assertion for the whole program. 
Separating the assertion language from programs is also vital
for verification of multi-language programs.
 We believe that it
is a meaningful topic to explore a uniform treatment of both
assertions for evaluations and judgements for programs, while keeping
the key features of the present logic.

\paragraph*{\bf Meta-Logical Study on Freshness.}
Freshness of names has recently been studied from the viewpoint of
formalising binding relations in programming languages and
computational calculi. Pitts and Gabbay
\cite{GabbayM:newapprasib,PittsAM:nomlfo} extend first-order logic 
with constructs to reason about freshness of names based on
the theory of permutations.  The key syntactic additions are the
(inter-definable) ``fresh'' quantifier $\FRESHQUANT$ and the freshness
predicate $\#$, mediated by a swapping (finite permutation)
predicate. Miller and Tiu \cite{MillerD:prothefgj} are motivated by the
significance of generic (or eigen-)~variables and quantifiers at the
level of both formulae and sequents, and split universal
quantification in two, introduce a self-dual freshness quantifier
$\nabla$ and develop the corresponding sequent calculus of Generic
Judgements. While these logics are not program logics, their logical
machinery may well be usable in the present context. As noted in
Proposition \ref{pro:decomposition}, reasoning about $\reach$ or
$\noreach$ is tantamount to reasoning about $\directreach$, which
denotes the support (the semantic notion of freely occurring
locations) of a datum/program.  A characterisation of support by
the swapping operation may lead to deeper understanding of
reachability axiomatisations.



\paragraph{\bf Acknowledgement}
We deeply thank the anonymous reviewers for valuable comments 
and suggestions. 
We thank Andrew Pitts for his comments on an early version of 
this paper. The example of the mutual recursion in \S~\ref{sec:example:1}
was given by Bernhard Reus. We thank him for his e-mail discussions 
on this example. 
This work is partially supported by EPSRC
GR/T04236, GR/S55545, GR/S55538, GR/T04724, GR/T03208, GR/T03258,  
and IST-2005-015905 MOBIUS.

\bibliographystyle{plain}
\bibliography{YHB-final}

\appendix
\section{Typing Rules}
\label{app:typing}
The typing rules are standard \cite{PierceBC:typsysfpl}, and listed
in Figure \ref{figure:typing:rules} for reference (we list only two first-order
operations).  
We take the equi-isomorphic approach \cite{PierceBC:typsysfpl}
for recursive types. 
In the first rule of Figure
\ref{figure:typing:rules}, $\mathtt{c}^C$ indicates a constant
$\mathtt{c}$ has a base type $C$.

We also list the typing rules for terms and formulae in Figure
\ref{fig:typingrules:formulae}.

\begin{myfigure}\small
\vspace{-3.6mm}
\[
\begin{array}{c}
        [\mbox{\em Var}]\,
        \infer
        {-}
        {\Gamma, x:\alpha\proves x:\alpha}
        \quad
        [\mbox{\em Label}]\,
        \infer
        {-}
        {\Gamma\cdot l:\alpha\proves l:\alpha}
        \quad
        [\mbox{\em Constant}]
                \,
                \infer
       {-}
       {\Gamma\proves \mathtt{c}:C}
       \\[5mm]
       [\mbox{\em Add}] 
       \,
       \infer
       {\Gamma\proves M_{1,2}:\NAT}
       {\Gamma\proves M_1\!+\!M_2:\NAT}
       \quad
       [\mbox{\em Eq}] 
       \,
       \infer
       {\Gamma\proves M_{1,2}:\NAT}
       {\Gamma\proves M_1\!=\!M_2:\BOOL}
       \\[5mm]
        \quad
        [\mbox{\em If}\,]
                \,
                \infer
        {
                \Gamma\proves M:\BOOL
                          \quad
                \Gamma\proves N_i:\alpha_i\ (i=1,2)
        }
        {
          \Gamma\proves 
          \IFTHENELSE{M}{N_1}{N_2}
                :\alpha
        }
        \\[5mm]
        [\mbox{\em Abs}]
                \,
                \infer
        {\Gamma,\ATb{x}{\alpha}\,\proves M:\beta}
        {\Gamma\proves 
        \lambda x^\alpha.M:\alpha\FS\beta}
        \quad
        [\mbox{\em App}]
                \,
                \infer
        {
          \Gamma \proves M:\alpha \FS \beta
          \quad
          \Gamma\proves N:\alpha
        }
        {
          \Gamma \proves MN:\beta
        }
        \\[5mm]
        [\mbox{\em Rec}]
                \;
                \infer
        {
          \Gamma, \ATb{x}{\alpha\FS\beta}\proves \lambda y^\alpha.M : \alpha\FS\beta
        }
        {
          \Gamma  \proves \mu x^{\alpha\FS\beta}.\lambda y^\alpha.M: \alpha\FS\beta
        }
\quad \quad 
        [\mbox{\em Iso}]
                \,
                \infer
        {\Gamma\,\proves M:\alpha\quad \alpha \WB \beta}
        {\Gamma\proves M:\beta}
        \\[5mm]
        [\mbox{\emph{Deref}}]
                \,
                \infer
        {
                \Gamma \proves M : \refType{\alpha}
        }
        {
                \Gamma \proves !M : \alpha
              }
        \quad
        [\mbox{\emph{Assign}}]
        \,
        \infer
        {
          \Gamma \proves M : \refType{\alpha}
          \quad
          \Gamma \proves N : \alpha
        }
        {
          \Gamma \proves M := N : \UNIT
        }
        \\[5mm]
        [\mbox{\em Ref}]
        \ONEPREMISERULE
        {
                \TYPES{\Gamma}{V}{\alpha} 
        }
        {
                \TYPES{\Gamma}{\refPrg{V}}{\refType{\alpha}}
        }
        \quad
        [\mbox{\em New}]
        \TWOPREMISERULE
        {
                \TYPES{\Gamma}{M}{\alpha}
        }
        {
                \TYPES{\Gamma, x : \REF{\alpha}}{N}{\beta}
        }
        {
                \TYPES{\Gamma}{\NEW{x}{M}{N}}{\beta}
        }
        \\[5mm]        
        [\mbox{\em Inj}]
                \,
                \infer
        {\Gamma \proves M:\alpha_i}
        {\Gamma \proves \inj{i}{M}:\alpha_1\!+\! \alpha_2}
        \quad
        [\mbox{\em Case}]
                \,
                \infer
        {\Gamma \proves M : \alpha_1\! +\! \alpha_2 \quad
        \Gamma,\ATb{x_i}{\alpha_i}\,\proves N_i : \beta
        }
        {
              \Gamma \proves 
                \CASE{M}{x_i^{\alpha_i}}{N_i}:\beta
        }
                \\[5mm]
        [\mbox{\em Pair}]
                \,
                \infer
        {\Gamma\proves M_i:\alpha_i\; (i=1,2)}
        {
               \Gamma\proves 
                {\ENCan{M_1, M_2}:\alpha_1\!\times\!\alpha_2}
        }
                \quad 
        [\mbox{\em Proj}]
                \,
                \infer
        {\Gamma\proves M : \alpha_1 \times \alpha_2}
        {
                \Gamma\proves 
                {\pi_i(M):\alpha_i\; (i=1,2)}
        }
                \\[5mm]
\end{array}
\]
\caption{Typing Rules}\label{figure:typing:rules}
\end{myfigure}

\begin{myfigure}
{\small
\[
\begin{array}{c}
\infer
{
  -
}
{
  \Gamma \proves x:\Gamma(x)
}
\quad
\infer
{
  -
}
{
  \Gamma \proves \mathsf{n}:\NAT
}
\quad
\infer
{
  -
}
{
  \Gamma\proves \mathsf{t,f}:\BOOL
}
\quad
\infer
{
  -
}
{
  \Gamma \proves l:\Gamma(l)
}
\quad
\infer
{
  \Gamma \proves e:\BOOL
}
{
  \Gamma \proves \neg e:\BOOL
}
\\[6.4mm]
\infer
{
  \Gamma \proves e_i:\alpha_i
}
{
  \Gamma \proves \ENCan{e_1, \ e_2}:\alpha_1\times\alpha_2
}
\quad
\infer
{
  \Gamma \proves e:\alpha_i
}
{
  \Gamma \proves \LOGICINJ{\alpha_1+\alpha_2}{i}{e}:\alpha_1+\alpha_2
}
\quad
\infer
{
  \Gamma \proves e:\refType{\alpha}
}
{
  \Gamma \proves !e:\alpha
}
\\[6.4mm]
\infer
{
  \Gamma  \proves e_{i}:\alpha_i
}
{
  \Gamma  \proves e_1 = e_2
}
\quad
\infer
{
  \Gamma  \proves {C}
}
{
  \Gamma  \proves \neg{C}
}
\quad
\infer
{
  \Gamma  \proves {C}_{1,2}
}
{
  \Gamma  \proves {C}_1 \star {C}_2
}
\star\in\ASET{\AND, \OR, \supset}
\quad
\infer
{
  \Gamma \cdot \ATb{x}{\alpha}  \proves {C}
}
{
  \Gamma \proves \QQQ x^\alpha.{C}
}\ \QQQ \in \{\forall, \exists\}
\\[6.4mm]
\infer
{
  \Gamma \cdot \ATb{x}{\REF{\alpha}{}}  \proves {C}
}
{
  \Gamma \proves \QQQ x.{C}
}\ \QQQ \in \{\nu, \OL{\nu}\}
\quad 
\ONEPREMISERULE
{
  \Gamma \proves C 
}
{
  \Gamma \proves \QQQ\TVX. C
}
\ \QQQ \in \{\forall, \exists\}
\quad 
\TWOPREMISERULE
{
  \Gamma \proves e : \REF{\alpha}{}
}
{
  \Gamma \proves C
}
{
  \Gamma \proves \allCon{e}{C}
}
\quad 
\TWOPREMISERULE
{
  \Gamma \proves e : \REF{\alpha}{}
}
{
  \Gamma \proves C
}
{
  \Gamma \proves \someCon{e}{C}
}
\\[6.4mm]
\infer
{
  \Gamma 
  \proves e_1 : \alpha \FS \beta
  \quad 
  \Gamma  \proves e_2 : \alpha
  \quad  
  \Gamma \cdot z : {\beta}  \proves C
}
{
  \Gamma \proves \ONEEVAL{e_1}{e_2}{z}{C}
}
\quad 
\ONEPREMISERULE
{
  \Gamma \proves C 
}
{
  \Gamma \proves \allworlds C
}
\quad 
\ONEPREMISERULE
{
  \Gamma \proves C 
}
{
  \Gamma \proves \someworld C
}
\\[6.4mm]
\TWOPREMISERULE
{
  \Gamma \proves e : \alpha
}
{
  \Gamma \proves e' : \REF{\beta}{}
}
{
  \Gamma \proves \REACH{e}{e'}
}
\quad 
\TWOPREMISERULE
{
  \Gamma \proves e : \REF{\alpha}{}
}
{
  \Gamma \proves e' : \beta
}
{
  \Gamma \proves         e\noreach e'
}
\end{array}
\]
}
\caption{Typing rules for terms and formulae}\label{fig:typingrules:formulae}
\end{myfigure}

\vfill\eject
\section{Proof Rules}
\label{app:rules}
\begin{myfigure}
{\small
\vspace{-3.2mm}
\[
\begin{array}{c}
        [\mbox{\em Var}]
        \,
        \infer
        {-}
        {
          \ASSERT{C\MSUBS{x}{u}}{x}{u}{C}@\emptyset
        }
        \quad
        [\mbox{\em Const}] 
        \,
        \infer
        {-}
        {
          \ASSERT{C\MSUBS{\mathsf{c}}{u}}{\mathtt{c}}{u}{C}@\emptyset
        }
                \\[4.8mm]
        [\mbox{\em Add}]
                \,
                \infer
        {
            \ASET{C}\,M_1:_{m_1}\ASET{C_0}@\VEC{e}_1
            \quad
            \ASET{C_0}\,M_2:_{m_2}\ASET{\,C'\MSUBS{m_1+m_2}{u}\,}@\VEC{e}_2
        }
        {
          \ASET{C}\, M_1+M_2\, :_u\,\ASET{C'}@\VEC{e}_1\VEC{e}_2
        }
                \\[4.8mm]
        [\mbox{\em In}_1]
                \,
        \infer
        {
                \ASSERTSURFACE{C}{M}{v}{C'\MSUBS{\LOGICINJ{}{1}{v}}{u}}{\VEC{e}}
        }
        {
                \ASSERTSURFACE{C}{\INJ{1}{M}}{u}{C'}{\VEC{e}}
        }
                \\[4.8mm]
        [\mbox{\emph{Case}}]
                \,
        \infer
        {
          \ASSERTSURFACE{C^{\minus \VEC{x}}}{M}{m}{C_0^{\minus \VEC{x}}}{\VEC{e}_1}
          \quad
          \ASSERTSURFACE{C_0\MSUBS{\LOGICINJ{}{i}{x_i}}{m}}{M_i}{u}{{C'\,}^{\minus \VEC{x}}}{\VEC{e}_2} 
        }
        {
          \ASSERTSURFACE{C}{\CASE{M}{x_i}{M_i}}{u}{C'}{{\VEC{e}_1\VEC{e}_2}}
        }    
        \\[4mm]
        [\mbox{\em Proj}_1]
                \,
        \infer
        {
                \ASSERTSURFACE{C}{M}{m}{C'\MSUBS{\pi_1(m)}{u}}{{\VEC{e}}}
        }
        {
                \ASSERTSURFACE{C}{\pi_1(M)}{u}{C'}{\VEC{e}}
        }    
        \\[4mm]
        [\mbox{\em Pair}]
                \,
        \infer
        {
          \begin{array}{l}
            \ASSERTSURFACE{C}{M_1}{{m_1}}{C_0}{\VEC{e}_1}
            \quad
            \ASSERTSURFACE{C_0}{M_2}{{m_2}}{C'\MSUBS{\ENCan{m_1,m_2}}{u}}{\VEC{e}_2}
          \end{array}
        }
        {
                \ASSERTSURFACE{C}{\ENCan{M_1,M_2}}{u}{C'}{\VEC{e}_1\VEC{e}_2}
        }
       \\[4.8mm]
       [\mbox{\emph{Abs}}]
        \infer
        {
          \ASSERTSURFACE
          {
            C \AND A^{\minus x\VEC{i}}
          }
          {M}
          {m}
          {C'}{\VEC{e}}
        }
        {
          \ASSERTSURFACE{A}{\lambda x.M}{u}{\allworlds \forall
            x\VEC{i}.(\ASET{C}\APP{u}{x}=m\ASET{C'})}
          {\,\emptyset}
        }
\quad 
        [\mbox{\em Rec-Ren}]
                \,
        \infer
        {
          \ASSERTSURFACE{A^{\minus x}}{\lambda y.M}{u}{B}{\VEC{e}}
	}
        {
        \ASSERTSURFACE{A^{\minus x}}{\mu x.\lambda y.M}{u}{B\MSUBS{u}{x}}{\VEC{e}}
        }
        \\[4.8mm]
        [\mbox{\emph{App}}]
        \infer
        {
        \ASSERTSURFACE{C}{\!M}{m}{C_0}{\VEC{e}}\quad 
        \ASSERTSURFACE{C_0\!}{\!N}{n}
        {\ONEEVAL{m}{n}{u}{C'}{\VEC{e}_2}}{\VEC{e}_1}
        }
        {\ASSERTSURFACE{C}{MN}{u}{C'}{\VEC{e}\VEC{e}_1\VEC{e}_2}}\\[4.8mm]
   [\mbox{\emph{If}}]
             \,
            \infer
       {
                 \ASSERT{C}{M}{b}{C_0}@\VEC{e}_1
                 \quad
                 \ASSERT{C_0\MSUBS{\true}{b}}{M_1}{u}{C'}@\VEC{e}_2
                 \quad
                 \ASSERT{C_0\MSUBS{\false}{b}}{M_2}{u}{C'}@\VEC{e}_2
        }
        {
                \ASSERT
                     {C}
                     {\IFTHENELSE{M}{M_1}{M_2}}
                     {u}
                     {C'}@\VEC{e}_1\VEC{e}_2
        }
        \\[4.8mm]
        %
        [\mbox{\emph{Deref}}]
                \,
        \infer
        {
          \ASSERTSURFACE{C}{M}{m}{C'\MSUBS{!m}{u}}{\VEC{e}}
        }
        {
          \ASSERTSURFACE{C}{!M}{u}{C'}{\VEC{e}}
        }
        \\[4mm]
        [\mbox{\emph{Assign}}]
                \,
        \infer
        {
          \ASSERTSURFACE{C}{M}{m}{C_0}{\VEC{e}_1}
          \quad
          \ASSERTSURFACE{C_0}{N}{n}{C'\LSUBS{n}{\,!m}}{\VEC{e}_2}
          \quad
          C_0\; \IMPLIES\; m=e' 
        }
        {
        \ASET{C}\ {M := N}\ \ASET{C'}@{\VEC{e}_1\VEC{e}_2e'}
        }
                \\[4.8mm]
        [\mbox{\em Ref}]
                \,
        \infer
        {
                \ASSERT{C}{M}{m}{C'}@\VEC{e}
        }
        {
                \ASSERT{C}{\refPrg{M}}{u}{
\HIDEe x.(C'\SUBST{!u}{m}\AND u\noreach i^{\TVXscript}\AND u=x)
}@\VEC{e}
        } 
                                  \\[4.8mm]
\end{array}
\]
}

\NI We require $C'$ is thin w.r.t. $m$ in [{\em Case}] and [{\em Deref}], and 
$C'$ is thin with respect to $m,n$ in [{\em App}, {\em Assign}].   
\caption{Derived Compositional Rules for Located Assertions}
\label{fig:rules:compositional:located}
\end{myfigure}

\begin{myfigure}
{\small
\[
\begin{array}{c}
\mbox{[{\em Inv}]}
        \infer
        {
          \ASSERT{C}{M}{m}{C'}@\VEC{w}
        }
        {
          \ASSERT{C\AND \allCon{\VEC{w}}{C_0}}{M}{m}{C'\AND \allCon{\VEC{w}}{C_0}}@\VEC{w}
        }
\\[6mm]
[\mbox{\emph{Inv-Val}}]
        \,
        \infer
        {
          \ASSERT{C}{V}{m}{C'}@\emptyset
        }
        {
          \ASSERT{C\AND C_0}{V}{m}{C'\AND C_0}@\emptyset
        }
\\[6mm]
[\mbox{\emph{Inv-$\noreach$}}]
        \,
        \infer
        {
          \ASSERT{C}{M}{m}{C'}@x \quad \text{no dereference occurs in $\VEC{e}$}
        }
        {
        \ASSERT{C\AND x \noreach \VEC{e}}{M}{m}{C'\AND x \noreach \VEC{e}}@x
        }
\\[6mm]
        [\mbox{\emph{Cons}}]
                \,
        \infer
        {
          C \ENTAILS C_0
          \ \ 
          \ASSERTSURFACE{C_0}{M}{u}{C_0'}{\VEC{e}}
          \ \ 
          C_0' \ENTAILS C'
        }
        {
          \ASSERTSURFACE{C}{M}{u}{C'}{\VEC{e}}
        }
\\[6mm]
        [\mbox{\emph{Cons-Eval}}]
\infer
{
\begin{array}{l}
\ASSERTSURFACE{C_0}{M}{m}{C'_0}{\VEC{e}}
\quad 
\ x\ \text{fresh};
\;\ \text{$\VEC{i}$ auxiliary}\\
\forall \VEC{i}.\EVAL{C_0}{x}{()}{m}{C'_0}
\ENTAILS
\forall \VEC{i}.\EVAL{C}{x}{()}{m}{C'}
\end{array}
}
{\ASSERTSURFACE{C}{M}{m}{C'}{\VEC{e}}}
        \\[6.4mm]
        [\mbox{\emph{$\AND$-$\ENTAILS$}}]
                 \,
        \infer
        {
          \ASSERTSURFACE{C\AND A}{V}{u}{C'}{\emptyset}
        }
        {
          \ASSERTSURFACE{C}{V}{u}{A\ENTAILS C'}{\emptyset}
        }
        \quad
        [\mbox{\emph{$\ENTAILS$-$\AND$}}]
                 \,
        \infer
        {
          \ASSERTSURFACE{C}{M}{u}{A\ENTAILS C'}{\VEC{e}}
        }
        {
          \ASSERTSURFACE{C\AND A}{M}{u}{C'}{\VEC{e}}
        }
        \\[6mm]
        [\mbox{\emph{$\OR$-Pre}}]
                 \,
        \infer
        {
          \ASSERTSURFACE{C_1}{M}{u}{C}{\VEC{e}}
          \quad
          \ASSERTSURFACE{C_2}{M}{u}{C}{\VEC{e}}
        }
        {
          \ASSERTSURFACE{C_1\OR C_2}{M}{u}{C}{\VEC{e}}
        }
        \quad
        [\mbox{\emph{$\AND$-Post}}]
                 \,
        \infer
        {
          \ASSERTSURFACE{C}{M}{u}{C_1}{\VEC{e}}
          \quad
          \ASSERTSURFACE{C}{M}{u}{C_2}{\VEC{e}}
        }
        {
          \ASSERTSURFACE{C}{M}{u}{C_1\AND C_2}{\VEC{e}}
        }
        \\[6mm]
        [\mbox{\emph{Aux}}_{\exists}]   
                \,
        \infer
        {
                \ASSERTSURFACE{C}{M}{u}{{C'\,}^{\minus i}}{\VEC{e}}
        }
        {
                \ASSERTSURFACE{\exists i.C}{M}{u}{C'}{\VEC{e}}
        }
\\[6mm]
        [\mbox{\emph{Aux}}_{\forall}V]   
                \,
        \infer
        {
                \ASSERTSURFACE{C^{\minus i}}{V}{u}{C'}{\VEC{e}}
        }
        {
                \ASSERTSURFACE{C}{V}{u}{\forall i^\alpha.C'}{\VEC{e}}
        }
\quad 
        [\mbox{\emph{Aux}}_{\forall}]   
                \,
        \infer
        {
                \ASSERTSURFACE{C^{\minus i}}{M}{u}{C'}{\VEC{e}}
		\quad
		\text{$\alpha$ is of a base type.}
        }
        {
                \ASSERTSURFACE{C}{M}{u}{\forall i^\alpha.C'}{\VEC{e}}
        }
        \\[6mm]
        [\mbox{\emph{Aux${}_{\text{inst}}$}}]
                 \,
        \infer
        {
          \ASSERTSURFACE{C(i^\alpha)}{M}{u}{C'(i^\alpha)}{\VEC{e}}
          \quad
          \text{$\alpha$ atomic}
        }
        {
        \ASSERTSURFACE{C(\mathsf{c}^\alpha)}{M}{u}{C'(\mathsf{c}^\alpha)}{\VEC{e}}}
        \qquad
        [\mbox{\emph{Aux${}_{\text{abst}}$}}]
                 \,
        \infer
        {
          \forall \mathsf{c}^\alpha.\;
          \ASSERTSURFACE{C(\mathsf{c}^\alpha)}{M}{u}{C'(\mathsf{c}^\alpha)}{\VEC{e}}
        }
        {
          \ASSERTSURFACE{C(i^\alpha)}{M}{u}{C'(i^\alpha)}{\VEC{e}} 
        }
        \\[6mm]
                              %
                                %
        [\mbox{\emph{Weak}}]
               \,
               \infer
        {\ASSERTSURFACE{C}{M}{m}{C'}{\VEC{e}}}
        {
          \ASSERTSURFACE{C}{M}{m}{C'}{\VEC{e}e'}
        }
        \quad
        \LEFTONEPREMISERULENAMED
        {Thinning}
        {
          \ASSERTSURFACE
          {C \AND !e' = i}
          {M}
          {m}
          {C' \AND !e' = i}
          {\VEC{e}e'}
          \quad
          \text{$i$ fresh}
        }
        {
          \ASSERTSURFACE
          {C}
          {M}
          {m}
          {C'}
          {\VEC{e}}
        }
\end{array}
\]
}
\hfill
\caption{Structural Rules for Located Judgements.} 
\label{fig:rules:structural:located}
\end{myfigure}

\begin{myfigure}
{\small
\[
\begin{array}{c}
         [\mbox{\em New}]
                 \,
         \infer
         {
                 \ASSERT{C}{M}{m}{C_0}@\VEC{e}_1
                 \quad 
                 \ASSERT{C_0\SUBST{!x}{m} \AND \NOTREACH{x}{\VEC{e}}}{N}{u}{C'}@\VEC{e}_2x
                 \quad 
                 x \notin \PFN{\VEC{e}} 
         }
         {
                \ASSERT{C}{\LET{x}{\REFPROG{M}}{N}}{u}{\nu x.C'}@\VEC{e}_1\VEC{e}_2
         }
\\[6mm]
         [\mbox{\em Rec}]
                 \,
         \infer
         {
                 \ASSERT{A^{\minus xi}\AND \forall j\lneq
         i.B(j)\MSUBS{x}{u}}{\lambda y.M}{u}{B(i)^{\minus
           x}}@\VEC{e}  
         }
         {
                 \ASSERT{A}{\mu x.\lambda y.M}{u}{\forall i.B(i)}@\VEC{e}
         }
\\[6mm]
        \LEFTTWOPREMISERULENAMED
        {Let}
        {
                \ASSERTSURFACE{C}{M}{x}{C_0}{\VEC{e}}
        }
        {
                \ASSERTSURFACE{C_0}{N}{u}{C'}{\VEC{e}'}
        }
        {
                \ASSERTSURFACE{C}{\LET{x}{M}{N}}{u}{C'}{\VEC{e}\VEC{e}'}
        }
\quad 
\mbox{[{\em LetOpen}]}\ 
        \infer 
        {
                \ASSERT{C}{M}{x}{\nu \VEC{y}.C_0}@\VEC{e}_1\quad 
                \ASSERT{C_0}{N}{u}{C'}@\VEC{e}_2
        }
        {
                \ASSERT{C}{\LET{x}{M}{N}}{u}{\nu \VEC{y}.C'}@\VEC{e}_1\VEC{e}_2
        } 
\\[6mm]
{\mbox{[{\em Simple}]}} 
\infer
{
  -
}
{
  \ASET{C\MSUBS{e}{u}}e:_u\ASET{C}@\VEC{e}
}
\quad 
        [\mbox{\em IfH}]
        \,
        \infer
        {
          \NOANCHORASSERTSURFACE{C \AND e}{M_1}{C'}{\VEC{e}}
          \quad
          \NOANCHORASSERTSURFACE{C \AND \neg e}{M_2}{C'}{\VEC{e}}
        }
        {
          \NOANCHORASSERTSURFACE{C}{\IFTHENELSE{e}{M_1}{M_2}}{C'}{\VEC{e}}
        }
        \\[6mm]
[\text{\it AppS}]
\infer
{
C\, \ENTAILS\,
\ASET{C}\, \APP{e}{(e_1..e_n)}= u\, \ASET{C'}@\VEC{e}'
}
{
\ASET{C}\ e(e_1..e_n):_u\,\ASET{C'}\,@\,\VEC{e}'
}\quad 
\mbox{[{\em Subs}]}\ 
        \infer 
        {
                \ASSERT{C}{M}{u}{C'}@\VEC{e}'\quad 
                u\not\in \PFN{e}
        }
        {
                \ASSERT{C\SUBST{e}{i}}{M}{u}{C'\SUBST{e}{i}}@\VEC{e}'
        } 
        \\[6mm]
        \LEFTTWOPREMISERULENAMED
        {Seq}
        {
                \NOANCHORASSERTSURFACE
                {C}{M}{C_0}{\VEC{e}}
        }
        {
                \NOANCHORASSERTSURFACE
                {C_0}{N}{C'}{\VEC{e}'}

        }
        {
                \NOANCHORASSERTSURFACE
                {C}{M; N}{C'}{\VEC{e}\VEC{e}'}
        }
        \quad
        \LEFTTWOPREMISERULENAMED
        {Seq-Inv}
        {
                \NOANCHORASSERTSURFACE
                {C_1}{M}{C'_1}{\VEC{e_1}}
        }
        {
                \NOANCHORASSERTSURFACE
                {C_2}{N}{C'_2}{\VEC{e_2}}

        }
        {
                \NOANCHORASSERTSURFACE
                {C_1 \ANDls \allCon{\VEC{e_1}}{C_2}}
                {M; N}
                {C'_2  \ANDls   \someCon{\VEC{e_2}}{C'_1}}
                {\VEC{e_1}\VEC{e_2}}
        }
\end{array}
\]
}
$C'$ is thin w.r.t. $m$ in [{\em New} and $x$ in [{\em Let, LetOpen}].
$C'_1$ and $C_2$ are tame in [{\em Seq-Inv}]. 

\caption{Other Located Proof Rules.} 
\label{figure:rules:derived}
\label{fig:rules:derived}
\end{myfigure}

\vfill\eject
\subsection{Proofs of Soundness}
\label{app:soundness}
We prove the soundness theorem. 
We use the following lemma. 

\begin{lemma}[Substitution and Thinning]
\label{lem:subs}\hfill
\begin{enumerate}[\em(1)]
\item $\MMM\models C \AND u = V$ iff $\MMM[u:V]\models C$.
\item Suppose $m,m_1,m_2\not\in \FV{\MMM,C}\cup\ASET{u,v}$. Then:
\begin{enumerate}[\em(a)]
\item
If $(\new \VEC{l})\MMM[m:V][u:{\INJ{i}{m}}]\models C$, 
then $(\new \VEC{l})\MMM[u:{\INJ{i}{V}}]\models C$.  

\item
If $(\new \VEC{l})\MMM[m:V][u:{\pi_1(m)}]\models C$, 
then $(\new \VEC{l})\MMM[u:{\pi_1(V)}]\models C$. 

\item
If $(\new \VEC{l})\MMM[m_1:V_1][m_2:V_2][u:{\TPL{m_1}{m_2}}]\models C$, 
then $(\new \VEC{l})\MMM[u:{\TPL{V_1}{V_2}}]\models C$. 

\item Suppose $l\not\in\FL{\MMM}$. 
Then $(\new \VEC{l}l)\MMM[m:l][u:V][l\mapsto V]\models C$
implies $(\new \VEC{l})\MMM[u:V]\models C$

\item
Suppose $l\not\in\FL{\MMM}$ and $\FV{V}\cup  \FL{V}=\emptyset$. 
Then $(\new l)\MMM[m:l][l\mapsto V]\models C$
implies $\MMM\models C$.

\item
Suppose $l\not\in\FL{\MMM}$ and $\FV{V}\cup  \FL{V}=\emptyset$. 
Then $\MMM[m:l][l\mapsto V]\models C$
implies $\MMM\models C$.

\end{enumerate}
\item 
$\MMM \models \exists m.(\someCon{x}(C\AND !x=m) \AND m=e)$ 
                iff 
$\updates{\MMM}{x}{\MAP{e}_{\xi,\sigma}} \models C$
\end{enumerate}
\end{lemma}
\begin{proof}
For (1), we derive:
\[
\begin{array}{llll}
\MMM\models C \AND u = V
& \LITEQ & 
\MMM\models C \AND \MMM\models u = V\\
& \LITEQ & 
\MMM\models C \AND \MMM[u:V]\WB \MMM\\
& \LITEQ & 
\MMM[u:V]\models C
\end{array}
\]
(2) is mechanical by induction on $C$. 
We only show some interesting cases. Others are similar. 
For (2-a), let 
$\MMM_1\LITEQ (\new \VEC{l})\MMM[m:V][u:{\INJ{i}{m}}]$ 
and $\MMM_2 \LITEQ (\new \VEC{l})\MMM[u:{\INJ{i}{V}}]$. 

Assume $C=e_1 = e_2$. Then, with $w$ fresh and   
$m\not\in \FV{e_1,e_2}$, we have 
$\MMM_1[w:e_1]\WB \MMM_1[w:e_2]$ iff 
$\MMM_2[w:e_1]\WB \MMM_2[w:e_2]$. Hence 
$\MMM_1\models e_1 = e_2$ iff  
$\MMM_2\models e_1 = e_2$. 

Assume $C=\forall x.C'$.
Then we have:
\[
\begin{array}{lllll}
\MMM_1\models \forall x.C' & \LITEQ & 
\forall L\in \mathcal{F}.\MMM_1[x:L]\models C'\\
 & \LITEQ & 
\forall L'\in \mathcal{F}.\MMM_2[x:L']\models C'& 
\mbox{such that } m\not\in\FV{L'}\\
& \LITEQ & \MMM_2\models \forall x.C' 
\end{array}
\]
Assume $C=\HIDEe x.C'$.
Then we have:
\[
\begin{array}{lllll}
\MMM_1\models \HIDEe x.C' & \LITEQ & 
\exists \MMM_0.(\MMM_1\WB (\new l)\MMM_0\AND \MMM_0[x:l]\models C')\\
 & \LITEQ & 
\exists \MMM_0'.(\MMM_2\WB (\new l)\MMM_0'
\AND \MMM_0'[x:l]\models C') & 
\mbox{such that } 
\MMM_0'\LITEQ \MMM_0/m\MSUBS{V}{m}\\
 & \LITEQ & 
\MMM_2\models \forall x.C' 
\end{array}
\]
Assume $C=\ONEEVAL{x}{y}{z}{C'}$.
Then we derive: 
\[
\begin{array}{llll}
\MMM_1 \models  \ONEEVAL{x}{y}{z}{C'} & \LITEQ & 
\exists \MMM_1'.(\MMM_1[z:xy]\Downarrow \MMM'_1\AND \MMM'_1 \models  C')\\
&&\qquad\qquad\qquad \mbox{with\ } 
\MMM_1'\LITEQ (\new \VEC{l})\MMM'[m:V][u:{\INJ{i}{m}}]\\ 
& \LITEQ & 
\exists \MMM_2'.(\MMM_2[z:xy]\Downarrow \MMM'_2\AND \MMM'_2 \models  C')\\
&&\qquad\qquad\qquad \mbox{with\ } 
\MMM_2' \LITEQ (\new \VEC{l})\MMM'[u:{\INJ{i}{V}}]
\mbox{ and (IH)}\\ 
 & \LITEQ & 
\MMM_2 \models \ONEEVAL{x}{y}{z}{C'} 
\end{array}
\]
Others (b-f) are similar. 
(3) is from \cite{ALIAS}. 
\end{proof}
Below we write:  
\begin{center} 
$\MMM \converges \MMM' \models C'$
\ for \ 
$\MMM \converges \MMM' \ANDl \MMM' \models C'$
\end{center}
We start with [\RULENAME{Var}].
\[
\begin{array}{llll}
\MMM \models C\SUBST{x}{u} &\THEN & 
\MMM \models C\AND u=x\\
&\THEN & \expands{\MMM}{u}{x} \models C 
& \text{Lemma \ref{lem:subs}(1)} 
\end{array}
\]
Similarly for [\RULENAME{Const}] using Lemma \ref{lem:subs}(1). 
Next, [\RULENAME{Add}] is proved as follows:  
\[
\begin{array}{llllllll}
        \MMM \models C
                &\Rightarrow & 
        \expands{\MMM}{m_1}{M_1} \converges \MMM_1
\models C_0
& \text{IH}
                \\
          &\Rightarrow & 
        \expands{\MMM_1}{m_2}{M_2} \converges \MMM_2
\models C'\SUBST{m_1+m_2}{u}
& \text{IH}\\
          &\Rightarrow & 
\expands{\expands{\expands{\MMM}{m_1}{M_1}}{m_2}{M_2}}{u}{m_1+m_2}
\converges \MMM' \models C'                \\
          &\Rightarrow & 
\expands{\MMM}{u}{M_1+M_2}
\converges \MMM'/m_1 m_2 \models C'                
& \text{Proposition \ref{pro:thin} (\ref{pro:thin_wrt_y})} 
\end{array}
\]
[\RULENAME{Inj$_1$}] is proved as:
\[
\begin{array}{llllllll}
        \MMM \models C
                &\quad\Rightarrow\quad
        \expands{\MMM}{m}{M} \converges (\new \VEC{l})\MMM'[m:V] 
\models C'\SUBST{\INJ{1}{m}}{u}
& \text{IH}
                \\
                &\quad\Rightarrow\quad
        \expands{\expands{\MMM}{m}{M}}{u}{\INJ{1}{m}}\converges 
(\new \VEC{l})\MMM'[m:V][u:\INJ{1}{V}]\models C'
& \text{Lemma \ref{lem:subs}(1)} 
                \\
                &\quad\Rightarrow\quad
(\new \VEC{l})\MMM'[u:\INJ{1}{V}]\models C'
& \text{Lemma \ref{lem:subs}(2-a)} \\
                &\quad\Rightarrow\quad
        \expands{\MMM}{u}{\INJ{1}{M}}\models C'\\
\end{array}
\]
[\RULENAME{Proj}] and [\RULENAME{Pair}] are similarly proved 
using Lemma \ref{lem:subs}(2-b,c) respectively. 


For [\RULENAME{Case}], we reason: 
\begin{align*}
    \MMM \models C
        & \Rightarrow\quad 
        \expands{\MMM}{m}{M} \converges (\new \VEC{l}')\MMM_0[m:\INJ{i}{V}]
\models C_0\\
& 
\quad\quad\quad        \text{if}\ 
\MMM = (\new \VEC{l})(\xi, \sigma), 
\ 
 (\new \VEC{l})(M\xi, \sigma)\converges 
 (\new \VEC{l}')(\INJ{i}{V}, \sigma'), 
\   \text{and}\ 
\MMM_0  = (\xi, \sigma')
                \\
                & \Rightarrow\quad
        (\new \VEC{l}')\expands{\MMM_0}{m}{\INJ{i}{V}} 
                       \models C_0 \AND m = \INJ{i}{x_i}
                \\
                & \Rightarrow\quad
 \expands{\expands{(\new \VEC{l}')\MMM_0}{m}{\INJ{i}{x_i}}}{u}{M_i} 
              \converges 
 \expands{\expands{(\new \VEC{l}'')\MMM'}{m}{\INJ{i}{V}}}{u}{W} 
\models C'
                \\
                & \Rightarrow\quad
\expands{(\new \VEC{l}'')\MMM'}{u}{W}  
\models C'
\end{align*}
The last line is by the thinness of $C'$ with respect to $m$. 

Now we reason for [\RULENAME{Abs}]. 
We note, if $A$ is stateless (cf.~Definition \ref{def:stateless}) and  
$\MMM\models A$, then: 
$\expands{\MMM}{u}{M}\Downarrow \MMM'$ with $u$ fresh implies 
$\MMM'\models A$. 
\[
\begin{array}{llll}
&  \MMM \models A \ENTAILS \expands{\MMM}{u}{\lambda x.M}
       \models \allworlds \forall x\VEC{i}.\EVAL{C}{u}{x}{m}{C'}\\[1mm]
\LITEQ & \MMM \models A \ENTAILS 
\expands{\expands{\expands{\expands{\MMM}{u}{\lambda x.M}}{x}{N_x}}{\VEC{i}}{\VEC{N}}}{k}{N}
\converges \MMM' \ANDl \MMM\WB \MMM'/ x\VEC{i}
\ANDl \MMM' \models \EVAL{C}{u}{x}{m}{C'}\\[1mm]
\LITEQ & \MMM \models A \ENTAILS 
((\expands{\expands{\expands{\expands{\MMM}{u}{\lambda x.M}}{x}{N_x}}{\VEC{i}}{\VEC{N}}}{k}{N}
\converges \MMM' \ANDl \MMM\WB \MMM'/ x\VEC{i} 
\ANDl \MMM' \models C)\\
       & \hspace{8cm}\ENTAILS 
\MMM'[m:ux]\converges \MMM'' \ANDl \MMM''\models C')\\[1mm]
\LITEQ & \MMM \models A \ENTAILS 
((\expands{\expands{\expands{\expands{\MMM}{u}{\lambda x.M}}{x}{N_x}}
 {\VEC{i}}{\VEC{N}}}{k}{N}
\converges \MMM' \ANDl \MMM\WB \MMM'/ x\VEC{i} 
\ANDl \MMM'\models C\AND A)\\
       & \hspace{8cm}\ENTAILS 
\MMM'[m:ux]\converges \MMM'' \ANDl \MMM''\models C')\\[1mm]
\subset & 
\MMM' \models A \AND C
\ENTAILS (\expands{\MMM'}{m}{M}\Downarrow \MMM'' \AND \MMM'' \models  C')
\end{array}
\]
[\RULENAME{App}] is reasoned as follows.  Below $k$ fresh. 
\begin{align*}
        \MMM \models C
        &\quad\Rightarrow\quad
        \expands{\MMM}{m}{M} \converges \MMM_0 \models C_0 
                \\
        &\quad\Rightarrow\quad
        \expands{\MMM}{n}{N} \converges 
        \MMM_1 \models C_1 \AND \ONEEVAL{m}{n}{n}{C'} 
                \\
        &\quad\Rightarrow\quad
\expands{\expands{\expands{\MMM}{m}{M}}{n}{N}}{u}{m n}
            \converges \MMM' \models C'    
                \\
        &\quad\Rightarrow\quad
\expands{\expands{\expands{\MMM}{m}{M}}{n}{N}}{u}{MN}
                  \converges \MMM' \models C'    
                \\
        &\quad\Rightarrow\quad
    \expands{\MMM}{u}{MN} \converges \MMM' / mn \models C'
\end{align*}
The last line is derived by the thinness of $C'$ with respect to $m,n$. 

For [\RULENAME{Deref}], we infer:
\begin{align*}
        \MMM \models C
                &\quad\Rightarrow\quad
        \expands{\MMM}{m}{M} \converges \MMM' \models C'\SUBST{!m}{u}
                \\
                &\quad\Rightarrow\quad
        \expands{\MMM}{m}{!M} \converges \MMM' / m \models C'
\end{align*}
For [\RULENAME{Assign}]  
Assume $u$ is fresh.
\begin{align*}
       \MMM \models C
               &\quad\Rightarrow\quad
       \expands{\MMM}{m}{M}\converges \MMM_0\models C_0\\ 
              &\quad\Rightarrow\quad
       \expands{\MMM_0}{n}{N}\converges \MMM'\models C'\LSUBS{n}{!m}
               \\
               &\quad\Rightarrow\quad
       \updates{\MMM'}{m}{n} \converges \MMM'' \models C' 
\quad \text{Lemma \ref{lem:subs}(3)}
               \\
               &\quad\Rightarrow\quad
   \expands{\MMM}{u}{M := N} \converges \expands{\MMM'' / mn}{u}{()} \models C'
\ANDl u = ()
\end{align*}
For [{\em Rec-Ren}], 
\begin{align*}
\MMM \models A &\quad\Rightarrow\quad
        \expands{\MMM}{u}{\lambda x.M}\models {B}
                \\
                &\quad\Rightarrow\quad
        \expands{\expands{\MMM}{f}{\mu f.\lambda x.M}}{u}{\lambda x.M} \models A
                \\
                &\quad\Rightarrow\quad
\expands{\expands{\MMM}{f}{\mu f.\lambda x.M}}{u}{\mu f.\lambda x.M} \models A
                \\
                &\quad\Rightarrow\quad
         \expands{\MMM}{u}{\mu f.\lambda x.M} 
\models f = u \IMPLIES A 
                \\
                &\quad\Rightarrow\quad
        \expands{\MMM}{u}{\mu f.\lambda x.M} \models A\SUBST{u}{f}
\quad \text{Lemma \ref{lem:subs}(1)}
\end{align*}


[{\em If}] is similar with [{\em Add}] using 
Proposition \ref{pro:thin_wrt_y}. 

[{\em Ref}] appeared in the main text (the second last line 
uses Lemma \ref{lem:subs}(2-d) to delete $m$). 

We complete all cases. \qed 

\subsection{Soundness of the Invariant Rule}
\label{subsec:invsound:proof}
Among the structural rules, 
 we prove the soundness of the main invariance rule, 
[{\em Inv}] in Figure \ref{fig:rules:structural:located}.

\begin{lemma} \label{ALEALE}
Suppose $C$ is tame and $\MMM\models C$.
Suppose $\MMM\EVOLVESinc{u_1..u_n}\MMM'$ and $\MMM\WB\MMM/u_1..u_n$.
Then $\MMM'\models C$.
\end{lemma}
\begin{proof}
By mechanical induction on $C$ noting it only contains 
evaluation formulae under $\allworlds$.
\end{proof}

\begin{lemma}
\label{tame:xfree}
Suppose $\MMM\models \allCon{\VEC{w}}{C}$ and $C$ is tame. 
Then for each $M$ and $\MMM'$ if
$\expands{\MMM}{u}{M}\converges\MMM'$
and 
$\MMM[z:\LET{\VEC{x}}{!\VEC{w}}{\LET{y}{M}{\VEC{w}:=\VEC{x}}}]
\converges
\MMM''$
s.t. $\MMM''/z\WB \MMM$
then we have
$\MMM' \models C$.
\end{lemma}
\begin{proof}
  For simplicity we assume $\VEC{w}$ is a singleton (the general case
  is the same).  Let $\MMM\models \allCon{w}{C}$ and $C$ be tame.
  Suppose $\expands{\MMM}{u}{M}\converges\MMM'$ such that only the
  content of $w$ is affected.  We let with appropriate closed 
  $V_0$:
  \begin{equation}
  \MMM[x:!w][y:\refPrg{V_0}][u:\LET{m}{M}{(y:=!w; w:=x; m)}]\converges\MMM''
  \quad\quad 
  \MMM\WB \MMM''/xyu
  \end{equation}
  Hence by Lemma \ref{ALEALE} we have:
  \begin{equation}\label{RERERE}
    \MMM''\models [!w]C
  \end{equation}
  Further note
  \begin{equation} \label{PEPEPE}
    \MMM''[w\mapsto !y]\converges \MMM'''
    \quad\quad 
    \MMM'\WB \MMM'''/xy
  \end{equation}
  By (\ref{RERERE}) and (\ref{PEPEPE}) we obtain 
    $\MMM'''\models C$. 
  By Lemma \ref{ALEALE} and this, 
we have $\MMM'\models C$, as required. 
\end{proof}

We now prove:

\begin{proposition}
The following rule is sound.
\rm 
\[ 
\begin{array}{c}
\text{[{\em Inv}]}
         \infer
         {
           \ASSERT{C}{M}{m}{C'}@\VEC{w}
	  \qquad
          C_0\mbox{ is tame}
         }
         {
         \ASSERT
         {C\AND \allCon{\VEC{w}}{C_0}}
         {M}{m}
         {C'\AND \allCon{\VEC{w}}{C_0}}@\VEC{w}
         }
\end{array}
\]
\end{proposition}
\begin{proof}
Assume $\ASSERT{C}{M}{u}{C'}@\VEC{w}$. Then by definition, 
for each $\MMM$ such that $\MMM\models C$ we have:
\begin{eqnarray} \label{inv:rule:1}
& \expands{\MMM}{u}{M}\converges\MMM'\models C'\\
& \MMM[z:\LET{\VEC{x}}{!\VEC{w}}{\LET{y}{ee'}{\VEC{w}:=\VEC{x}}}]
\converges\MMM''\ \text{s.t.} \ \MMM''/z \WB \MMM
\label{inv:rule:1b}
\end{eqnarray}
Then: 
 \[
 \begin{array}{rll}
 \MMM \proves C \ANDl \rlap{$\allCon{\VEC{w}}{C_0}$} \\
 \Rightarrow &\MMM \proves
 \allCon{\VEC{w}}{\allCon{\VEC{w}}{{C_0}}} \quad\quad \text{(by the axiom 
 $\allCon{\VEC{w}}{\allCon{\VEC{w}}{{C_0}}} \LITEQ  \allCon{\VEC{w}}{{C_0}}$)}\\ 
 \Rightarrow &
\forall \MMM',M.((\expands{\MMM}{u}{M}\converges \MMM'\ANDl \\
&\quad \MMM[z:\LET{\VEC{x}}{!\VEC{w}}{\LET{y}{ee'}{\VEC{w}:=\VEC{x}}}]
\converges\MMM''\WB
\MMM[z:()]
 \ENTAILS \MMM'\models \allCon{\VEC{w}}{C_0})\\
 \Rightarrow &\MMM'\models C' \hfill\qquad 
\text{((\ref{inv:rule:1},\ref{inv:rule:1b}) above)}\\ 
 \Rightarrow & \MMM'\models C' \ANDl \allCon{\VEC{w}}{C_0} 
\end{array}
 \]
Hence we have 
 $\ASSERT{C\AND {\allCon{\VEC{w}}{C_0}}}{M}{m}{C'\AND {\allCon{\VEC{w}}{C_0}}}@\VEC{w}$, as required. 
\end{proof}

\subsection{Soundness of [{\em LetOpen}] and [{\em Subs}]}
\label{sub:letopen_subs}
We prove soundness of [{\em LetOpen}] and [{\em Subs}] used in 
\S~\ref{subsec:newvar}. 
For [{\em LetOpen}] (we prove the case that $\VEC{y}$ is a singleton), 
we derive: 
\[
\begin{array}{llllllll}
        \MMM \models C
             &\Rightarrow & 
\expands{\MMM}{x}{M} \converges \MMM'
\models \new y.C_0
& \text{Assumption}
                \\
                &\LITEQ & 
\expands{\MMM}{x}{M} \converges \MMM'
\AND \exists \MMM_0.(\MMM' \WB (\new l)\MMM_0 \AND \MMM_0 \models C_0)
\end{array}
\]
Also we have: 
\[
\begin{array}{llllllll}
        \MMM_0 \models C_0
             &\Rightarrow & 
\expands{\MMM}{u}{N} \converges \MMM_0'
\models C'
& \text{Assumption}
\end{array}
\]
Combining these two, 
we have: 
\[
\begin{array}{llllllll}
        \MMM \models C
             &\Rightarrow & 
\expands{\MMM}{x}{M} \converges \MMM'
\AND \exists \MMM_0.(\MMM' \WB (\new l)\MMM_0 \AND \MMM_0[u:N]
\Downarrow \MMM_0' \AND \MMM_0' \models C')\\
            &\Rightarrow & 
\expands{\MMM}{u}{\mathtt{let} \ x = M \ \mathtt{in} \ N} \converges 
\MMM'' \AND \MMM'' \models C' \ \mbox{with } \MMM''/x = \MMM_0'
\end{array}
\]
The last line is by thinness. 

For [{\em Subs}] (we prove the case that $\VEC{e}$ is a singleton), we have: 
\[
\begin{array}{lllllll}
\MMM \models C &\THEN & 
\MMM[u:M]\Downarrow \MMM' \AND \MMM'\models C\\
 &\THEN & 
\forall \MMM_0.(\MMM_0\models i=e \AND 
\MMM_0[u:M]\Downarrow \MMM' \THEN \MMM'\models i=e) & (u\not\in 
\PFN{e})\\
 &\THEN & \forall \MMM_0.(\MMM_0\models (C\AND i=e) \AND 
\MMM_0[u:M]\Downarrow \MMM'\models (C'\AND i=e))
\end{array}
\]

\section{Soundness of the Axioms}
\label{app:axioms}

This appendix lists 
the omitted proofs from Section \ref{sec:axioms}.  
We first prove the basic lemma and
propositions. In \S~\ref{subsec:axioms:cq}, we show the
axioms for the content quantifications.
In \S~\ref{app:proof:AHI}, 
we prove \AIHax-axioms.

\subsection{Proofs of Lemma \ref{lemma:29410}}
\label{subsec:axioms:eq}

\NI For (\ref{lem:sat:renaming}), both directions are simultaneously
established by induction on $C$, proving for both $C$ and its
negation. If $C$ is $e_1=e_2$, we have, letting $\MMM\DEFEQ(\new
\VEC{l})(\xi, \sigma)$, $\delta\DEFEQ\MMSUBS{u}{v}$ and
$\xi'\DEFEQ\xi\delta$:
\begin{reasoning}
\lefteqn{\MMM\models e_1=e_2}\\
&\;\THEN\;&
\MMM[x:e_1]
\congmodel
\MMM[x:e_2]\\
&\;\THEN\;&
(\new \VEC{l})(\xi\CD x:\MAP{e_1}_{(\xi,\sigma)},\, \sigma)
\cong_{\ID}
(\new \VEC{l})(\xi\CD x:\MAP{e_2}_{(\xi,\sigma)},\, \sigma)\\
&\;\THEN\;&
(\new \VEC{l})(\xi'\CDs x\!:\!\MAP{e_1\delta}_{(\xi',\sigma)},\, \sigma)
\cong_{\ID} 
(\new \VEC{l})(\xi'\CDs x\!:\!\MAP{e_2\delta}_{(\xi',\sigma)},\, \sigma)
\quad & (\ast)\\
&\;\THEN\;&
\MMM\delta[x:e_1\delta]
\congmodel
\MMM\rho[x:e_2\delta]\\
&\;\THEN\;&
\MMM\delta\models (e_1=e_2)\delta
\end{reasoning}
Above $(\ast)$ used $\MAP{e_i}_{(\xi,\sigma)}\DEFEQ \MAP{e_i\delta}_{(\xi',\sigma)}$.
Dually for its negation. The rest is easy by induction.
(\ref{lem:sat:perm}) is by precisely the same reasoning.
(\ref{lem:sat:exchange}) is immediate from (\ref{lem:sat:renaming}) and (\ref{lem:sat:perm}). 
(\ref{lemma:29410:6}) is similar, for which we again show a base case. 
\begin{reasoning}
\lefteqn{\MMM'\models e_1=e_2}\\
&\IFFs&
\MMM[x:e_1]\congmodel
\MMM[x:e_2]
\qquad\qquad&\text{(By Definition)}\\
&\IFFs&
\MMM[x:e_1][u:e]\congmodel
\MMM[x:e_2][u:e]
\qquad\qquad&\text{(congruency of $\WB$)}\\
&\IFFs&
\MMM[u:e][x:e_1]\congmodel
\MMM[u:e][x:e_2]
\qquad\qquad&\text{(By (\ref{lem:sat:exchange}))}
\end{reasoning}
Dually for the negation. 
For (\ref{lemma:29410:1}), 
the ``only if'' direction:
{
\begin{reasoning}
\lefteqn{\MMM \models e_1 = e_2}\\
&\IFFs &
\MMM[u:e_1]\congmodel \MMM[u:e_2]
\qquad&\text{(By Definition)}\\
&\IFFs &
\MMM[u:e_1][v:e_2]\congmodel \MMM[u:e_2][v:e_2]\ANDl\\
&  &
\MMM[u:e_2][v:e_2]\congmodel \MMM[u:e_2][v:e_1]
\qquad&\text{(By (\ref{lem:sat:exchange}))}\\
&\THENs &
\MMM[u:e_1][v:e_2]\congmodel \MMM[u:e_2][v:e_1].
\end{reasoning}}
Operationally, the encoding of models simply removes all references to
$u, v$ and replaces them by positional information: hence all relevant
difference is induced, if ever, by behavioural differences between
$e_1$ and $e_2$, which however cannot exist by assumption. The ``if''
direction is immediate by projection.

(\ref{lemma:29410:3}) 
is best argued using concrete models.
For the former, 
let $\MMM=(\new \VEC{l})(\xi, \sigma)$ and let $\xi(x)=W$. We infer:
\begin{eqnarray*}
\expands{\expands{\MMM}{u}{x}}{v}{e} 
& \DEFEQ &
(\new \VEC{l})(\xi\CDs u:W\CDs v:e\xi,\ \sigma)\\
& \DEFEQ &
(\new \VEC{l})(\xi\CDs u:W\CDs v:(e\SUBST{u}{x})\xi,\ \sigma)
\end{eqnarray*}
For the latter, 
let $\MMM=(\new \VEC{l})(\xi, \sigma)$ and  $W=\MAP{e}_{\xi, \sigma}$ (the
standard interpretation of $e$ by $\xi$ and $\sigma$).
We then have
\begin{eqnarray*}
\expands{\expands{\MMM}{u}{e}}{v}{e'}
& \congmodel &
(\new \VEC{l})(\xi\CDs u:W\CDs v:\MAP{e'}_{\xi, \sigma},\ \sigma)\\
& \DEFEQ &
(\new \VEC{l})(\xi\CDs u:W\CDs v:\MAP{e'\SUBST{e}{u}}_{\xi, \sigma},\ \sigma)
\end{eqnarray*}
The last line is because the interpretation is homomorphic.
\qed 

\subsection{Proof of Proposition \ref{pro:necessity}}
\label{subsec:axioms:necessity}

\begin{propositionapp}{\ref{pro:necessity}}{}\hfill
{
\begin{enumerate}[\em(1)]
\item $\allworlds (C_1\entails C_2) \ENTAILS \allworlds C_1\entails
\allworlds C_2$; 
$\allworlds C \ENTAILS C$; 
$\allworlds \allworlds C \LITEQ \allworlds C$; 
$C \entails \someworld C$. Hence $\allworlds C \entails
\someworld C$.
\item {\em (permutation and decomposition)}
\begin{enumerate}[\em(a)]
\item $\allworlds e_1 = e_2 \LITEQ e_1 = e_2$ 
and $\allworlds e_1 \not= e_2 \LITEQ e_1 \not= e_2$
if $e_i$ does not contain dereference. 
\item $\allworlds (C_1 \AND C_2) \LITEQ \allworlds C_1 \AND \allworlds C_2$. 

\item $\allworlds C_1 \OR \allworlds C_2\ENTAILS \allworlds (C_1 \OR C_2)$. 

\item $\allworlds \forall x.C \ENTAILS \forall x.\allworlds C$   
and $\allworlds \forall x. \allworlds C \LITEQ 
\allworlds \forall x.C$. 

\item $\exists x.\allworlds C \ENTAILS \allworlds \exists x.C$.

\item $\allworlds \HIDEa x.C \LITEQ \HIDEa x.\allworlds C$; 
and 
$\HIDEe x.\allworlds C \ENTAILS \allworlds \HIDEe x.C$.

\item $\allworlds \exists\TVX.C \LITEQ \exists\TVX.\allworlds C$ and 
$\allworlds \forall\TVX.C \LITEQ \forall\TVX.\allworlds C$. 

\item 
$\allworlds \allCon{x}{C} \LITEQ 
\allCon{x}{\allworlds C} \LITEQ \allworlds C$
and  
$\someCon{x}{\allworlds C}\LITEQ \allworlds C
\ENTAILS \allworlds \someCon{x}{C}$.
\end{enumerate}
\end{enumerate}
}
\end{propositionapp}
(1) is obvious by definition. For (2-a), suppose 
$\MMM\models e_1 = e_2$. Then by definition $\WB$, 
for all $\MMM'$ such that $\MMM\EVOLVES\MMM'$, 
we have $\MMM'[u:e_1] \WB \MMM'[u:e_2]$. 
Hence $\MMM\models \allworlds e_1 = e_2$, as required. 
Similarly for $e_1 \not = e_2$. For (2-b), we have:
\[ 
\begin{array}{rlll}
\MMM \models \allworlds (C_1 \AND C_2)  
\LITEQ & \forall \MMM'. (\MMM\EVOLVES\MMM' \ENTAILS 
\MMM' \models C_1 \AND \MMM' \models C_2)\\
\LITEQ & \forall \MMM'. (\MMM\EVOLVES\MMM' \ENTAILS 
\MMM' \models C_i) \quad (i=1,2)\\
\LITEQ & \MMM \models \allworlds C_1 \AND \allworlds C_2
\end{array}
\]
For (2-c), we derive:
\[ 
\begin{array}{rlll}
\MMM \models \allworlds C_1 \OR \allworlds C_2  
\LITEQ & \forall \MMM'. (\MMM\EVOLVES\MMM' \ENTAILS 
\MMM' \models C_i)\quad (i=1 \OR i=2)\\
\Rightarrow & \forall \MMM'. (\MMM\EVOLVES\MMM' \ENTAILS 
\MMM' \models C_1\OR C_2)\\
\LITEQ & \MMM \models \allworlds (C_1 \OR C_2)
\end{array}
\]
For (2-d(1)), we derive, 
with $u$ fresh:  
\[ 
\begin{array}{lllll}
\MMM \models \allworlds \forall x.C\\ 
\quad \LITEQ\quad 
 \forall \MMM'.(\MMM[u:N]\Downarrow \MMM' 
\ENTAILS \forall L\in \mathcal{F}. 
(\MMM'[x:L]\Downarrow \MMM'' \ENTAILS \MMM'' \models C))
\\
\quad \Rightarrow  
\quad \forall \MMM'_0,L'\in \mathcal{F}, \forall N.
(\MMM[x:L'][u:N\MSUBS{L'}{x}] \Downarrow \MMM_0' \ENTAILS \MMM_0' \models C)
\mbox{\ such that } u \not\in \FV{L'}\\
\quad \Rightarrow 
\quad \forall \MMM'_0,L'\in \mathcal{F},\forall N. 
(\MMM[x:L'][u:N] \Downarrow \MMM_0' \ENTAILS \MMM_0' \models C)
\mbox{\ such that } u\not\in \FV{L'}\\
\quad \Rightarrow \quad 
 \MMM\models \forall x.\allworlds C
\end{array}
\]
Note that $x\not\in \FV{N}$ in the second line. To derive the third
line, we use the fact for all $L\in\mathcal{F}$ such that 
$u\not\in \FV{L}$ and all $N$, 
if $\MMM[x:L][u:N\MSUBS{L}{x}]\Downarrow \MMM'$, then 
$\MMM[x:L][u:N]\Downarrow \MMM'$. 

For (2-d(2)), by $\allworlds C \ENTAILS C$, 
we have 
$\allworlds \forall x.\allworlds C \ENTAILS \allworlds \forall x. C$. 
The other direction is 
proved with $\allworlds \allworlds C \LITEQ \allworlds C$, 
as $\allworlds \forall x. C \LITEQ 
\allworlds \allworlds \forall x. C \ENTAILS 
\allworlds  \forall x. \allworlds C$.  

For (f-1), 
we derive, with $u$ fresh:  
\[ 
\begin{array}{llll}
\MMM\models \HIDEa x.\allworlds C\\ 
\quad \LITEQ 
\quad \forall \MMM'.((\new l)\MMM' \WB \MMM 
\ENTAILS \forall \MMM'',N.(\MMM'[x:l][u:N] \Downarrow \MMM''\DEFEQ (\new \VEC{l})\MMM'''[x:l][u:V] 
\ENTAILS \MMM''\models C))\\
\quad \LITEQ 
\quad \forall \MMM_0,N.(\MMM[u:N] \Downarrow \MMM_0 
\ENTAILS \forall \MMM_0'.(\MMM_0 \WB 
(\new l)\MMM_0'  \ENTAILS \MMM_0'[x:l]\models C)) \\
\quad \quad \quad \mbox{such
that $l,x\not\in \FV{N}\cup\FL{N}$}\\
\quad \quad \quad \mbox{with }\MMM_0 \DEFEQ  (\new \VEC{l})(\new l)\MMM'''[u:V], \quad 
\MMM_0' \DEFEQ (\new \VEC{l})\MMM''[u:V]\\
\quad \LITEQ 
\quad \MMM\models \allworlds \HIDEa x. C 
\end{array}
\]

For (f-2), 
we derive, with $u$ fresh:  
\[ 
\begin{array}{llll}
\MMM\models \HIDEe x.\allworlds C\\ 
\quad \LITEQ 
\quad \exists \MMM'.((\new l)\MMM' \WB \MMM 
\AND \forall \MMM''.(\MMM'[x:l][u:N] \Downarrow \MMM''\DEFEQ (\new \VEC{l})\MMM'''[x:l][u:V] 
\ENTAILS \MMM''\models C))\\
\quad \Rightarrow
\quad \forall \MMM_0.(\MMM[u:N] \Downarrow \MMM_0 
\ENTAILS \exists \MMM_0'.(\MMM_0 \WB 
(\new l)\MMM_0'  \AND \MMM_0'[x:l]\models C))\\
\quad\quad\quad \mbox{with }\MMM_0 \DEFEQ  (\new \VEC{l})(\new l)\MMM'''[u:V], \quad 
\MMM_0' \DEFEQ (\new \VEC{l})\MMM''[u:V]\\
\quad \LITEQ 
\quad \MMM\models \allworlds \HIDEe x. C 
\end{array}
\]
(g) is trivial. For (h), we prove the first equation. 
With $u$ fresh, we have: 
\[ 
\begin{array}{rlll}
\MMM\models \allworlds \allCon{x}{C} 
& \LITEQ & \forall N.(\MMM[u:N]\Downarrow \MMM' \ENTAILS 
\forall L\in \mathcal{F}.\updates{\MMM'}{x}{L}\models C)\\
& \Rightarrow & \forall L\in \mathcal{F},\forall N.
\MMM[u:((x:=L);N)]\Downarrow \MMM' \ENTAILS \MMM'\models C\\
& \LITEQ &
\MMM\models \allCon{x}{\allworlds C} 
\\
& \Rightarrow &
\MMM\models \allworlds C 
\end{array}
\]
The last line is by the axiom 
$\allCon{x}{C} \ENTAILS C$. 
For other direction, with $u$ fresh, 
\[ 
\begin{array}{rlll}
\MMM\models \allworlds C 
& \LITEQ &\forall N.(\MMM[u:N]\Downarrow \MMM'\ENTAILS \MMM'\models C)\\
& \Rightarrow & \forall L\in \mathcal{F},\forall N.
(\MMM[u:(x:=L;N)]\Downarrow \MMM'\ENTAILS \MMM'\models C)\\
& \Rightarrow & 
\MMM\models \allCon{x}{\allworlds C} \\
& \LITEQ & \forall L,L'\in \mathcal{F},\forall N.
(\updates{\MMM}{x}{L}[u:(N;x:=L')] \Downarrow \MMM' \ENTAILS \MMM'\models C)\\
& \LITEQ & \forall L'\in \mathcal{F},\forall N.
(\updates{\MMM}{x}{!x}[u:(N;x:=L')] 
\Downarrow \MMM' \ENTAILS \MMM'\models C)\\
& \LITEQ & \forall L'\in \mathcal{F},\forall N.
(\MMM[u:(N;x:=L')] 
\Downarrow \MMM' \ENTAILS \MMM'\models C)\\
& \Rightarrow & 
\MMM\models \allworlds \allCon{x}{C} 
\end{array}
\]
In Line 5, we use the fact $!x$ is a functional term. 
In Line 6, we note that $\updates{\MMM}{x}{!x}\LITEQ \MMM$. 
The equation for $\someCon{x}{C}$ is similar. 
This concludes the proofs.

\subsection{Axioms for Content Quantification}
\label{subsec:axioms:cq}
The axiomatisation of content quantification in \cite{ALIAS} uses 
the well-known axioms \cite[\S 2.3]{MENDELSON} for standard
quantifiers. Despite the  presence of local state, most of the axioms stay
valid. 
\begin{proposition}[Axioms for Content Quantifications] 
\label{pro:axiom:content}
Recall $A$ denotes the stateless formula.  
\begin{enumerate}[\em(1)]
\item $[!x]A \LITEQ A$
\item $[!x]!y=z \LITEQ x \not= y \AND !y = z$ 
\item $\allCon{x}{([!x]C_1\ENTAILS C_2)}\ENTAILS([!x]C_1\ENTAILS \allCon{x}{C_2})$.
\item $[!x][!x]C\LITEQ [!x]C$
\item $[!x][!y]C\LITEQ [!y][!x]C$
\item $[!x](C_1\AND C_2)\LITEQ [!x]C_1\AND [!x]C_2$
\item $[!x]C_1\OR [!x]C_2\ENTAILS [!x](C_1\OR C_2)$
\end{enumerate}
\end{proposition}
\begin{proof}
For (1), assume $\MMM\models \allworlds A$. By definition, 
for all $N$, if $\MMM[u:N]\Downarrow \MMM'$, then $\MMM'\models A$. 
This implies:  
for all $V$ and $L\in \mathcal{F}$, 
if $\MMM[u:x:=V;L]\Downarrow \MMM'$, then $\MMM'\models A$, 
which means $\MMM\models \allCon{x}{A}$. 
Others are 
proved as in \cite[Appendix C]{ALIAS}.
\end{proof}

\subsection{Proof of Proposition \ref{pro:hiding}}
\label{app:pro:hiding}
\begin{propositionapp}{Axioms for $\forall$, $\exists$ 
and $\HIDEe$} 

Below we assume there is no capture of variables in types and formulae.
\begin{enumerate}[\em(1)]
\item {\rm (introduction)}
$C \ENTAILS \HIDEe x.C$ if $x\not\in\FV{C}$ 

\item {\rm (elimination)}
$\HIDEe x.C \LITEQ C$
if $x\not\in\FV{C}$ and $C$ is monotone. 

\item 
For any $C$ we have
$C \ENTAILS \exists x.C$.
Given $C$ such that $x\not\in \FV{C}$
and $C$ is thin with respect to $x$, we have
$\exists x.C \ENTAILS C$. 

\item 
For any $C$ we have $\forall x.C \ENTAILS C$.
For $C$ such that $x\not\in \FV{C}$ and $C$ is thin with respect to $x$,
we have $C \ENTAILS \forall x.C$. 


\item $\HIDEe x.(C_1 \AND C_2) \ENTAILS \HIDEe x. C_1 \AND \HIDEe x. C_2 $. 
\item $\HIDEe x.(C_1 \OR C_2) \LITEQ \HIDEe x.C_1 \OR \HIDEe x.C_2$. 

\item $\HIDEe y.\forall x.C \ENTAILS \forall x.\HIDEe y.C$   

\item $\exists x.\HIDEe y. C \ENTAILS \HIDEe y. \exists x.C$  
and $\exists x^\alpha.\HIDEe y. C \LITEQ \HIDEe y.\exists x^\alpha.C$
with $\alpha\in \ASET{\UNIT,\BOOL,\NAT}$. 

\item 
$\HIDEe y.\HIDEa x.C \ENTAILS \HIDEa x.\HIDEe y.C$; 
and 
$\HIDEe y.\HIDEe x.C \LITEQ \HIDEe x.\HIDEe y.C$.    

\item $\HIDEe y.\exists\TVX.C \LITEQ \exists\TVX.\HIDEe y.C$; and 
$\HIDEe y.\forall\TVX.C \ENTAILS \forall\TVX.\HIDEe y.C$. 

\item 
$\HIDEe y.\allCon{x}{C} \ENTAILS \allCon{x}{\HIDEe y.C}$
and  
$\HIDEe y.\someCon{x}{C} \ENTAILS \someCon{x}{\HIDEe y. C}$
\end{enumerate}
\end{propositionapp}
(1) is by definition. For (2), we have: 
\[
\begin{array}{rcll}
\MMM\models \HIDEe x.C 
& \Rightarrow
& \exists \MMM',l.((\new l)\MMM' \cong \MMM \ANDl \MMM'[x:l]\models C)\\
& \Rightarrow
& \exists \MMM',l.((\new l)\MMM' \cong \MMM \ANDl \MMM'\models C) &
\text{Lemma \ref{lemma:29410} (\ref{lemma:29410:6})}\\
& \Rightarrow
& (\new l)\MMM'\models C 
& \text{$C$ is monotone}
\end{array}
\]
For (7), we derive: 
\[
\begin{array}{lclllll}
\MMM\models \HIDEe y. \forall x.C 
&\!\LITEQ\! & \exists \MMM_0,
\forall L\in \mathcal{F}.(\MMM\WB (\new l)\MMM_0 \AND 
(\MMM_0[x:L]\Downarrow \MMM_0'\ENTAILS \MMM_0'\models C))\\
&\!\Rightarrow\! & 
\forall L\in \mathcal{F}, \exists \MMM_0.(
\MMM[x:L]\WB ((\new l)\MMM_0)[x:L] \AND \MMM_0[x:L]\models C) 
\mbox{ such that } l\not\in \FL{L}\\
& & \MMM\models \forall x. \HIDEe y. C
\end{array}
\]
For (8-1), we derive:
\[
\begin{array}{lllllll}
\MMM\models \exists x.\HIDEe y. C 
& \LITEQ 
& \exists L\in \mathcal{F}.(\MMM[x:L]\Downarrow \MMM' \AND 
\exists \MMM_0.(\MMM'\WB (\new l)\MMM_0 \AND \MMM_0[y:l]\models C)\\
& \Rightarrow 
& \exists L\in \mathcal{F},\MMM_0.(\MMM\WB (\MMM[x:L])/x \WB 
((\new l)\MMM_0)/x \AND \MMM_0[y:l]\models C)\\
& \LITEQ & \exists L\in \mathcal{F},\MMM_0'.(\MMM\WB 
(\new l)\MMM_0' \AND \MMM_0'[y:l][x:L]\models C)\quad \mbox{with} \ \MMM_0'=\MMMM/x\\
& \LITEQ &
\MMM\models \HIDEe y. \exists x.C
\end{array}
\]
Note that the other direction does not generally hold. 
Consider $\MMM\models \HIDEe y. \exists x.C$. This is equivalent to: 
\[ 
\exists L\in \mathcal{F},\MMM_0'.(\MMM\WB 
(\new l)\MMM_0 \AND \MMM_0[y:l][x:L]\models C)
\] 
Since $L$ might contain the new reference $l$ hidden in $\MMM$, 
$\MMM[x:L\MSUBS{l}{y}]$ is undefined (hence we cannot permute 
$[y:l]$ and $[x:L\MSUBS{l}{y}]$). 

For (8-2), we only have to prove 
$\HIDEe y.\exists x^\alpha.C
\ENTAILS \exists x^\alpha.\HIDEe y. C$ with 
$\alpha\in \ASET{\UNIT,\BOOL,\NAT}$.  We derive:
\[
\begin{array}{lllllll}
\MMM\models \HIDEe y.\exists x.C
&\LITEQ 
\exists \MMM_0,\mathsf{c}.(\MMM\WB (\new l)\MMM_0\AND \MMM_0[y:l][x:\mathsf{c}]
\models C)\\
&\LITEQ 
\exists \mathsf{c},\MMM_0.
(\MMM[x:\mathsf{c}]\WB ((\new l)\MMM_0)[x:\mathsf{c}]\AND 
\MMM_0[x:\mathsf{c}][y:l]\models C)
& \LITEQ  \MMM\models \exists x.\HIDEe y.C
\end{array}
\]
For (9-2), we have:
\[
\begin{array}{lllllll}
\MMM \proves \new x.\new y.C 
& \LITEQ & \exists \MMM'.(\MMM \WB (\new l)\MMM' \AND \MMM'[x:l]
\models \new y.C )\\
& \LITEQ & \exists \MMM',\MMM''.(\MMM \WB (\new l)\MMM' \AND
\MMM'[x:l]
\WB (\new l')\MMM''\AND \MMM''[y:l']\models  C )\\
& \LITEQ & \exists \MMM',\MMM''.(\MMM \WB (\new l)\MMM' 
\WB (\new ll')\MMM'' \AND \MMM''[x:l] \WB (\new l')\MMM''\AND
\MMM''[y:l']\models  C)\\
& \LITEQ & \MMM \proves \new y.\new x.C 
\end{array}
\]
For (11), we derive: 
\[
\begin{array}{lclllll}
\MMM \proves \new y.\allCon{x}C 
& \LITEQ & \exists \MMM_0.(\MMM \WB (\new l)\MMM_0 \AND \forall L\in
\mathcal{F}.(\updates{\MMM_0[y:l]}{x}{L}\models C)) \\
& \LITEQ & \forall L\in
\mathcal{F}.\exists \MMM_0.(\MMM \WB (\new l)\MMM_0 \AND
\updates{\MMM_0[y:l]}{x}{L}\models C))\\
&&\qquad\mbox{ such that } l,y\not\in \FV{L}\cup\FL{L}\\
& \LITEQ &
\forall L\in
\mathcal{F}.\exists \MMM_0.
(\updates{\MMM}{x}{L} \WB (\new l)(\updates{\MMM_0}{x}{L}) \AND
\updates{\MMM_0}{x}{L}[y:l]\models C))\\
& \LITEQ &
\MMM \proves \allCon{x}\HIDEe y.C 
\end{array}
\]
The remaining claims are similar.

\subsection{Proof of Theorem \ref{theorem:elimination}}
\label{app:elimination}
\begin{theoremapp}{\ref{theorem:elimination}}{}  
Suppose all reachability predicates in $C$ are finite. 
Then there exists $C'$ such that $C\LITEQ C'$ and 
no reachability predicate occurs in $C'$. 
\end{theoremapp}
As the first step, we define a simple
inductive method for defining reachability from a datum of a finite
type. 

\begin{definition}{\rm ($i$-step reachability)} \
Let $\alpha$ be a finite type. Then the $i$-step reachability 
predicate $\iReachable{x^\alpha}{y^{\refType{\beta}}}{i^{\NAT}}$ 
(read:{\em ``a reference $y$ is reachable from $x$ in at most $i$-steps''})
is inductively
given as follows (below we assume $y$ is typed
$\refType{\beta}$, $C\in\ASET{\UNIT, \BOOL, \NAT}$, and omit types
when evident). 
\begin{eqnarray*}
\iReachable{x^\alpha}{y}{0} 
& \LITEQ & 
x=y\\
\iReachable{x^C}{y}{n+1} 
& \LITEQ & 
\falsity\\
\iReachable{x^{\alpha_1\times\alpha_2}}{y}{n+1} 
& \LITEQ & 
\OR_i \iReachable{\pi_i(x)}{y}{n}\ \OR\ \iReachable{x}{y}{n}\\
\iReachable{x^{\alpha_1+\alpha_2}}{y}{n+1} 
& \LITEQ & 
\exists x'.(x'=\INJtp{1}{x}{}\ \ANDl\ \iReachable{x'}{y}{n} )\ \OR\\
& & 
\exists x'.(x'=\INJtp{2}{x}{}\ \ANDl\ \iReachable{x'}{y}{n} )\ \OR\\
& & 
\iReachable{x}{y}{n}\\
\iReachable{x^{\refType{\alpha}}}{y}{n+1} 
& \LITEQ & 
\iReachable{!x}{y}{n} \ \OR\
\iReachable{x}{y}{n}
\end{eqnarray*} 
\end{definition}
With $C$ being a base type,
$\iReachable{x^C}{y}{0}\LITEQ x=y \LITEQ \falsity$
(since a reference $y$ cannot be equal to a datum of a base type).

A key lemma follows.

\begin{proposition} \label{prop:A}
If $\alpha$ is finite,
then the logical equivalence 
$x^\alpha\reachable y
\IFF
\exists i.\iReachable{x^\alpha}{y}{i}$ 
is valid, i.e. is true in any model.
\end{proposition}
\begin{proof} 
For the ``if'' direction,
we show, by induction on $i$, 
$\iReachable{x^\alpha}{y}{i} 
\ENTAILS
x^\alpha\reachable y$.
For the base case, we have $i=0$, in which case 
$\iReachable{x^\alpha}{y}{0} 
\ENTAILS 
x=y
\ENTAILS 
x\reachable y$. 

For induction, let the statement holds up to $n$. We only show the case of a
product. Other cases are similar.
\begin{reasoning}
\iReachable{x^{\alpha_1\times\alpha_2}}{y}{n+1} 
&\;\;\THENS\;\;&
\OR_i \iReachable{\pi_i(x)}{y}{n}\ \OR\
\iReachable{x}{y}{n}\\
&\;\;\THENS\;\;&
\OR_i \pi_i(x)\reachable y\ \OR\
x\reachable y
\end{reasoning}
But if 
$\pi_1(x)\reachable y$ 
then 
$x\reachable y$
by the definition of reachability. 
Similarly when $\pi_2(x)\reachable y$, hence done.

For the converse, we show the contrapositive, showing:
\[
\MMM\models \neg \exists i.\iReachable{x^\alpha}{y}{i} 
\quad\THENS\quad
\MMM\models \neg x^\alpha\reachable y
\]
If we have
$\MMM\models \neg \exists i.\iReachable{x^\alpha}{y}{i}$
with $\alpha$ finite,
then 
the reference $y$ is not among references
reachable from $x$ (if it is, then either $x=y$
or $y$ is the content of a reference reachable from $x$
because of the finiteness of $\alpha$, so that
we can find some $i$ such that $\MMM\models \iReachable{x^\alpha}{y}{i}$), 
hence done. 
\end{proof}
Now let us define the predicate $x^\alpha\reachable^\circ y^{\refType{\beta}}$
with $\alpha$ finite, by the axioms given in 
Proposition \ref{pro:notreach}
which we reproduce below ($C\in\ASET{\UNIT, \BOOL, \NAT}$).
\[
\begin{array}{rcl}
       \reachB{x^C}{y^{\refType{\beta}}}
                & \LITEQ & 
        \falsity\\
       \reachB{x^{\alpha_1\times\alpha_2}}{y^{\refType{\beta}}}
                & \LITEQ & 
       \exists x_{1, 2}.(x = \PAIR{x_1}{x_2} \AND \OOR_{i = 1, 2} \reachB{x_i}{y})\\
       \reachB{x^{\alpha_1 + \alpha_2}}{y^{\refType{\beta}}}
                & \LITEQ & 
       \exists x'.( (\OOR_{i = 1, 2} x=
\INJtp{i}{x'}{})\ANDl \reachB{x'}{y})\\
        \reachB{x^{\refType{\alpha}}}{y^{\refType{\beta}}}
                & \LITEQ & 
        x = y \ \OR \ \reachB{!x}{y}
\end{array}
\] 
The inductive definition is possible due to finiteness. We now show:

\begin{proposition}\label{prop:B}
If $\alpha$ is finite, then the logical equivalence,  
$\reachB{x^{\alpha}}{y^{\refType{\beta}}}
\LITEQ
\exists i.\iReachable{x^{\alpha}}{y^{\refType{\beta}}}{i}$, is
valid.
\end{proposition}
\begin{proof}
$\iReachable{x^{\alpha}}{y^{\refType{\beta}}}{i}
\ENTAILS
\reachB{x^{\alpha}}{y^{\refType{\beta}}}$
is by induction on $i$. The converse is by induction on $\alpha$.
Both are mechanical and omitted.
\end{proof}

\begin{corollary}
If $\alpha$ is finite, then the logical equivalence
$
\REACH{x^{\alpha}}{y^{\refType{\beta}}}
\LITEQ
\reachB{x^{\alpha}}{y^{\refType{\beta}}}
$ 
is valid, i.e. $\reachable$ is completely characterised by the axioms 
for $\reachB{}{}$ given above.
\end{corollary}
\begin{proof}
Immediate from Propositions \ref{prop:A} and  \ref{prop:B}.
\end{proof}

\subsection{Proof of Proposition \ref{pro:notreachfunc}}
\label{app:notreachfunc}
\begin{propositionapp}{\ref{pro:notreachfunc}}{}  
For an arbitrary $C$, the following is valid
with $i, \TVX$ fresh: 
$$
\allworlds\ASET{C \AND x\noreach
fy\VEC{w}}\APPs{f}{y}\!=\!z\ASET{C'}@\VEC{w} 
\; \ENTAILS\; 
\allworlds \forall \TVX,
i^{\TVXscript}.  \ASET{C \AND x\noreach fiy\VEC{w}}\APPs{f}{y}\!=\!z
\ASET{C'\AND x\noreach fiyz\VEC{w}}@\VEC{w}$$
\end{propositionapp}

\begin{proof} 
The proof traces the transition of state using the elementary fact
that the set of names and labels in a term always gets smaller as
reduction goes by. 
Suppose we have
$$
\MMM\models \allworlds \ASET{x\!\notreach\!fyw\AND
C}\APPs{f}{y}\!=\!z\ASET{C'}@w$$
The definition of the evaluation formula says:
\begin{center}
$
(\MMM\EVOLVES \MMM_0 \ANDl 
\MMM_0\models x\!\notreach\!fywi\AND C)\ 
\ENTAILS \ 
\exists \MMM'. (\MMM[z:fy]\converges\MMM'\AND \MMM'\models C')$. 
\end{center}
We prove such $\MMM'$ always satisfies  
$\MMM'\models x\notreach fiyzw$. 
Assume 
$$\MMM_0 \WB (\new \VEC{l})(\xi, \sigma_0\uplus \sigma_x)$$ 
with $\xi(x)=l$, $\xi(y)=V_y$, $\xi(f)=V_f$ and  
$\xi(w)=l_w$ such that 
$$\ncl{\FL{V_f,V_y,l_w}}{\sigma_0\uplus \sigma_x}=
\FL{\sigma_0}=\dom{\sigma_0}$$ 
and $l_x\in \dom{\sigma_x}$. 
By this partition, during evaluation of $z:fy$, 
$\sigma_{x}$ is unchanged, i.e.  
$$
(\new \VEC{l})(\xi\cdot z:fy, \sigma_0\uplus \sigma_x)
\MRED (\new \VEC{l})(\xi\cdot z:V_fV_y, \sigma_0\uplus \sigma_x)
\MRED (\new \VEC{l}')(\xi\cdot z:V_z, \sigma_0'\uplus \sigma_x)
$$
Then obviously there exists $\sigma_1$ such that 
$\sigma_1\subset \sigma_0'$ and 
$$\ncl{\FL{V_z,l_w}}{\sigma_0'\uplus \sigma_x}=
\FL{\sigma_1}=\dom{\sigma_1}$$
Hence by Proposition \ref{pro:partition}, we have 
$\MMM_0\models x\!\notreach\! fyiwz$, completing the proof. 
\end{proof}

\subsection{Proof of Propositions \ref{pro:localinv:firstorder}}
\label{app:proof:AHI}

\begin{propositionapp}{\ref{pro:localinv:higherorder}}{}
Assume 
$C_0\!\LITEQ C'_0\AND\VEC{x}\notreach iy\AND \VEC{g}\reachable^\bullet \VEC{x}$,\
$C'_0$ is stateless except $\VEC{x}$,\    
$C$ is anti-monotone,
$C'$ is monotone,
$i,m$ are fresh and 
$\{\VEC{x},\VEC{g}\} \cap (\FV{C, C'}\cup\ASET{\VEC{w}})=\emptyset$.  
Then the following is valid: 
\[
\begin{array}{ll}
\text{\AIHax}\quad \quad 
&
\forall \TVX.\forall i^\TVX.
\APP{m}{()}\!=\!u
\ASET{(\new\VEC{x}.\exists \VEC{g}.E_1)\AND E} 
\; \ENTAILS \; 
\forall \TVX.\forall i^\TVX.
\APP{m}{()}\!=\!u\ASET{E_2\AND E} 
\end{array}
\]
with
\begin{enumerate}[$\bullet$]
\item 
$E_1 
\LOGICEQ 
\LocalInv{u}{C_0}{\VEC{x}}
\AND 
\allworlds \forall yi.\eval{C_0\ANDs C}{u}{y}{z}{C'}@\VEC{w}\VEC{x}$\ 
\ and
\item   
$E_2 \LOGICEQ 
\allworlds 
\forall y.\eval{C}{u}{y}{z}{C'}@\VEC{w}$.
\end{enumerate}
\end{propositionapp}

\begin{proof}
W.l.o.g. we assume all vectors are unary, setting
$\VEC{r}=r$, 
$\VEC{w}=w$,
$\VEC{x}=x$  and 
$\VEC{g}=g$. 
The proof proceeds as follows, starting from the current model $\MMM_0$.
\begin{enumerate}[{\bf St{a}ge 1.}]
\item
We take $\MMM$ such that:
$$
\MMM_0\EVOLVESinc{m}\MMM
$$
We then take off the hiding, name it $x$ and the result is called
$\MMM_{\ast}$
$$
(\new l)(\MMM_{\ast}/x )\WB \MMM.
$$
\item
We further let $\MMM$ evolve so that:
$$
\MMM\EVOLVESinc{u}\MMM'
$$
We then again take off the corresponding hiding, name it $x$ and the result is called
$\MMM'_{\ast}$
$$
(\new l)(\MMM'_{\ast}/x )\WB \MMM'
$$
\item
We show if $\MMM_\ast$ satisfies $C_0$ then again $\MMM'_\ast$ satisfies $C_0$ again:
$$
\MMM_\ast\models C_0
\qquad\ENTAILS\qquad
\MMM'_\ast\models C_0
$$
using $\LocalInv{u}{C_0}{\VEC{x}}$ as well as the unreachability of $x$ from $u$.
\end{enumerate}
By reaching Stage 3, we know if 
$\MMM\models C$ then it is also the case
$\MMM_\ast\models C_0\AND C$ hence we can use the assumption 
(together with monotonicity of $C'$): 
$$
\forall yi.\eval{C_0\ANDs C}{u}{y}{z}{C'}@\VEC{w}\VEC{x} 
$$
hence we know we arrive at $C'$ as a result.

We now implement these steps. We set:
\begin{eqnarray}
E   & \LITEQ &\truth.
\end{eqnarray}
The trivialisation of $E$ (taken as truth) 
is just for simplicity and does not affect the argument. Now fix an arbitrary $\MMM_0$ and suppose we reach:
\begin{equation}
\MMM_0\EVOLVESinc{m}\MMM
\end{equation}
This gives the status of the post-condition of the whole formula
(to be precise this is through the encoding in (\ref{therethereagain})
in \S~\ref{subsec:soundness} to relate $\APP{m}{()}$ and 
the transition above). Assuming the hidden $x$ in the formula in $E_1$ is about 
a (fresh) $l$ we can set:
\begin{equation}
\MMM
\;\DEFEQ\;
(\new l)(\new \VEC{l}')(\xi,{\sigma}\cdot[{l}\mapsto{V}])
\;\models\;
\new x.\exists g.E_1
\end{equation}
as well as by revealing $l$:
\begin{equation}
\MMM_\ast
\;\DEFEQ\;
(\new \VEC{l}')(\xi\CD x:l\CD g:U,{\sigma}\cdot[{l}\mapsto{V}])
\;\models\;
E_1
\end{equation}
Note by assumption we have:
\begin{equation}
l\not\in\FL{\xi,\sigma}.
\end{equation}
Further $U$ does not contain any hidden or free locations from $\MMM$
by $g\reach^\bullet\VEC{x}$.

Now we consider 
the right-hand side of $E_1$, 
$\allworlds \forall yi.\eval{C_0\ANDs C}{u}{y}{z}{C'}@\VEC{w}\VEC{x}$
by taking for fresh $N$:
\begin{equation}
\MMM[f:N] \converges \MMM'
\end{equation}
Corresponding to the relationship between $\MMM$ and $\MMM_\ast$ we set:
\begin{equation}
\MMM_\ast[f:N] \converges \MMM'_\ast
\end{equation}
Note we have
\begin{equation}
(\new l)(\MMM'_\ast/xg) \WB \MMM'
\end{equation}
We now show:
\begin{equation}
\MMM_\ast \models C_0 
\quad\ENTAILS\quad
\MMM'_\ast \models C_0 
\end{equation}
that is $C_0$ is invariant under the evaluation (effects) of $N$.
Assume 
\begin{equation}
\MMM_\ast\models C_0
\end{equation}
First observe 
\begin{equation}
\MMM_\ast\models C_{0} \ANDl {x} \notreach y{r}{w}
\end{equation}
Now in the standard way $N$ can be approximated by a finite term, that is a term which
does not contain recursion except divergent programs. We take $N$ as such an approximation
without loss of generality. Such $N$ can be
written as a sequence of let expressions including assignments.
Without loss of generality we focus on a ``let'' expression which is
either a function call or an assignment. Then at each evaluation we have either:
\begin{enumerate}[$\bullet$]
\item The let has the form $\LET{x}{uV}{M'}$ that is it invokes $u$; 
\item The let has the form $\LET{x}{WV}{M'}$ where $W$ is not $u$.
\item The let has the form $w':=V;M'$.
\end{enumerate}
We observe:
\begin{enumerate}[$\bullet$]
\item
In the first case $u$ is directly invoked: thus
by the invariance $\LocalInv{u}{C_0}{\VEC{x}}$,
$C_0$ continues to hold. Note $w'$ is not $x$ since 
$N$ has no access to $x$ except through $u$.
\item
In the second case of the let (i.e. $u$ is not called),
since $x$ is disjoint from 
all visible data, by Proposition \ref{pro:notreachfunc} 
we know $x$ (hence the content of $x$) is never touched by 
the execution of the function body after the invocation,
until again
$u$ is called (if ever): since
$C_0$ is insensitive to state change except at $x$ (by
being stateless except $x$),
it continues to hold again in this case.
\item
In the third case again $x$ is not touched hence $C_0$ continues to hold.
\end{enumerate}
Thus we have:
\begin{equation} 
\MMM'_\ast\models C_0
\end{equation}
Now suppose we have
\begin{equation} 
\MMM\models C
\end{equation}
By anti-monotonicity of $C$ we have
\begin{equation} 
\MMM_\ast/xg\models C
\end{equation}
By Lemma \ref{lemma:29410} (\ref{lemma:29410:6}), 
we can arbitrarily weaken a disjoint extension (at $x$ and $g$) so that:
\begin{equation} 
\MMM_\ast\models C
\end{equation}
Thus we know:
\begin{equation} 
\MMM'_\ast\models C_0\ANDl C
\end{equation}
Now we can apply:
\begin{equation} 
\MMM'\models 
\new x.\exists g.\forall y.\eval{C_0\ANDs C}{u}{y}{z}{C'}@\VEC{w}x\ 
\end{equation}
by which we know:
\begin{equation} \label{thisisitthisisit}
\MMM'_\ast[z:uy]\converges \MMM''_\ast \models C'
\end{equation}
Accordingly let
\begin{equation} \label{thisisitthisisitanother}
\MMM'[z:uy]\converges \MMM'' \WB (\new l)(\MMM''_\ast/x)
\end{equation} 
for which we know, by (\ref{thisisitthisisit}) and (\ref{thisisitthisisitanother})
together with monotonicity of $C'$:
\begin{equation} 
\MMM'' \models C'
\end{equation}
Hence we know:
\begin{equation} \label{wwwC}
\MMM\models \ASET{C}\APP{u}{y}=z\ASET{C'}@{w}
\end{equation}
which is the required assertion.
\end{proof}

\subsection{Proof of Proposition \ref{pro:nuelim}}
\label{proof:pro:nuelim}
\NI 
\begin{propositionapp}{\ref{pro:nuelim}}{}
Let $x\not\in \FV{C}$ and $m, i, \TVX$ be fresh. Then the following is valid:
\[
\begin{array}{l}
\forall \TVX, i^\TVXscript.
\APP{m}{()}\!=\!u
\ASET{\HIDEe \VEC{x}.([!\VEC{x}]C\AND \VEC{x}\noreach ui^\TVXscript)}
\;\; \ENTAILS\;\; 
\APP{m}{()}\!=\!u
\ASET{C}
\end{array}
\]
\end{propositionapp}
\begin{proof}
For simplicity, set $\VEC{x}$ to be a singleton $x$.
Assume 
$$
\MMM[u:m()]
\converges 
\MMM'
$$
By assumption we can set
$$
\MMM'\WB (\new l)(\new \VEC{l}')(\xi\CD u:V, \sigma\CD l\mapsto W)
$$
such that
$$
(\new \VEC{l}')(\xi\CD u:V\CD x:l, \sigma\CD l\mapsto W)\models
[!{x}]C
$$
where $l$ is not reachable from anywhere else in the model.
By Lemma \ref{lem:subs} we obtain
$(\new \VEC{l}')(\xi\CD u:V, \sigma)\models C$, 
that is 
$\MMM'\models C$, 
as required.
\end{proof}

\subsection{Proof of Proposition \ref{pro:transfer}}
\label{proof:pro:transfer}
Assume  
${C_0}$ is 
stateless 
except
${\VEC{x}}$
and suppose:
\begin{equation}
\MMM\models
\LocalInv{f}{C_0}{\VEC{x}}
\ANDl
\ASET{\truth}\APP{g}{f}=z\ASET{\truth}.
\end{equation}
Further assume $\MMM\EVOLVES \MMM_0$ and 
\begin{equation}
\MMM_0\models {C_0\AND \VEC{x}\noreach g\VEC{r}}
\ \;\text{and}\;\ 
\MMM_0[z:fg]\converges \MMM'.
\end{equation}
By $\LocalInv{f}{C_0}{\VEC{x}}$ we know that once
$C_0$ holds and $f$ is invoked, it continues to hold.  By
$\ASET{\truth}\APP{g}{f}=z\ASET{\truth}$, we know the application
$gf$ always terminates. Now this application invokes $f$ zero or more
times. First time it can only apply $f$ to some $\VEC{x}$-unreachable
datum. Similarly for the second time, since the context cannot obtain
$\VEC{x}$-reachable datum (given $g$ itself is $\VEC{x}$-unreachable).
By induction the same holds up to the last invocation. In each
invocation, $C_0$ is invariant. Further, other
computations in $fg$ never touch the content of $\VEC{x}$,
hence because of ${C_0}$ being stateless except $\VEC{x}$,
we know $C_0$ is again invariant in such computations. Thus we conclude
that $C_0$ still holds in the post-condition, and that the return value
being $\VEC{x}$-unreachable, i.e. $\VEC{x}\noreach z$, as required.\qed

\section{Derivations for Examples in Section \ref{sec:example:1}}
\label{app:examples1}

\NI This appendix lists the derivations omitted in
Section \ref{sec:example:1}.

\subsection{Derivation for $\mathtt{mutualParity}$}

\label{app:mutualfact}
\NI Let us define: 
\[
\begin{array}{rcll}
M_x & \DEFEQ & \lambda n.\IFTHENELSE{y=0}{\mathtt{f}}{\mathtt{not}((!y)(n-1))}\\
M_y & \DEFEQ & \lambda n.\IFTHENELSE{y=0}{\mathtt{t}}{\mathtt{not}((!x)(n-1))}
\end{array}
\]
We also use: 
\[
\begin{array}{lllll}
IsOdd'(u,gh,n,xy) \;\ \LOGICEQ IsOdd(u,gh,n,xy) 
\AND\; !x=g \;\AND\; !y=h\\
IsEven'(u,gh,n,xy) \;\ \LOGICEQ\;
IsEven(u,gh,n,xy) \AND\; !x=g \;\AND\; !y=h
\end{array}
\]
{\small
\begin{myfiguretwo}
\begin{inference}
1.\quad
\ASET{
(n\geq 1\ENTAILS IsEven'(!y,gh,n-1,xy)) 
\ANDl n=0}
\ \mathtt{f}:_z
\ASET{z=Odd(n)\ANDl !x=g \ANDl !y=h}@\emptyset\\ 
\thankstoBig{Const}
2.\quad
\ASET{(n\geq 1\ENTAILS IsEven'(!y,gh,n-1,xy)) 
 \ANDl n\geq 1
}\\
\quad \quad \quad \mathtt{not}((!y)(n-1)):_z
\ASET{z=Odd(n)\ANDl !x=g \ANDl !y=h}@\emptyset 
\thankstoBig{Simple, App}
3.\quad
\ASET{
n\geq 1\ENTAILS IsEven'(!y,gh,n-1,xy)}\\ \quad \quad \quad 
 \IFTHENELSE{n=0}{\mathtt{f}}{\mathtt{not}((!y)(n-1))}:_m
\ASET{z=Odd(n)\ANDl !x=g \ANDl !y=h}@\emptyset 
\thankstoBig{IfH}
4.\quad 
\ASET{\truth}
\ \lambda n.\IFTHENELSE{n=0}{\mathtt{f}}{\mathtt{not}((!y)(n-1))}:_u\\
\quad \quad \quad
\ASET
{ \ 
\allworlds \forall gh, 
n\geq 1.\EVALEFFECT{IsEven'(h,gh,n-1,xy)}{u}{n}{z}
{z=Odd(n)\;\AND\; !x=g \;\AND\; !y=h}
{\emptyset}
\}
@\emptyset\\
\thankstoBig{Abs, $\forall$, Conseq}
5.\quad 
\ASET{\truth}
\ M_x :_u
\ASET
{ \ 
\forall gh, 
n\geq 1.(IsEven(h,gh,n-1,xy) \ENTAILS  IsOdd(u,gh,n,xy))}@\emptyset 
\thankstoBig{Conseq}
6.\quad 
\ASET{\truth}
\ x:=M_x
\ASET
{ \ 
\forall gh, 
n\geq 1.(IsEven(h,gh,n-1,xy) \ENTAILS  IsOdd(!x,gh,n,xy))\ANDl !x=g}@x\\
\thankstoBig{Assign}
7.\quad 
\ASET{\truth}
\ y:=M_y 
\ASET
{ \ 
\forall gh, 
n\geq 1.(IsOdd(g,gh,n-1,xy) \ENTAILS  IsEven(!y,gh,n,xy))\ANDl !y=h}@y
\nothanks
8.\quad 
\ASET{\truth}
\ \mathtt{mutualParity}\\
\quad\quad\{
\forall gh.n\geq 1.(
(IsEven(h,gh,n-1,xy)\AND 
IsOdd(g,gh,n-1,xy) ) \ENTAILS\\  
\quad\quad\quad
(IsEven(!y,gh,n,xy)\AND 
IsOdd(!x,gh,n,xy) \AND !x=g \AND !y=h) \ }@xy
\thankstoBig{$\AND$-Post}
9.\quad 
\ASET{\truth}
\ \mathtt{mutualParity}\\
\quad\quad\{
\forall n\geq 1gh.(
(IsEven(h,gh,n-1,xy)\AND 
IsOdd(g,gh,n-1,xy) \AND !x = g \AND !y = h) \ENTAILS\\  
\quad\quad\quad
(IsEven(!y,gh,n,xy)\AND 
IsOdd(!x,gh,n,xy)\AND !x = g \AND !y= h) \}@xy
\thankstoBig{Conseq}
10.\quad 
\ASET{\truth}
\ \mathtt{mutualParity}\\
\quad\quad\{
\forall n\geq 1gh.(
(IsEven(!y,gh,n-1,xy)\AND 
IsOdd(!x,gh,n-1,xy) \AND !x = g \AND !y = h) \ENTAILS\\  
\quad\quad\quad
(IsEven(!y,gh,n,xy)\AND 
IsOdd(!x,gh,n,xy)\AND !x = g \AND !y= h) \}@xy
\thankstoBig{Conseq}
11.\quad 
\ASET{\truth}
\ \mathtt{mutualParity}\\
\quad\quad\{
\forall n\geq 1.(
\exists gh.(IsEven(!x,gh,n-1,xy)\AND 
IsOdd(!y,gh,n-1,xy) \AND !x = g \AND !y = h) \ENTAILS\\  
\quad\quad\quad
\exists gh.(IsEven(!y,gh,n,xy)\AND 
IsOdd(!x,gh,n,xy)\AND !x = g \AND !y= h) \}@xy
\thankstoBig{Conseq}
12.\quad 
\ASET{\truth}
\ \mathtt{mutualParity}
\{\exists gh.IsOddEven(gh,!x!y,xy,n) \}@xy
\end{inference}
\caption{$\mathtt{mutualParity}$ derivations \label{fig:mutualparity}}
\end{myfiguretwo}
}\unskip%
\noindent{}Figure \ref{fig:mutualparity} lists the derivation for 
$\mathtt{MutualParity}$.
In Line 4,  $h$ in the evaluation formula
can be replaced by $!y$ and vice versa 
because of $!y = h $ and the universal 
quantification of $h$. 
\[ 
\forall h.(!y = h \AND \{ C \}h\bullet n = z  \{ C' \}) 
\quad \LITEQ \quad  
\forall h.(!y = h \AND \{ C \}(!y)\bullet n = z  \{ C' \}) 
\]
In Line 5, we use the following axiom for the evaluation 
formula from \cite{GLOBAL}: 
\[ 
\ASET{C\AND A}\ \APP{e_1}{e_2}=z\ASET{C'}
\quad \LITEQ \quad  
A\;\ENTAILS\;\ASET{C}\APP{e_1}{e_2}=z\ASET{C'}
\]
where $A$ is stateless
and we set $A= IsEven(h,gh,n-1,xy)$. 
Line 9 is derived as Line 4 by replacing $h$ and $g$ by $!y$ and $!x$,
respectively.
Line 11 is the standard logical implication 
($\forall x.(C_1 \ENTAILS C_2) 
   \ENTAILS (\exists x.C_1 \ENTAILS \exists x.C_2)$).

\subsection{Derivation for Meyer-Sieber}
\label{app:MS} 
\NI 
For the derivation of (\ref{MSassert}) we use: 
\begin{eqnarray*}
E & \LOGICEQ& \forall f.
(
\allworlds \ASET{\truth}\APP{f}{()}\ASET{\truth}@\emptyset
\ \ENTAILS \ 
\allworlds \ASET{C}\APP{g}{f} \ASET{C'}
)
\end{eqnarray*}
We use the following $[\mbox{\em LetRef}]$ which is derived 
by $[\mbox{\em Ref}]$ where $C'$ is replaced by $[!x]C'$.
\[
         [\mbox{\em LetRef}]
                 \,
        \infer
         {
                 \ASSERT{C}{M}{m}{C_0}
                 \quad 
                 \ASSERT{[!x]C_0\AND !x=m \AND 
 		  \NOTREACH{x}{\VEC{e}}}{N}{u}{C'}
                 \quad 
                 x \notin \PFN{\VEC{e}}
         }
         {
                 \ASSERT{C}{\LET{x}{\REFPROG{M}}{N}}{u}{\nu x.C'}
         }
\]
with $C'$ think w.r.t. $m$. 
The derivation follows. Below $M_{1,2}$ is the body of the
first/second lets, respectively. 
\begin{DERIVATION}
        \LINE
          {1.}
          {
          \ASET{Even(!x)\AND [!x]C'}
  	\ \IFTHENELSE{even(!x)}{()}{\DIVPROG()} \
  	\ASET{[!x]C'}
  	@\emptyset
          }
  	{(If)}
          \LINE
          {2.}
          {
          \ASET{[!x]C}
  	\ gf\ 
         \ASET{[!x]C'}	  
          }
  	{(cf. \S~\ref{sub:GF})}
          \LINE
          {3.}
          {
          \ASET{Even(!x)\AND [!x]C}
  	\ gf\ 
          \ASET{Even(!x)\AND [!x]C'}	  
          }
  	{(2, Inv)}
          \LINE
          {4.}
          {
          \ASET{E\AND [!x]C \AND Even(!x) \AND x\noreach gi}
  	\mathtt{let} \ f = ... \ \mathtt{in} \ (gf;...) 
  	\ASET{[!x]C'\AND x\noreach i}
          }
  	{(3, Seq, Let)}
          \LINE
          {5.}
          {
          \ASET{E\AND C}
  	 \ \mathtt{MeyerSieber}\ 
          \ASET{\HIDEe x.([!x]C' \AND  x\noreach i)}
  	}
  	{(4, LetRef)}
          \LASTLINE
          {6.}
          {
          \ASET{E\AND C}
  	\ \mathtt{MeyerSieber}\ 
          \ASET{C'}
  	}
 	{(9, Prop.~\ref{pro:nuelim})}
\end{DERIVATION}

\subsection{Derivation for Object}
\label{app:object}
We need the following generalisation. 
The procedure $u$ in {\sf (AIH)} is of  function type
$\alpha\Rightarrow \beta$: 
when values of other types such as $\alpha\times \beta$ 
or $\alpha+ \beta$ are returned,
we can make use of a generalisation.  
For simplicity we restrict our
attention to the case when types do not contain recursive or 
reference types.
\begin{eqnarray*}
\LocalInv{u^{\alpha\times\beta}}{C_0}{\VEC{x}}
&\LOGICEQ&
\AND_{i=1,2}\LocalInv{\pi_i(u)}{C_0}{\VEC{x}}\\[1mm]
\LocalInv{u^{\alpha+\beta}}{C_0}{\VEC{x}}
&\LOGICEQ&
\AND_{i=1,2}\forall y_i.(u=\injection_i(y_i)\ENTAILS
\LocalInv{y_i}{C_0}{\VEC{x}})\\
\LocalInv{u^{\alpha}}{C_0}{\VEC{x}}
&\LOGICEQ&
\truth\qquad (\alpha \in \{ \UNIT,\NAT,\BOOL\})
\end{eqnarray*}
Using this extension, we can generalise {\sf (AIH)} so that 
the cancelling of $C_0$ is possible for all components of $u$. 
For example, if $u$ is a pair of functions, those two functions
need to satisfy the same condition as in {\sf (AIH)}. This is what we
shall use for $\mathtt{cellGen}$. We call the resulting generalised axiom
{\sf (AIH${}_\mathsf{c}$)}.

Let $\mathtt{cell}$ be the internal $\lambda$-abstraction of
$\mathtt{cellGen}$. First, it is easy to obtain:
\begin{equation}
\ASET{\truth}
\  \mathtt{cell}:_o\
\ASET{
I_0 
\ANDl
G_1
\ANDl
G_2
\ANDl
E'
}
\end{equation}
where, with $I_0\LOGICEQ !x_0=!x_1 \ANDl x_0\noreach iw$ 
(noting $x \noreach v \LITEQ \truth$) 
and $E'\LOGICEQ !x_0=z$.
\begin{eqnarray*}
G_1 &\LOGICEQ & \allworlds \ASET{I_0}\APP{\pi_1(o)}{()}=v\ASET{v=!x_0\AND I_0}@\emptyset\\
G_2 &\LOGICEQ & \allworlds \forall w.\ASET{I_0}\APP{\pi_1(o)}{w}\ASET{!x_0=w\AND I_0}@x_0x_1
\end{eqnarray*}
which will become, after taking off the invariant $I_0$:
\begin{eqnarray*}
G'_1 &\LOGICEQ & \allworlds \APP{\pi_1(o)}{()}=v\ASET{v=!x_1}@\emptyset\\
G'_2 &\LOGICEQ & \allworlds \forall w.\APP{\pi_1(o)}{w}\ASET{!x_0=w}@x_0.
\end{eqnarray*}
Note $I_0$ is stateless except for $x_0$.
In $G_1$, notice the empty effect set means $!x_1$ does not change from
the pre to the post-condition. 
We now present the inference. 
Below we set 
$\mathtt{cell}'\DEFEQ\LET{y}{\REFPROG{0}}{\mathtt{cell}}$ 
and $i,k$ fresh. 
\begin{DERIVATION}
          \LINE
          {1.}
         {
 	   \ASET{\truth}
 	   \  \mathtt{cell}:_o\
 	   \ASET{
 	     I_0 
 	     \AND
 	     G_1
 	     \AND
 	     G_2
 	     \AND
 	     E'
 	   }
          }
  	 {}
          \LINE
 	 {2.}
 	 {
 	   \ASET{\truth}
 	   \  \mathtt{cell}':_o\
 	   \ASET{
 	     !x_0=!x_1
 	     \AND
	     G_1
 	     \AND
 	     G_2
 	     \AND
 	     E'
	   }
 	 }
  	 {(LetRef)}
          \LINE
 	 {3.}
 	 {
 	   \ASET{\truth}
 	   \,  \LET{x_1\!}{\!z}{\!\mathtt{cell}'}:_o
 	   \ASET{
 	     \new x_1.(I_0
	     \!\AND\!
 	     G_1
 	     \!\AND\!
 	     G_2
 	      )
 	     \AND
 	     E'
 	   }
 	   \!\!\!
 	   \!\!\!
 	   \!\!\!
 	 }
  	 {(LetRef)}
          \LINE
 	 {4.}
 	 {
 	   \ASET{\truth}
 	   \  \LET{x_1}{z}{\mathtt{cell}'}:_o\
 	   \ASET{
 	     G'_1
 	     \AND
 	     G'_2
 	     \AND
 	     E'
 	   }
 	   \!\!\!\!\!
 	 }
  	 {(AIH${}_\mathsf{c}$, ConsEval)}
         \LINE
 	 {5.}
 	 {
 	   \ASET{\truth}
 	   \  \LET{x_{0,1}}{z}{\mathtt{cell}'}:_o\
 	   \ASET{
 	     \new x.
	     (\
 	     x\noreach k \!\AND\!
 	     Cell(o,x) \AND !x = z 
 	     \ )
 	   }
 	   \!\!\!\!\!
 	 }
  	 {(LetRef)}
          \LASTLINE
 	 {6.}
 	 {
 	   \ASET{\truth}
 	   \  \mathtt{cellGen}:_u\
 	   \ASET{
 	     CellGen(u)	     
 	   }\ .
 	 }
  	 {(Abs)}
\end{DERIVATION}

\end{document}